\documentclass[runningheads]{llncs}

\usepackage[T1]{fontenc}
\usepackage{graphicx}
\usepackage{hyperref}
\usepackage{color}

\spnewtheorem{observation}{Observation}{\bfseries}{\itshape}
\spnewtheorem{rrule}{Reduction Rule}{\bfseries}{\itshape}
\usepackage{amsmath} %because amsmath needs to be before cleveref ...
\usepackage{cleveref}
\Crefname{observation}{Observation}{Observations}
\Crefname{rrule}{Reduction Rule}{Reduction Rules}
\Crefname{figure}{Fig.}{Figs.}
\Crefname{claim}{Claim}{Claims}

\usepackage{amsthm}
\usepackage{amsmath}
\usepackage{amssymb}
\usepackage[shortlabels]{enumitem}

\usepackage{todonotes}
\usepackage{xspace}
\usepackage{tikz-cd}
\usepackage{booktabs}
\usepackage{comment}
\usepackage{algorithm}
\usepackage[noend]{algpseudocode}
\usepackage{bm}
\usepackage{tikz}
\usepackage{thmtools} %restatable
\usetikzlibrary{decorations.pathmorphing}
\usetikzlibrary{calc}
\usetikzlibrary{positioning, arrows.meta, shapes.geometric, calc, backgrounds}
\usepackage{fullpage}

\newcommand{\NP}{{\sf NP}\xspace}
\newcommand{\coNP}{{\sf coNP}\xspace}

\newcommand{\FPT}{{\sf FPT}\xspace}

\newcommand{\poly}{\ensuremath{\operatorname{poly}}}
\newcommand{\prefix}{\pi}
\newcommand{\suffix}{\sigma}
\newcommand{\total}{\tau}
\newcommand{\combined}{\gamma}

\newcommand{\vcn}{\ensuremath{\operatorname{vcn}}\xspace} % Vertex cover number
\newcommand{\tw}{\ensuremath{\operatorname{tw}}\xspace} % Tree-width
\newcommand{\fen}{\ensuremath{\operatorname{fen}}\xspace} % feedback edge set
\usepackage{etoolbox}
\usepackage{placeins}
\ifdefined\ShortVersion
\newcommand{\sv}[1]{#1}
\newcommand{\lv}[1]{}
\newcommand{\appendixText}{}
\newcommand{\toappendix}[1]{\gappto{\appendixText}{{#1}}}
\else
\newcommand{\sv}[1]{}
\newcommand{\lv}[1]{#1}
\newcommand{\appendixText}{}
\newcommand{\toappendix}[1]{#1}
\fi
\newcommand{\appmark}{$\star$}

\newcommand{\appsection}[2]{\section{#1}\label{#2}\toappendix{\sv{\section{Omitted material from \Cref{#2}: #1}}}}

\newcommand{\appsubsection}[2]{\subsection{#1}\label{#2}\toappendix{\sv{\subsection{Omitted material from \Cref{#2}: #1}}}}

\newenvironment{apprestatable}[2]{\sv{\restatable[\appmark]{#1}{#2}}\lv{\restatable{#1}{#2}}}{\endrestatable}

\begin{document}

\newcommand{\MST}{\ensuremath{\operatorname{MST}}\xspace} %min spanning tree
\newcommand{\STEINER}{\ensuremath{\operatorname{STEINER}}} 
\newcommand{\MSTG}{{\sc MST-Game}\xspace} %
\newcommand{\MSTGs}{{\sc MSTG}\xspace}
\newcommand{\MSTGcore}{{\sc MSTG Core Non-Membership}\xspace}
\newcommand{\kMSTGcore}{{\sc $k$-MSTG Core Non-Membership}\xspace}
\newcommand{\supp}{\ensuremath{\operatorname{supp}}} %support
\newcommand{\opt}{\ensuremath{\operatorname{opt}}}
\newcommand{\DP}{\ensuremath{\operatorname{DP}}}

\Crefname{rrule}{Reduction Rule}{Reduction Rules}
\Crefname{claim}{Claim}{Claims}

\title{Core stability recognition for minimum-cost spanning tree games: Parameterized perspective}
%
%\titlerunning{Abbreviated paper title}
% If the paper title is too long for the running head, you can set
% an abbreviated paper title here
%

\author{Michal {Dvořák}\orcidID{0000-0002-5048-773X} \and
Ioannis {Kakatelis}\orcidID{0009-0003-5335-2309} \and
Dušan {Knop}\orcidID{0000-0003-2588-5709}}

\authorrunning{Dvořák et al.}
% First names are abbreviated in the running head.
% If there are more than two authors, 'et al.' is used.
%
\institute{Czech Technical University, Prague, Czech Republic \email{\{michal.dvorak,ioannis.kakatelis,dusan.knop\}@fit.cvut.cz}}

\maketitle              % typeset the header of the contribution

\begin{abstract}
 Minimum-cost spanning tree game (MSTG) is a cooperative game played on an undirected edge-weighted graph $(G,w)$ representing the network, where each vertex corresponds to a player and each edge has an associated cost~$w$. A distinguished vertex $s \in V(G)$ represents the supply or source. For any coalition of players $S$, the characteristic cost function $c(S)$ is defined as the minimum cost of a spanning tree with respect to $w$, connecting exactly the vertices in $S \cup \{s\}$. In this paper we study the computational complexity of deciding core membership for MSTG. In general, deciding whether a given allocation is in the core is \textsf{coNP}-hard~(Faigle et al.~\cite{faigle97coNPhard}). We study the core recognition problem under the name \MSTGcore. We extend the hardness to graphs which are very close to being planar. On the positive side, we present several algorithmic results within the framework of parameterized complexity. We show that \MSTGcore is fixed-parameter tractable when parameterized by the support size of the allocation. Turning into structural parameters of graphs, we show that the problem admits an \FPT algorithm parameterized by treewidth and signed neighborhood diversity. Last but not least, we investigate kernelization. While in general graphs, under standard complexity-theoretical assumptions, \MSTGcore does not admit a polynomial kernel parameterized by the vertex cover number, we design a cubic kernel in planar graphs. Furthermore, in general graphs, we obtain quadratic kernel for signed neighborhood diversity and linear kernel for the parameter feedback edge number.

\end{abstract}

\clearpage
\section{Introduction}
The minimum-cost spanning tree (\MST) problem is a classical and fundamental problem in network optimization, with broad applications in communication, transportation, and computer networks~\cite{papadimitriou98,schrijver03}. Given a connected edge-weighted graph, where vertices may represent cities or users and edges represent possible connections with associated costs, the objective is to select a subset of edges that connects all vertices while minimizing the total cost. Efficient algorithms to solve the \MST problem are well known~\cite{chung96,jarnik30,kruskal56,prim57}.

The minimum-cost spanning tree game (MSTG) extends the MST into a cooperative game-theoretic framework~\cite{Moulin95,rothe2024economics}. In MSTG, we are given an undirected edge-weighted graph $(G,w)$ representing the network, where each vertex corresponds to a player and each edge has an associated cost~$w$. A distinguished vertex $s \in V(G)$ represents the supply or source. For any coalition of players $S$, the characteristic cost function $c(S)$ is defined as the minimum cost of a spanning tree with respect to $w$, connecting exactly the vertices in $S \cup \{s\}$. Thus, MSTG models situations in which agents must jointly establish connections to access a common resource or service, sharing the total cost of the resulting network.

Minimum-cost spanning tree games arise naturally in cost allocation problems, where a joint infrastructure or service -- such as a communication network or energy distribution system -- must be built collectively. MST games were introduced by Claus and Kleitman~\cite{ClausK1973} and later extensively studied~\cite{GranotH81,bird1976cost,megiddo1978computational}. MSTGs are therefore a key example of cost-sharing cooperative games, where players form coalitions to minimize joint expenses, and the central question concerns how the total cost should be allocated among participants in a way that satisfies fairness and stability criteria.

In this paper, we focus on the \emph{core} of MSTGs, which represents the set of imputations for which no coalition has an incentive to deviate. An \emph{imputation} is an efficient and individually rational allocation of the total cost among players. Formally, for a cost game with player set $V$ and characteristic cost function $c\colon 2^V\to \mathbb{R}$, an allocation vector $\bm{x} \in \mathbb{R}^V$ (or, equivalently, an allocation function $x\colon V \to \mathbb{R}$; we use the functional notation in our paper) is an
\emph{imputation} if it satisfies:
\begin{description}
    \item[Efficiency] $\sum_{v \in V} x(v) = c(V)$.
    \item[Individual Rationality] $x(v) \le c(\{v\})$ for all $v \in V$.
\end{description} 
Among all imputations, particular interest is given to those that are stable against deviations by coalitions of players. The set of such stable imputations is known as the \emph{core}, defined as
\[
    \operatorname{core}(c)
    =
    \left\{ x\colon V \to \mathbb{R} \mid \sum_{v \in V} x(v) = c(V), \sum_{v \in S} x(v) \le c(S) \ \forall S \subseteq V \right\}
    \,.
\]

An allocation in the core ensures that no subset of players can form a coalition to construct its own spanning tree at a lower cost, thus guaranteeing that the grand coalition of all $V$ is stable.

A fundamental and appealing property of MSTGs is that their core is always nonempty. Furthermore, an allocation in the core can be found efficiently in linear time, given the MST of the underlying weighted graph. The player $v$ is allocated the cost of the first edge on the unique path from $v$ to the supply vertex in the minimum spanning tree. This rule is due to Bird~\cite{bird1976cost}. However, while finding \emph{some} allocation in the core is computationally easy, understanding the full structural boundaries of the core is much more complex. Our aim is to provide a combinatorial characterization of the core, which facilitates the recognition of whether a given allocation is stable.

In particular, we are interested in the following algorithmic question: Given an allocation $x\colon V \to \mathbb{R}$, does it belong to the core? In general, this question is \coNP-hard in the strong sense~\cite{faigle97coNPhard}. In other words, the problem remains hard even if the input weights $w$ and values $x$ are encoded in unary. Surprisingly, while the classical complexity of this problem is well-established, its parameterized complexity has not yet been systematically explored. To bridge this gap, we investigate the core recognition problem through the lens of parameterized complexity, aiming to identify structural parameters of the underlying graph $G$ that render the problem tractable.

%\subsubsection{Related Work.}
\smallskip
\noindent\textbf{Related Work.}
While finding some allocation in the core is computationally easy~\cite{bird1976cost} the core often contains a vast number of possible allocations. Consequently, a rich line of literature focuses on identifying specific core allocations that satisfy additional fairness and economic properties. For example, Dutta and Ray~\cite{DuttaR89} formalized the concept of egalitarian equivalence under participation constraints to find allocations that maximize fairness within the core. Similarly, Dutta and Kar~\cite{dutta2004cost} and Bergantiños and Vidal-Puga~\cite{bergantinos2007fair} have proposed and axiomatized allocation rules for MSTGs that reside firmly in the core while satisfying crucial properties like cost monotonicity. 

However, enumerating, verifying, or optimizing these highly desirable fair-and-stable allocations requires a deep structural understanding of the core's boundaries, which brings us to its computational complexity. Testing whether a given allocation belongs to the core is \coNP-hard in the strong sense~\cite{faigle97coNPhard}. This intractability result is a well-documented hurdle in combinatorial optimization games; Deng et al.~\cite{deng1999algorithmic} demonstrated that the complexity of the core is often deeply tied to the integrality gaps of the underlying graph problems. Because exact core membership testing is often intractable, recent literature has explored relaxed stability concepts, such as algorithms for computing a $4$-approximate core for MSTGs~\cite{kumabe2024lipschitz}. Furthermore, structural restrictions on the underlying network have been shown to drastically reduce complexity. Bachrach et al.~\cite{bachrach2014sharing} demonstrated that while testing for relaxed core imputations in general connectivity games is \coNP-complete, the problem becomes polynomial-time solvable when the network topology is restricted to a tree. The core stability has also been explored in another class of games called information graph games~\cite{nunez2022stable,kuipers1993core}. An information graph game is a special case of \MSTGs where the underlying graph has binary weights on the edges.

%\subsubsection{Our contribution.}
\smallskip
\noindent\textbf{Our contribution.}
Motivated by structural restrictions, our work systematically explores the parameterized complexity of the \MSTGs core recognition problem which we call \MSTGcore (see \Cref{sec:preliminaries} for proper definitions).

We begin by strengthening the known \NP-hardness results in several directions~(\Cref{sec:np_hardness}). First, we show that the problem remains \NP-hard even when the input graph is close to being planar. We even obtain \NP-hardness when the underlying graph is planar; nevertheless, we relax the efficiency condition in this case. Second, we observe that the problem remains \NP-hard even for graphs with maximum degree $3$. Lastly, we establish \NP-hardness for the restriction with binary edge weights and allocation. This is directly motivated by the graph information game, demonstrating that identifying core allocations remains computationally difficult even under severe weight restrictions.

Turning to positive algorithmic results, we establish a structural connection between \MSTGcore and the classic \textsc{Steiner Tree} problem~(\Cref{sec:support_fpt}). By exploiting this connection, we show that \MSTGcore admits single-exponential \FPT algorithm parameterized by the support size of the allocation, i.e., by the size of the set $\supp(x)=\{v\in V \mid x(v)\neq 0\}$.

The same connection to \textsc{Steiner Tree} further enables us to design an \FPT algorithm parameterized by treewidth~(\Cref{sec:treewidth}) based on dynamic programming over a nice tree decomposition of the underlying graph.

In~\Cref{sec:snd}, we show that \MSTGcore is solvable in $O(2^{\operatorname{snd}}\log (\operatorname{snd})(n+m))$ time, where $\operatorname{snd}$ is the signed neighborhood diversity of the underlying weighted graph. The algorithm is obtained by reducing the problem to multiple computations of a minimum spanning tree in the type graph and then greedily computing the best possible solution. Beyond this \FPT algorithm, we investigate kernelization with respect to this parameter. We design a data reduction rule that groups vertices with the same behavior inside each part of a $w$-uniform partition (see \Cref{sec:preliminaries} for definition) of the vertex set. The reduction rule alone produces a kernel with $O(\operatorname{snd}^2)$ vertices. Applying the classical framework of Frank and Tardos~(see~\Cref{thm:FT}), we compress this into a truly polynomial kernel.

Considering the parameter feedback edge number of the underlying graph, we show that \MSTGcore admits a kernel with a linear number of vertices and edges~(\Cref{sec:fen}). To achieve this, we introduce two novel data reduction rules. The first rule rigorously prunes leaf vertices; notably, applying this rule in isolation yields a remarkably simple linear-time algorithm for trees. The second rule compresses weights of the edges and values of vertices along long paths consisting of degree-$2$ vertices. As a result of independent structural interest, we formally show that a potential blocking coalition admits only a constant number of `meaningful' configurations along such path -- intuitively, a deviating coalition must either consume the entire path or strictly restrict itself to specific segment near the endpoints. 

Our kernelization analysis continues in \Cref{sec:vertex_cover} with the vertex cover number. First, we show that in general graphs, \MSTGcore does not admit a polynomial kernel under the classical assumption $\NP\not\subseteq \coNP/_{\operatorname{poly}}$ by providing a polynomial parameter transformation from the \textsc{Red-Blue Dominating Set} problem. Second, we introduce a new reduction rule that groups vertices of degree $2$ with common endpoints into constantly many structurally distinct `types', which, together with leaf pruning, yields a kernel with $O(\vcn^3)$ vertices in planar graphs.

%\subsubsection{Paper Organization}
\smallskip
\noindent\textbf{Paper Organization.}
The rest of the paper is organized as follows.
In~\Cref{sec:preliminaries} we introduce the main notation and terminology used throughout the paper. In~\Cref{sec:np_hardness,sec:support_fpt,sec:treewidth,sec:fen,sec:vertex_cover,sec:snd} we present the main results of our paper, as described above. Finally, in~\Cref{sec:conclusion} we conclude the paper and discuss future research directions.

\sv{
Statements where proofs or details are omitted due to space constraints are marked with~\appmark. The omitted material is available in the appendix.
}

\appsection{Preliminaries}{sec:preliminaries}
For a nonnegative integer $k$ we denote $[k]:=\{1,2,\ldots,k\}$. \lv{Let $X$ be a set. A \emph{partition} of $X$ is a set $\{X_1,\ldots, X_\ell\}$ of pairwise disjoint subsets of $X$ with $\bigcup_{i = 1}^{\ell} X_i=X$. The set of all partitions of $X$ is denoted by $\Pi(X)$.}

%\subsubsection{Graph Theory.}
\smallskip
\noindent\textbf{Graph Theory.}
We assume the reader is familiar with basic concepts of graph theory. For more definitions and comprehensive overview, refer to the monograph of Diestel~\cite{diestel}. All considered graphs are finite, simple and undirected. We denote by $V(G)$, resp. $E(G)$ the sets of vertices and edges of graph $G$ and we use $n=|V(G)|$ and $m=|E(G)|$. $N_G(v)$ is the \emph{neighborhood} of a vertex $v\in V(G)$, and $N_G(S)=\bigcup_{v\in S}N_G(v)$ for $S\subseteq V(G)$. The subgraph of $G$ induced by the vertex set $X\subseteq V(G)$ is denoted $G[X]$. Let $G$ be a graph and $e$ any edge. We let $G+e$ be the graph resulting by the addition of $e$ to $G$. Note that if the endpoints of the edge  $e$ do not exist in~$G$, we add them as new vertices. Formally, $G+e:=(V(G)\cup e,E(G)\cup \{e\})$. \lv{We extend the notation $G+e$ naturally to sets of edges: $G+\{e_1,e_2,\ldots e_k\}:=G+e_1+e_2+\cdots+e_k$.} A graph $G$ is \emph{$1$-planar} if it admits a drawing in the plane such that each edge is crossed at most once. A graph $G$ is a \emph{$k$-apex graph} if there is a set of at most $k$ vertices in $G$ whose removal renders $G$ planar.

\toappendix{
Let $G$ and $G'$ be two graphs with possibly overlapping vertex sets. For $S\subseteq V(G)$ we let $S|_{G'}$ be the restriction of $S$ to $G'$, i.e., $S|_{G'}:=S\cap V(G')$. Similarly, if $H$ is a subgraph of $G$, then $H|_{G'}$ is the restriction of $H$ to the vertices and edges of $G'$. Formally $H|_{G'}:=(V(H)\cap V(G'),E(H)\cap E(G'))$. Note that $H|_{G'}$ is also a (possibly empty) graph.
}%toappendix

An \emph{edge-weighted graph} is a pair $(G,w)$, where $G=(V,E)$ is a graph and $w\colon E\to \mathbb{Q}^+_0$. Let $S\subseteq V$. The set of minimum spanning trees of $G[S]$ is denoted $\MST_G(S)$ and $w(\MST_G(S))$ is the weight of a minimum spanning tree of $G[S]$ (or $+\infty$ if $G[S]$ is disconnected). A \emph{Steiner tree for $S$} is a subtree $T_S$ of $G$ with $V(T_S)\supseteq S$ ($S$ is the set of \emph{terminals}). We denote $\STEINER_G(S)$ the set of Steiner trees for $S$ of minimum weight and $w(\STEINER_G(S))$ is the minimum weight of a Steiner tree for $S$ (or $\infty$ if no Steiner tree for $S$ exists).

For any function $x\colon V\to \mathbb{Q}^+_0$ and $S\subseteq V$ we let $x(S):=\sum_{v\in S}x(v)$ and for a any subgraph $H$ of $G$ we let $x(H):=x(V(H))$. Similarly, we let $w(H):=w(E(H)):=\sum_{e\in E(H)}w(e)$.

%\subsubsection*{Problem Definition}
\smallskip
\noindent\textbf{Problem Definition.}
The central problem of our paper is \MSTGcore defined as follows. The input is a $4$-tuple $(G,s,x,w)$, where $(G,w)$ is an edge-weighted graph, $x\colon V\to \mathbb{Q}^+_0$ is an allocation\footnote{While the general notion of an allocation in cooperative game theory is defined over $\mathbb{R}$, we restrict our focus to the (non-negative) rationals to ensure the problem is amenable to algorithmic and computational complexity analysis.}, and $s\in V(G)$ is the supply vertex. The task is to decide whether there is a set $S\subseteq V(G)$, $S\ni s$ such that $x(S)>w(\MST_G(S))$.

%\begin{center}
%\begin{tabular}{|r|l|}
%\hline
%                    & \MSTGcore \\\hline
%     \textsc{Input:} & Edge-weighted graph $(G,w)$, allocation $x\colon V\to \mathbb{Q}^+_0$, $s\in V(G)$ \\\hline
%     \textsc{Question:}& Is there a set $S\ni s$ such that $x(S)> w(\MST_G(S))$?\\\hline
%\end{tabular}
%\end{center}
As a special case, if $x$ is an imputation, i.e., $x(G)=w(\MST_G(V))$ and $x(v)\leq w(\MST_G(\{v,s\}))$, we are deciding whether $x$ is in the core of the MSTG given by $(G,w,s)$. More precisely, $(G,s,x,w)$ is a yes-instance of \MSTGcore if and only if the imputation~$x$ is \emph{not} in the core. We note that our algorithms do not assume the efficiency condition to hold. On the other hand, all but one of our hardness reductions produce instances where the efficiency condition holds. To establish an equivalent structural characterization of the problem, observe that if $T$ is any spanning tree of $G[S]$ and $x(S)>w(T)$, it follows that $x(S)>w(T)\geq w(\MST_{G}(S))$. Consequently, the task in \MSTGcore is equivalent to finding a subtree $T$ of $G$ with $V(T)\ni s$ such that $x(T)>w(T)$.

\smallskip\noindent\textit{Weights.}
All considered weights and allocations are assumed to be nonnegative. While our hardness reductions produce integral weights and allocations, our algorithms operate over rational numbers. Some reduction rules may introduce non-integral fractions, but the largest denominator in any reduced form is at most $2$. Note that any such instance can be scaled to integers by multiplying everything by the least common multiple of the denominators. Consequently, to ensure generality across both algorithmic and hardness results, we assume that the weights and allocations take values from $\mathbb{Q}^+_0$.

%\subsubsection{Parameterized Complexity.} 
\smallskip\noindent\textbf{Parameterized Complexity.}
We assume the reader to be familiar with the concepts of parameterized complexity. For more definitions and comprehensive overview, refer to the monograph of Cygan et al.~\cite{cygan2015parameterized}. A parameterized problem is (in the class) \FPT if it is solvable in $f(k)\cdot n^{O(1)}$ time for some computable function~$f$ where $k$ is the parameter and $n$ is the input size.

\smallskip\noindent\textit{Kernelization.}
Let $f$ be a computable function. A \emph{kernel of size $f$} for a parameterized problem $L$ is an algorithm that given an instance $(\mathcal{I},k)$ of $L$, runs in $\poly(|\mathcal{I}|+k)$ time, and returns an equivalent instance $(\mathcal{I}',k')$ of $L$ such that $|\mathcal{I}'| + k' \leq f(k)$.

\smallskip\noindent\textit{Parameterized Reductions.}
Let $L$ and $L'$ be two parameterized problems. A \emph{parameterized reduction from $L$ to $L'$} is an algorithm that given an instance $(\mathcal{I},k)$ of $L$, runs in $f(k)\cdot |\mathcal{I}|^{O(1)}$ time, and returns an equivalent instance $(\mathcal{I}',k')$ of $L'$ such that $k'\leq g(k)$ for some computable functions $f,g$. If both $f$ and $g$ are polynomial functions, the reduction is referred to as \emph{polynomial parameter transformation} (ppt). We utilize ppt to show a non-existence of polynomial compressions~(see \cite[Definition 15.8.]{cygan2015parameterized}) (hence polynomial kernels) under the assumption $\NP\not\subseteq \textsf{coNP}/_{\operatorname{poly}}$. We remark that $\NP\subseteq \textsf{coNP}/_{\operatorname{poly}}$ implies that the polynomial hierarchy collapses to the third level~\cite{Yap1983}.

In our paper we consider the following structural parameters of (edge-weighted) graphs: \emph{vertex cover number}, \emph{feeback edge number}, \emph{treewidth}, and \emph{signed neighborhood diversity}. The definitions of these follow. Let $G=(V,E)$ be a graph. A set $S\subseteq V$ is a \emph{vertex cover} of $G$ if $G[V\setminus S]$ is edgeless. The \emph{vertex cover number} of~$G$ is the size of the smallest vertex cover of $G$ and is denoted as $\vcn(G)$. A set of edges $F\subseteq E$ is \emph{feedback edge set} of $G$ if $G\setminus F$ is acyclic. The smallest size of a feedback edge set of $G$ is the \emph{feedback edge number} of $G$, denoted by $\operatorname{fen}(G)$. Note that $\operatorname{fen}(G)=m-n+c$ where $c$ is the number of connected components of $G$.

%\subsubsection{Treewidth}
%\noindent\textbf{Treewidth}.
%\paragraph{Treewidth.}
\smallskip
\noindent\textit{Treewidth}.
A \emph{tree decomposition} of a graph $G$ is a pair $(\mathcal{T},\beta)$, where $\mathcal{T}$ is a tree and $\beta\colon V(\mathcal{T})\to 2^V$ is a mapping that assigns to  each node $t\in V(\mathcal{T})$ a set $\beta(t)\subseteq V$ of vertices and this set is referred to as \emph{bag} of~$t$. Moreover, the following three conditions hold:
\begin{enumerate}[label=(T\arabic*)]
    \item $\bigcup_{t \in V(\mathcal{T})} \beta(t) = V(G)$.\label{item:tw1}
    \item For every edge $\{u,v\} \in E(G)$, there exists a node $t \in V(\mathcal{T})$, such that $\{u,v\}\subseteq \beta(t)$.\label{item:tw2}
    \item For every vertex $v \in V(G)$, the subgraph $\mathcal{T}[\{t \in V(\mathcal{T})\mid v\in \beta_{t}\}]$ is connected.\label{item:tw3}
\end{enumerate}
The \emph{width} of a tree decomposition $(\mathcal{T},\beta)$ is $\max_{t \in V(\mathcal{T})}|\beta(t)| - 1$ and the \emph{treewidth} of $G$, denoted by $\operatorname{tw}(G)$, is the minimum width of a tree decomposition of $G$.

%It is known that computing a tree decomposition of minimum width is in $\FPT$ parameterized by treewidth ~\cite{bodlaender1993linear, kloks1994treewidth} and also efficient algorithms exist for achieving near-optimal-tree-decompositions~\cite{korhonen2023improved}. 

\toappendix{
\sv{\subsubsection{Extended preliminaries for treewidth}}
Towards applying dynamic programming on tree decompositions, we assume that the underlying tree is rooted at some node $r\in V(\mathcal{T})$. This creates the usual ancestor-descendant relationship between the nodes of $\mathcal{T}$. In order to avoid ambiguity, we refer to the vertices of the tree $\mathcal{T}$ in a tree decomposition as \emph{nodes}. A Tree decomposition $(\mathcal{T},\beta)$ rooted at $r\in V(\mathcal{T})$ is called \emph{nice} if $\beta(r)=\emptyset$ and every node $t\in V(\mathcal{T})$ is of one of the following five types:

\begin{description}
    \item[Leaf node] if $t$ is a leaf of $\mathcal{T}$ and $\beta(t)=\emptyset$.
    \item[Introduce vertex node] if $t$ has a unique child $t'$ in $\mathcal{T}$ and there exists $v \notin \beta(t')$ such that $\beta(t) = \beta(t') \cup \{v\}$. We say that $v$ is \emph{introduced} at $t$.
    \item[Introduce edge node] if $t$ has a unique child $t'$ in $\mathcal{T}$, $\beta(t)=\beta(t')$ and $t$ is labelled with an edge $\{u,v\}\in E(G)$. We say that $\{u,v\}$ is \emph{introduced} at $t$.
    \item[Forget node] if $t$ has a unique child $t'$ in $\mathcal{T}$ and there exists $v \in \beta(t')$ such that $\beta_{t} = \beta_{t'} \setminus \{v\}$. We say that $v$ is \emph{forgotten} at $t$.
    \item[Join node] if $t$ has exactly two children $t_1$ and $t_2$ and $\beta(t) = \beta(t_1) = \beta(t_2)$.
\end{description}

We require that every edge of $E(G)$ is introduced exactly once in the whole decomposition. As a consequence of \ref{item:tw3}, for every vertex $v$ there is a unique highest node $t(v)$ for which $v\in \beta(t(v))$. Moreover, if we didn't consider introduce edge nodes, the parent of $t(v)$ is a node where $v$ is forgotten. We assume that the introduce edge nodes, introducing an edge with endpoint $u$ in our nice tree decomposition are inserted precisely between the node $t(u)$ and its parent that forgets $u$.

The associated subgraph $G_t$ of $G$ at node $t$ is the graph consisting of all the vertices and edges introduced in the subtree $\mathcal{T}_t$, which is the subtree of $\mathcal{T}$ consisting of all the descendants of the node $t$ (including $t$). Based on the discussion above, we can assume that if $t$ is a join node, then $G_t[\beta(t)]$ is an independent set.
}

%\subsubsection{Signed Neighborhood Diversity} 
\smallskip\noindent\textit{Signed Neighborhood Diversity.}
Let $G=(V,E)$ be a graph. A partition $\{V_1,V_2,\ldots V_d\}$ of $V(G)$ is \emph{uniform} if for any $i \in [d]$ and $u,v \in V_{i}$ it holds that $N_G(u) \setminus \{v\} = N_G(v) \setminus \{u\}$. In other words, for any two sets $V_i,V_j$ (possibly $i=j$) either $\forall v_i\in V_i,v_j\in V_j: \{v_i,v_j\}\in E(G)$ or $\forall v_i\in V_i,v_j\in V_j: \{v_i,v_j\}\notin E(G)$. The \emph{neighborhood diversity} $\operatorname{nd}(G)$ of a graph $G$ is the smallest integer $d$ such there is a uniform partition of $V(G)$ into $d$ sets.

Let $(G,w)$ be a weighted graph, a \emph{$w$-uniform partition} is a uniform partition of $V(G)$ satisfying moreover that for any two edges $e_1=\{u_1,v_1\},e_2=\{u_2,v_2\}$ if $u_1\in V_i,u_2\in V_i$ and $v_1\in V_j$ and $v_2\in V_j$, then $w(e_1)=w(e_2)$. The \emph{signed neighborhood diversity} $\operatorname{snd}(G)$ of $(G,w)$ is the smallest $d$ such that there is $w$-uniform partition of $V(G)$ into $d$ sets. Note that the parts of $w$-uniform partition of size $\operatorname{snd}$ coincide with the equivalence classes given by the equivalence $u\sim v\Leftrightarrow N_G(u)\setminus \{v\}=N_G(v)\setminus \{u\} \wedge \mbox{$\forall e\ni u$} \mbox{$\forall f\ni v$}: w(e)=w(f)$. Consequently, a $w$-uniform partition of size $\operatorname{snd}(G)$ can be computed in polynomial time for any weighted graph $(G,w)$.

%\subsubsection{Reduction of weights.}
\smallskip
\noindent\textbf{Reduction of weights.}
In our kernelization procedures, we make use of the framework introduced by Frank and Tardos~\cite{FrankT87} to reduce the size of the weights. This technique has been successfuly used in many kernelization procedures to reduce weights in the input instance for linear objective functions when the dimension is bounded by the parameter~\cite{BlazejCKSSV22_ft,Bentert19_ft,BevernFT19_ft,GurskiRR19_ft,EtscheidKMR15_ft,GoebbelsGRY17_ft,ChaplickFGK019_ft,GurskiRR19_ft}.

\begin{theorem}[Frank and Tardos~\cite{FrankT87}]\label{thm:FT}
    There is a polynomial-time algorithm which, given a rational vector $(z_1,z_2,\ldots,z_d)\in \mathbb{Q}^d$ and a positive integer $M$, finds a vector $(\widetilde{z}_1,\widetilde{z}_2,\ldots,\widetilde{z}_d)\in \mathbb{Z}^d$ such that for every $(c_1,\ldots,c_d)\in[-M,M]^d$ we have $\sum_{i=1}^dc_iz_i\leq 0$ if and only if \mbox{$\sum_{i=1}^dc_i\widetilde{z}_i\leq 0$}. Moreover, $\max_i\widetilde{z}_i\leq 2^{O(d^3)}M^{O(d^2)}$.
\end{theorem}

\begin{apprestatable}{lemma}{lemalgreducingweights}\label{lem:alg_reducing_weights}
    There is a polynomial time algorithm, which, given an instance $(G,s,x,w)$ of \MSTGcore with $d_1$ vertices and $d_2$ edges, produces $x'$ and $w'$ such that the instances $(G,s,x,w)$ and $(G,s,x',w')$ are equivalent and $(G,s,x',w')$ is of total bit-size $O((d_1+d_2)^4)$.
\end{apprestatable}
\toappendix{
\sv{\lemalgreducingweights*}
\begin{proof}
    Let $v_1,v_2,\ldots,v_{d_1}$ and $e_1,e_2,\ldots,e_{d_2}$ be the enumerations of vertices and edges of $G$. Let $d=d_1+d_2$ We consider the $d$-dimensional vector
    \[ (x(v_1),x(v_2),\ldots,x(v_{d_1}),w(e_1),w(e_2),\ldots,w(e_{d_2}))
    \] 
    and we run the algorithm from \Cref{thm:FT} for $M=1$. We thus obtain a vector \[(\widetilde{x(v_1)},\widetilde{x(v_2)},\ldots,\widetilde{x(v_{d_1})},\widetilde{w(e_1)},\widetilde{w(e_2)},\ldots,\widetilde{w(e_{d_2})}).\] 
    We let $x'(v_i)=\widetilde{x(v_i)}$, and $w'(e_i)=\widetilde{w(e_i)}$. It remains to argue that the instances $(G,s,x,w)$ and $(G,s,x',w')$ are equivalent. We show that for any subtree $T$ of $G$ we have $x(T)>w(T)$ if and only if $x'(T)>w'(T)$. In fact, we show that $x(T)-w(T)\leq 0$ if and only if $x'(T)-w'(T)\leq 0$. Note that this condition holds for any subgraph $T$ of $G$, we do not use the fact that $T$ is a tree. To see this, consider the $d$-dimensional coefficient vector $(c_1,c_2,\ldots,c_{d_1},c_{d_1+1},\ldots,c_{d_2})\in \{-1,0,1\}^d$ given as follows. For $i\in[d_1]$ we let $c_i=1$ if $v_i\in V(T)$ and $c_i=0$ otherwise. For $i\in\{d_1+1,d_1+2,\ldots,d_2-1,d_2\}$ we let $c_i=-1$ if $e_{i-d_1}\in E(T)$ and $c_i=0$ otherwise. Observe that 
    \[
        x(T)-w(T)=\sum_{i=1}^{d_1}c_ix(v_i)+\sum_{i=d_1+1}^{d_2}c_iw(e_{i-d_1})
    \]
    and similarly
    \[
        x'(T)-w'(T)=\sum_{i=1}^{d_1}c_ix'(v_i)+\sum_{i=d_1+1}^{d_2}c_iw'(e_{i-d_1}).
    \]
    By the guarantees from the algorithm from \Cref{thm:FT} $x(T)-w(T)\leq 0$ if and only if $x'(T)-w'(T)\leq 0$ and this proves the claim. Note that each new weight and value is of size at most $2^{O(d^3)}$, thus can be described using $O(d^3)$ bits. There are $d$ such numbers to be described, thus the numbers in the instance can be described by $O(d^4)$ bits. The rest of the instance (vertex adjacencies etc.) can be clearly described using at most $O(d^2)$ bits.
\end{proof}
}%toappendix

%\subsubsection{Exponential Time Hypothesis.} 

%The Exponential Time Hypothesis ensures, in a few words the non-existence of an algorithm for $3-\sat$ running in time $2^{o(\hat{n})}$, where $\hat{n}$ is the number of variables in the input formula $\varphi$ (\cite{impagliazzo2001complexity}). In fact a $2^{o(\hat{n}+\hat{m})}$ algorithm is impossible, by the application of the Sparsification Lemma (\cite{impagliazzo2001problems}), where $\hat{m}$ is the number of the clauses in the input formula $\varphi$.

\appsection{Improved {\sf NP}-hardness results}{sec:np_hardness}
It is known that \MSTGcore is strongly \NP-hard, that is, it is \NP-hard even if the weights $w$ and the allocation $x$ on the input are integral and encoded in unary. This is a consequence of the previously known reduction from \textsc{X3C} by Faigle et al.~\cite{faigle97coNPhard}. In this section, we strengthen the \NP-hardness in several directions.

Our primary objective is to determine the hardness of \MSTGcore on planar graphs. To this end, we approach the planar case through a sequence of three progressively tighter reductions. First, we observe that \NP-hardness trivially extends to $1$-planar graphs via a straightforward subdivision rule~(\Cref{obs:nph_oneplanar}). Next, we push the boundary closer to planarity by demonstrating that the problem remains \NP-hard on $2$-apex graphs -- that is, graphs that can be made planar by the removal of at most two vertices~(\Cref{thm:nph_twoapex}). Finally, we provide a reduction that yields a planar graph; however, this requires relaxing the efficiency condition that is satisfied in the previous constructions~(\Cref{thm:nph_planar}).

In addition to these topological constraints, we show that the problem  remains \NP-hard even under severe local and numerical constraints. Specifically, we show that \MSTGcore remains \NP-hard even when the maximum degree of the underlying graph is $3$. Next, we show that the problem remains \NP-hard even when $w$ and $x$ attain values from $\{0,1\}$ or when the underlying graph is unweighted ($w(e)=1$ for every $e$), and $x(v)\in \{0,2\}$~(\Cref{thm:mstg_hardness_01}). We note that our result is tight, as any stricter restriction yields a trivially solvable case. Crucially, these degree and weight restrictions can be applied directly to the constructions for planar and $1$-planar graphs, hence these hardness results hold even under these degree and weight constraints. In the case of $2$-apex graphs, only the weights can be reduced, however, the planar graph after the removal of the two vertices has constant degree.

We begin by defining four auxiliary operations.

\begin{definition}\label{defn:operations}
    Let $(G,s,x,w)$ be an instance of \MSTGcore. In each of the following operations, we formally create a new instance $(G',s',x',w')$.
    \begin{description}
        \item[$(a,b,c)$-expansion] of a vertex $v\in V(G)$ is an operation of adding a leaf $\ell$ to $v$, subtracting $c$ from $x(v)$ and setting $x(\ell)=b$ and $w(\{\ell,v\})=a$.
        \item[Binary tree expansion] of a vertex $v\in V(G)$ is an operation of replacing $v$ by a binary tree with $\deg(v)$ leaves and root $r$. The original edges incident to $v$ are now connected to the leaves. The root inherits the value $x(v)$ and all other newly created vertices $u$ and edges $e$ have $x(u)=0,w(e)=0$.
        \item[Contraction] of an edge $e=\{u,v\}$ with $w(e)=0$ is the operation of contracting the edge $e$ into a new vertex $v_e$, and setting $x'(v_e)=x(v)+x(u)$. If any parallel edges are created by the operation, then keep the one with lowest value $w(e)$. If one of the vertices was $s$, $v_e$ will be the new supply vertex $s'$.
        \item[$(w_1,w_2,\ldots,w_t)$-subdivision] of an edge $e\in E(G)$ is an operation of subdividing the edge $e$ exactly $t-1$ times, creating a path of length $t$ with $t-1$ new internal vertices. The weights $w'(e)$ of the $t$ newly created edges are set to $w_1,w_2,\ldots,w_t$ and the $x$-values of the newly created vertices are $0$.
    \end{description}
\end{definition}

\begin{apprestatable}{lemma}{lemoperationscorrectness}\label{lem:operations_correctness}
    Let $b,c,w_1,w_2,\ldots,w_t$ be nonnegative (rational) numbers. Vertex $(b+c,b,c)$-expansions, binary tree expansions, edge contractions, and Edge $(w_1,w_2,\ldots,w_t)$-subdivisions of edges $e$ with $\sum_{i=1}^tw_i=w(e)$ preserve yes- and no-instances of \MSTGcore and the efficiency condition.
\end{apprestatable}
\toappendix{
\sv{\lemoperationscorrectness*}
\begin{proof}
    Let $(G,s,x,w)$ be the original instance of \MSTGcore and let $(G',s',x',w')$ be the new instance created by an application of one of the operations. In each of the following case we show that $(G,s,x,w)$ is a yes-instnace if and only if $(G',s',x',w')$ is a yes-instance.

        \smallskip\textbf{Vertex $(b+c,b,c)$-expansion.} $\Rightarrow$: Any solution $T$ of $(G,s,x,w)$ containing $v$ is transformed to a solution $T':=T+\{\ell,v\}$. Notice that $x'(T')=x(T)$ and $w'(T')=w(T)$.

        $\Leftarrow$: Any solution $T'$ of $(G',s',x',w')$ avoiding $\ell$ can be transformed to $T''=T+\{\ell,v\}$ with $x'(T'')-w'(T'')>x'(T')-w'(T')$, which can then be transformed to a solution $T=T''\setminus \{\ell\}$ satisfying $x(T)=x'(T'')$ and $w(T)=w'(T'')$.
        
        \smallskip\textbf{Binary tree expansion.} $\Rightarrow$: Let $T$ be a solution for $(G,s,x,w)$ and suppose that $v\in V(T)$ and let $e_1,e_2,\ldots,e_q$ be all the edges of $T$ incident to $v$. Replace $v$ by the root of the binary tree together with all the paths to leaves connected to the edges $e_1,e_2,\ldots,e_q$ in the binary tree to obtain a solution $T'$ in $(G,s,x,w)$.

        $\Leftarrow$: Consider a solution $T'$ in $(G',s',x',w')$. Notice that whenever one of the vertices of the binary tree is used, then we can always include the root vertex for free and obtain a solution $T''$ with $x'(T'')\geq x'(T')$ and $w'(T'')=w'(T')$ because all the edges have zero cost. Hence, we create a solution $T$ by including $v$ together with all the edges that were incident to the leaves of the binary tree in $T'$. Note that $x(T)=x'(T'')$ and $w(T)=w'(T'')$.
        
        \smallskip\textbf{Edge contraction.} $\Rightarrow$: Any solution $T$ of $(G,s,x,w)$ that uses exactly one of $u$, $v$ can be transformed to a solution $T'$ that uses both $u$ and $v$ by simply attaching the other vertex as a leaf via the edge $e$. Note that since $x$ is nonnegative, we have $x'(T')\geq x(T)$, and $w'(T')=w(T)$. Hence we can assume that $T$ contains both $u$ and $v$. It is not left to distinguish two cases.
        \begin{description}
            \item[Case 1] $T$ uses the edge $\{u,v\}$. Contract it in $T$, creating a new tree $T'$ satisfying $x'(T')=x(T)$ and $w'(T')=w(T)$.
            \item[Case 2] $T$ does not use the edge $\{u,v\}$. Since $T$ contains both $u$ and $v$, there is a (unique) $u$-$v$ path $P$. Add $e$ to $T$ and remove the cheapest edge on $P$. Since $w(e)=0$ and $w$ is nonnegative, this creates a new tree $T'$ with $w'(T')\leq w(T)$ and then apply the argument from Case 1.
        \end{description}

        $\Leftarrow$: If $T'$ is a solution for $(G',s',x',w')$, then either $T'$ doesn't touch $v_e$, hence it is also a solution for $(G,s,x,w)$, or it contains $v_e$. If it contains $v_e$, we let $T=(T'\setminus \{v_e\})+e$. This is a solution for $(G,s,x,w)$.
        
        \smallskip\textbf{Edge $(w_1,w_2,\ldots,w_t)$-subdivision.} $\Rightarrow$: Any solution $T$ for $(G,s,x,w)$ either avoid the edge $e$, then it is a solution for $(G',s',x',w')$ immediately, or it contains $e$. If it contains $e$, we replace $e$ in $T$ by the newly created path, creating a solution $T'$ for $(G',s',x',w')$. Note that since $\sum_{i=1}^tw_i=w(e)$, we have $w'(T')=w(T)$.

        $\Leftarrow$: Consider a solution $T'$ for $(G',s',x',w')$ minimizing $|V(T')|$. It is not hard to see that $T'$ either uses all the edges of the new path or nothing from the path. If it used only some part of the path, we could remove these edges and obtain a solution $T''$ with $w'(T'')\leq w'(T')$ and $x'(T'')=x'(T')$, thus contradicting the minimality of $T'$. Hence we can either replace the whole path by $e$ (if $T'$ uses the path) or leave $T'$ intact to obtain a solution $T$ for $(G,s,x,w)$.

    It remains to argue about the efficiency condition. Denote $V=V(G)$ and $V'=V(G')$. We show that whenever $x(G)=\MST_G(V)$, then $x'(G')=\MST_{G'}(V')$, i.e., the operations preserve the efficiency condition. We again perform case analysis based on the operation performed.
    
        \smallskip\textbf{Vertex $(b+c,b,c)$-expansion.} Notice that $x'(G')=x(G)+a-c$ and $w(\MST_{G'}(V')=w(\MST_G(V))+b$, notice that $a=b+c$, the claim follows.
        
        \smallskip\textbf{Binary tree expansion.} Here the claim is trivial as $x'(G')=x(G)$ and $w'(\MST_G(V'))=w(\MST_G(V))$.
        
        \smallskip\textbf{Edge contraction.} Note that $x'(G')=x(G)$ trivially. We argue that $w(\MST_{G'}(V'))=w(\MST_G(V))$. Recall that for any $v'\in N_G(u)\cap N_G(v)$ the parallel edges $\{v',u\},\{v',v\}$ which are created by the initial contraction of $e$ are resolved by keeping only the edge with the minimum weight. Notice that there is a minimum spanning tree $T$ of $G$ that contains $e$. By contracting $e$ we obtain a minimum spanning tree $T'$ of $G'$. If $T'$ contained a heavier edge to $v'$, we could replace it already in $T$ by the lighter one, which contradicts the assumption that $T$ is a minimum spanning tree of $G$. Hence $w(\MST_{G'}(V'))\leq w(\MST_{G}(V))$.

        To see the other inequality, note that any minimum spanning tree of $T'$ in $G'$ can be mapped back to $G$ by `decontracting' $v_e$ into the zero-weight edge $e$. This forms a valid spanning tree $T$ of $G$ with $w(T)=w'(T')$. Hence $w(\MST_{G}(V))\leq w(\MST_{G'}(V'))$.
        
        \smallskip\textbf{Edge $(w_1,w_2,\ldots,w_t)$-subdivision.} Here we again have $x'(G')=x(G)$ trivially. Note that if a minimum spanning tree uses the edge $e$, we can replace it in the spanning tree by the entire path. Since $\sum_{i=1}^tw_i=w(e)$, the weight is unchanged. Hence $\MST_{G'}(V')=\MST_{G}(V)$.

\end{proof}
}%toappendx

\subsection{{\sf NP}-hardness for $1$-planar and $2$-apex graphs}

\begin{apprestatable}{observation}{obsnphoneplanar}\label{obs:nph_oneplanar}
        \MSTGcore is \NP-hard even when the underlying graph is a $1$-planar graph and the efficiency condition is satisfied.    
\end{apprestatable}
\sv{
\begin{proof}[Proof sketch]
    Draw the graph arbitrarily in the plane and $(0,0,0,\ldots,w(e))$-subdivide each edge $e$ with more crossings enough so that each original crossing is accounted for one of the newly created edges. See \Cref{fig:subdivisions} in the appendix for an illustration.
\end{proof}
}
\toappendix{
\sv{\obsnphoneplanar*}
\begin{proof}
    We reduce from \MSTGcore. Consider an arbitrary drawing\footnote{Specifically, any drawing where each edge has polynomially many crossings, such as a straight-line drawing.} of the underlying graph~$G$ in the plane. For each edge $e$ with $q>1$ crossings, we apply a $(\underbrace{0,0,\ldots,0}_{q-1},w(e))$-subdivision (see \Cref{fig:subdivisions} for illustration). The resulting graph is $1$-planar. To see this, note that the original drawing of $G$ can be adjusted so that each newly created edge participates in at most one crossing.
    Correctness follows from \Cref{lem:operations_correctness}.
\end{proof}
\begin{figure}
    \centering
      \begin{tikzpicture}[
    vertex/.style={circle, draw=black, fill=blue!80!black, minimum size=8pt, inner sep=0pt},
    newvertex/.style={circle, draw=black, fill=orange!80!black, minimum size=8pt, inner sep=0pt},
    crossingvertex/.style={circle, draw=gray!60, fill=gray!30, minimum size=5pt, inner sep=0pt},
    edge/.style={draw=blue!60!cyan, line width=2.5pt},
    crossedge/.style={draw=gray!70, line width=1.2pt, dashed},
    labelstyle/.style={font=\small\sffamily}
]

    % ================= LEFT SIDE: Original Edge =================
    \begin{scope}[shift={(0,0)}]

        % Main Vertices
        \node[vertex,label = {right=4pt:\textbf{$u$}}] (u) at (0, 2.5) {};
        \node[vertex,label = {right=4pt:\textbf{$v$}}] (v) at (0, -2.5) {};
        
        % The edge e
        \draw[edge] (u) -- (v) node[left=5pt, text=blue!60!cyan, font=\bfseries, pos=0.4] {$w(e)$};
        
        % Crossing edges and their endpoints
        % Crossing 1
        \node[crossingvertex] (c1l) at (-1.5, 1.2) {};
        \node[crossingvertex] (c1r) at (1.5, 0.8) {};
        \draw[crossedge] (c1l) -- (c1r);
        
        % Crossing 2
        \node[crossingvertex] (c2l) at (-1.5, 0.2) {};
        \node[crossingvertex] (c2r) at (1.5, -0.2) {};
        \draw[crossedge] (c2l) -- (c2r);
        
        % Crossing 3
        \node[crossingvertex] (c3l) at (-1.5, -0.8) {};
        \node[crossingvertex] (c3r) at (1.5, -1.2) {};
        \draw[crossedge] (c3l) -- (c3r);
        
    \end{scope}

    % ================= MIDDLE: Transformation Arrow =================
    \draw[-{Stealth[scale=1.5]}, line width=1.5pt, gray!60] (2.2, 0) -- (5, 0) 
        node[midway, above=4pt, text=black, font=\small\bfseries\sffamily, align=center] {$(0, 0, w(e))$-subdivide $e$};

    % ================= RIGHT SIDE: Subdivided Edge =================
    \begin{scope}[shift={(200pt,0pt)}]

        % Vertices
        \node[vertex,label = {right=4pt:\textbf{$u$}}] (u2) at (0, 2.5) {};
        \node[newvertex, label={right=4pt:\textbf{$v_1$}}] (v1) at (0, 0.5) {};
        \node[newvertex, label={right=4pt:\textbf{$v_2$}}] (v2) at (0, -0.5) {};
        \node[vertex,label = {right=4pt:\textbf{$v$}}] (v2_end) at (0, -2.5) {};
        
        % Subdivided edges
        \draw[edge] (u2) -- (v1) node[pos=0.6, left=5pt, text=blue!60!cyan, font=\bfseries] {$0$};
        \draw[edge] (v1) -- (v2) node[pos=0.7, left=5pt, text=blue!60!cyan, font=\bfseries] {$0$};
        \draw[edge] (v2) -- (v2_end) node[pos=0.5, left=5pt, text=blue!60!cyan, font=\bfseries] {$w(e)$};
        
        % Crossing edges and their endpoints
        % Crossing 1
        \node[crossingvertex] (c1l2) at (-1.5, 1.2) {};
        \node[crossingvertex] (c1r2) at (1.5, 0.8) {};
        \draw[crossedge] (c1l2) -- (c1r2);
        
        % Crossing 2
        \node[crossingvertex] (c2l2) at (-1.5, 0.2) {};
        \node[crossingvertex] (c2r2) at (1.5, -0.2) {};
        \draw[crossedge] (c2l2) -- (c2r2);
        
        % Crossing 3
        \node[crossingvertex] (c3l2) at (-1.5, -0.8) {};
        \node[crossingvertex] (c3r2) at (1.5, -1.2) {};
        \draw[crossedge] (c3l2) -- (c3r2);
        
    \end{scope}

\end{tikzpicture}

    \caption{Illustration for the proof of \Cref{obs:nph_oneplanar}. In the original graph (on the left) the edge $e=\{u,v\}$ (in blue) had $q=3$ crossings. We performed $(0,0,w(e))$-subdivision on $e$, creating two new vertices $v_1,v_2$. The gray edges are all the edges crossing $e$ in the considered drawing. The new edge weights are $w'(\{u,v_1\})=w'(\{v_1,v_2\})=0$, and $w'(\{v_2,v\})=w(e)$.}
    \label{fig:subdivisions}
\end{figure}
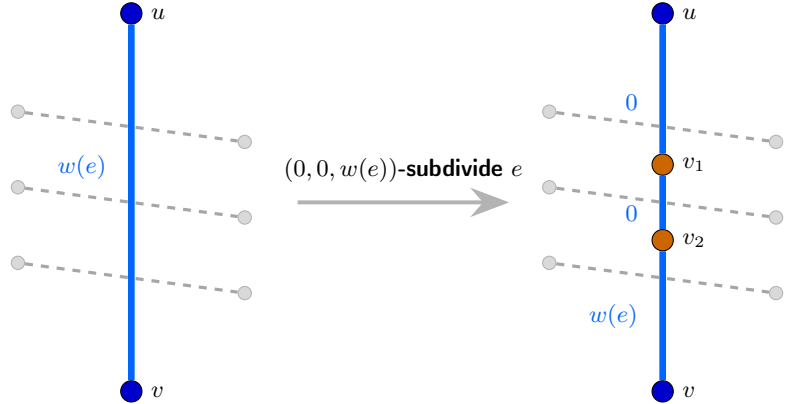

}%toappendix

\begin{apprestatable}{theorem}{thmnphtwoapex}\label{thm:nph_twoapex}
    \MSTGcore is \NP-hard even when the underlying graph is a $2$-apex graph and the efficiency condition is satisfied.
\end{apprestatable}
\toappendix{
\sv{\thmnphtwoapex*}
\begin{proof}
    Reduce from \textsc{Vertex Cover} restricted to $3$-regular planar graphs. Such variant is \NP-hard~\cite{GareyJ79}. Let $(G,k)$ be an instance of \textsc{Vertex Cover}. We create an instance $(G',s,x,w)$ of \MSTGcore as follows. The construction is inspired by the one presented in the paper of Faigle et al. from the \textsc{X3C} problem~\cite{faigle97coNPhard}. We refer the reader to~\Cref{fig:vc_reduction} for visualization of the construction. Let $\widehat{G}$ be the graph obtained from $G$ by subdividing each edge exactly once. Let $R\subseteq V(\widehat{G})$ be the set of original vertices of $G$ and let $B=V(\widehat{G})\setminus R$ be the vertices corresponding to the subdivisions. Notice that the \textsc{Vertex Cover} problem translates to finding a set $S\subseteq R,|S|\leq k$ such that $N_{\widehat{G}}(S)=B$. Note that $\widehat{G}$ is a bipartite planar graph of maximum degree $3$. To create $G'$, add a vertex $g$ (the \emph{guard}) connected to all of $R$, then, add the vertex $\nu$ (the \emph{steiner vertex}) and connect it to $R\cup \{g\}$. Finally, add the supply vertex $s$ and connect it to~$g$. We now specify the weights $w$ and the allocation $x$. All vertices $v$ outside $R\cup B$ satisfy $x(v)=0$. The vertices $v_r\in R$ satisfy $x(v_r)=k$, and $x(v_b)=2k+4$ for $v_b\in B$. The weights of the edges are as follows:
    \begin{align*}
        &w(\{s,g\})=|B|\cdot (k+2)-(k+1),\\
        \forall v\in R\cup \{\nu\}:&w(\{g,v\})=k+1,\\
        \forall v_r\in R:&w(\{\nu,v_r\})=k,\\
        \forall \{v_r,v_b\}\in E(\widehat{G}):&w(\{v_r,v_b\})=k+2.
    \end{align*}
    This finishes the description of the instance $(G',s,x,w)$.

\definecolor{setblue}{HTML}{BBDEFB}
\definecolor{setred}{HTML}{FFCDD2}
\definecolor{setgreen}{HTML}{C8E6C9}
\definecolor{setgray}{HTML}{E0E0E0}
\definecolor{darkgreen}{HTML}{2E7D32}
\definecolor{darkpurple}{HTML}{6A1B9A}
\begin{figure}
\centering
\begin{tikzpicture}[
    vertex/.style={circle, draw, minimum size=7mm, thick},
    font=\small,
]

    % --- Vertices ---
    \node[vertex, fill=setgreen] (s) at (0,3) {$s$};
    \node[vertex, fill=setgray] (g) [below=of s] {$g$};
    \node[vertex, fill=setgreen!40] (nu) [below=of g] {$\nu$};
    
    % Set R
    \node[vertex, fill=setred] (r2) [below=of nu] {$v_{r_{2}}$};
    \node[vertex, fill=setred] (r1) [left=of r2] {$v_{r_{1}}$};
    \node[vertex, fill=setred] (r3) [right=of r2] {$v_{r_{3}}$};
    
    % Set B
    \node[vertex, fill=setblue] (b2) [below=of r2] {$v_{b_{2}}$};
    \node[vertex, fill=setblue] (b1) [left=of b2] {$v_{b_{1}}$};
    \node[vertex, fill=setblue] (b3) [right=of b2] {$v_{b_{3}}$};

    % Subgraph \widehat{G}
    \path[fill=setgray!55, fill opacity=0.35, rounded corners=2mm]
        ($(r1.north west)+(-0.45,0.75)$)
        rectangle
        ($(b3.south east)+(0.45,-0.35)$);
    \draw[gray, thick, dashed, rounded corners=2mm]
        ($(r1.north west)+(-0.45,0.75)$)
        rectangle
        ($(b3.south east)+(0.45,-0.35)$);
    \node[anchor=north west, color=gray!70!black]
        at ($(r1.north west)+(-0.25,0.55)$) {$\widehat{G}$};

    % --- Left Column: Node Definitions (Aligned with Nodes) ---
    \node[left=6cm of s, anchor=west] {supply node: $x(s) = 0$};
    \node[left=6cm of g, anchor=west] {guard vertex: $x(g) = 0$};
    \node[left=6cm of nu, anchor=west] {steiner vertex: $x(\nu) = 0$};
    \node[left=6cm of r2, anchor=west] {set $R$: $x(v) = k$};
    \node[left=6cm of b2, anchor=west] {set $B$: $x(v) = 2k+4$};

    % --- Right Column: Edge Weights (Aligned with Edge Gaps) ---

    \node[](gnu) at($(g)!0.5!(nu)$) {};
    % Midpoint g to nu
    \node[right=2.43cm of gnu, anchor=west, color=darkpurple]
        {$g \leftrightarrow R \cup \{\nu\}$: $w = k + 1$};
        
    % Midpoint nu to R
    \node[](nuR) at ($(nu)!0.5!(r2)$){};
    \node[right=2.43cm of nuR, anchor=west, color=darkgreen]
        {$\nu \leftrightarrow R$: $w = k$};
        
    % Midpoint R to B
    \node[](RB) at ($(r2)!0.5!(b2)$){};
    \node[right=2.43cm of RB, anchor=west] 
        {$R \leftrightarrow B$: $w = k + 2$};

    %Midpoint s to g
    \node[](sg) at($(s)!0.5!(g)$){};
    \node[right=2.43cm of sg, anchor=west]{$s\leftrightarrow g$: $w = |B|(k+2) - (k+1)$};

    % --- Edges ---
    % s to g (ultra-thick black)
    \draw[line width=2.5pt] (s) -- (g);

    % g to nu (purple, weight k+1)
    \draw[thick, darkpurple] (g) -- (nu);

    % g to R (purple, weight k+1)
    \draw[thick, darkpurple] (g) to[bend right=30] (r1);
    \draw[thick, darkpurple] (g) to[bend right=30] (r2); 
    \draw[thick, darkpurple] (g) to[bend left=30] (r3);

    % nu to R (green, weight k)
    \foreach \i in {1,2,3} {
        \draw[thick, darkgreen] (nu) -- (r\i);
    }

    % B to R (black, weight k+2)
    \draw[thick] (b1) -- (r1);
    \draw[thick] (b1) -- (r2);
    
    \draw[thick] (b3) -- (r1);
    \draw[thick] (b2) -- (r3);
    
    \draw[thick] (b2) -- (r2);
    \draw[thick] (b3) -- (r3);

\end{tikzpicture}

    \caption{Construction of the \MSTGcore instance from a \textsc{Vertex Cover} instance from the proof of \Cref{thm:nph_twoapex}. Vertex $x$-values are listed on the left, while edge weights $w$ are indicated by color-coded labels on the right. The black edges between sets $R$ and $B$ correspond to the edges of the graph $\widehat{G}$.}
    \label{fig:vc_reduction}
\end{figure}
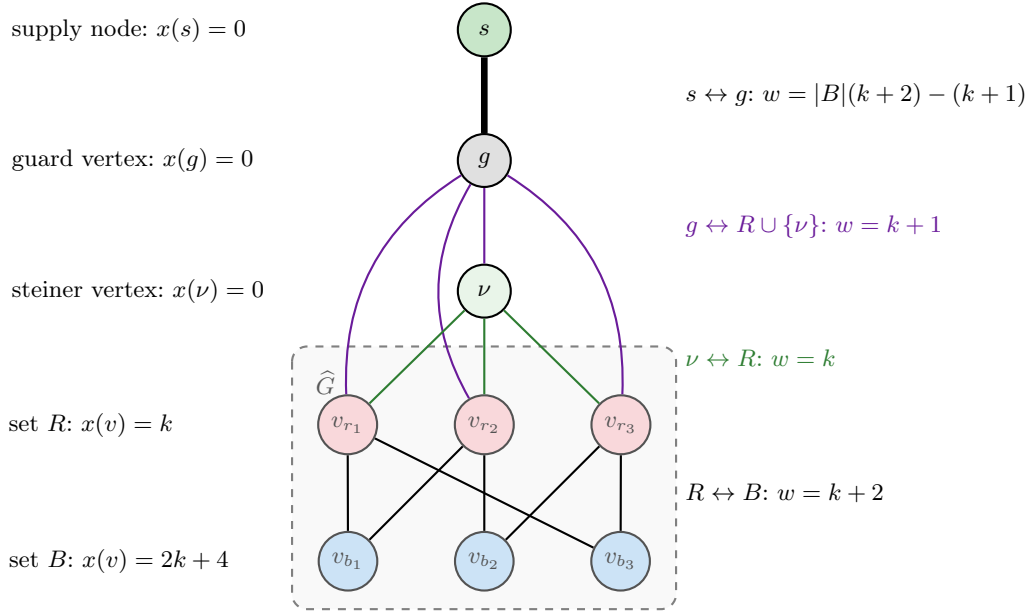
    
     \begin{claim}\label{claim:vc_spanning_tree}
       $x(G')=w(\MST(G'))$ 
    \end{claim}
    \begin{proof}
    Observe that $x(G')=|B|\cdot (2k+4)+|R|\cdot k$. The minimum spanning tree of $G'$ looks as follows. Consider running the Kruskal's algorithm for finding the minimum spanning tree on $G'$. First, all edges from $\nu$ to $R$ will be added. Then one of the edges with weight $k+1$ connecting the created component to the guard $g$. We can assume that the edge is $\{\nu,g\}$. Then, $|B|$ edges from $R$ to $B$ of weight $k+2$ are added and finally the edge $\{s,g\}$ is added. The total weight of the minimum spanning tree of $G'$ is thus:
    \[
    w(\MST(G'))=|R|\cdot k+(k+1)+|B|\cdot (k+2)+|B|\cdot(k+2)-(k+1)
    \]
    which is the same as $x(G')$.
    \end{proof}

    \begin{claim}\label{claim:vc_first_dir}
        If $(G,k)$ is a yes-instance of \textsc{Vertex Cover}, then $(G',s,x,w)$ is a yes-instance of \MSTGcore.
    \end{claim}
    \begin{proof}
    Suppose $(G,k)$ is yes-instance of \textsc{Vertex Cover}. As observed above a vertex cover $S\subseteq V(G)$ of $G$ of size $|S| \leq k$, is the same as a vertex set $S\subseteq R$ in $\widehat{G}$ such that $N_{\widehat{G}}(S)=B$. We construct a solution $T$ in $(G',s,x,w)$ as follows. Include the vertices $s$ and $g$ and connect $g$ to $R$ using the edges $\{g,v_r\},v_r\in S$. Finally, add all vertices of $B$ as leaves. Since $N_{\widehat{G}}(S)=B$, there is an available neighbor in $S$ for every $v_b\in B$. We have
    \begin{align*}
        x(T)&=|B|\cdot (2k+4)+|S|\cdot k,\\
        w(T)&=|B|\cdot (k+2) + |S|\cdot (k+1)+ |B|\cdot (k+2)-(k+1) .
    \end{align*}
    It follows that $x(T)-w(T)=-|S|+(k+1)\geq 1>0$, because $|S|\leq k$. Thus, $(G',s,x,w)$ is a yes-instance.
    \end{proof}

    \begin{claim}\label{claim:vc_second_dir}
        If $(G',s,x,w)$ is a yes-instance of \MSTGcore, then $(G,k)$ is a yes-instance of \textsc{Vertex Cover}.
    \end{claim}
    \begin{proof}
        Suppose that there is a solution $T$ for $(G',s,x,w)$, i.e., $x(T)-w(T)>0$ and $s\in V(T)$. Suppose moreover that $x(T)-w(T)$ is maximum possible. Consider the allocation $y\colon V(G')\to \mathbb{Z}$ given by the greedy algorithm computing the core element from the minimum spanning tree from the proof of \Cref{claim:vc_spanning_tree}: $y(v)$ is given by the weight of the first edge on the path from $v$ to $s$ in a minimum spanning tree, and $y(s)=0$ (see~\cite{bird1976cost}). The allocation $y$ is given by:
        \begin{align*}
            \forall v_b\in B:&y(v_b)=k+2, \\
            \forall v_r\in R:&y(v_r)=k, \\
            &y(\nu)=w(\{g,\nu\}) = k+1, \\
            &y(g)=w(\{s,g\})=|B|\cdot (k+2)-(k+1).
        \end{align*}

        Since $y$ is in the core, we have $y(T)\leq w(T)$. In particular, $x(T)>y(T)$, thus $x(T)-y(T)>0$. Let $V_B=V(T)\cap B$ and $V_R=V(T)\cap R$.

        \begin{claim}\label{claim:vc_guard_steiner}
            $g\in V(T)$ and $\nu \notin V(T)$.
        \end{claim}
        \begin{proof}
            Notice that $g$ is the only vertex connected to $s$, so $g\in V(T)$. For the sake of contradiction assume that $\nu \in V(T)$. Then $x(T)= |V_B|\cdot (2k+4)+|V_R|\cdot k$. On the other hand, $y(T)=|V_B|\cdot (k+2)+|V_R|\cdot k+(k+1)+|B|\cdot (k+2)-(k+1)$. Notice that $x(T)-y(T)=(|V_B|-|B|)(k+2)\leq 0$ because $|V_B|\leq |B|$. A contradiction, because $x(T)-y(T)>0$.
        \end{proof}

        \begin{claim}\label{claim:vc_leaves}
            All vertices in $V_B$ are leaves in $T$.
        \end{claim}
        \begin{proof}
            For the sake of contradiction, suppose that there is a vertex $v_b\in V_B$ with $\deg_T(v_B)\geq 2$. Let $v_r\in V_R$ be a neighbor of $v_b$ in $T$ such that the tree $T\setminus e$ where $e=\{v_r,v_b\}$ contains no path from $g$ to $v_r$. Note that such $v_r$ must exist as otherwise the paths together with the edges incident to $v_b$ would create a cycle in $T$, which is impossible, because $T$ is a tree. Since $g$ and $v_r$ lie in distinct connected components of $T\setminus e$, we can connect them via the edge $\{g,v_r\}$, creating a new tree $T'$. Notice that $w(T')=w(T)-(k+2)+(k+1)=w(T)-1$, but this contradicts the maximality of $x(T)-w(T)$.
        \end{proof}

        By \Cref{claim:vc_guard_steiner,claim:vc_leaves}, every vertex of $V_B$ is connected via exactly one edge to a vertex of $V_R$ and every vertex of $V_R$ is connected to $g$ via an edge of weight $k+1$. It follows that
        \begin{align*}
            x(T)&= |V_B|\cdot (2k+4)+|V_R|\cdot k \\
            w(T)&=|V_B|\cdot (k+2)+|V_R|\cdot(k+1)+|B|(k+2)-(k+1)
        \end{align*}
        The assumption $x(T)>w(T)$ is equivalent to $x(T)-w(T)>0$. Substituting the expressions for $x(T)$ and $y(T)$ into this inequality yields
        \begin{align}\label{eqn:vc}
            |V_B|\cdot (k+2)-|V_R|-|B|\cdot(k+2)+(k+1)>0 
        \end{align}

        \begin{claim}\label{claim:vc_vbb}
            $|V_B|=|B|$.
        \end{claim}
        \begin{proof}
        It is sufficient to prove both inequalities $|V_B| \leq |B|$ and $|V_B| \geq |B|$. For the first inequality, we note that $|V_B| \leq |B|$ holds because $V_B \subseteq B $. Now for proving $|V_B| \geq |B|$, for the sake of contradiction we can assume that $|V_B| < |B|$. Since both numbers are integers, this is equivalent to $|V_B| \leq |B| -1$. By plugging into (\ref{eqn:vc}) we get:
        \[
        (|B|-1)\cdot (k+2)-|V_R|-|B|\cdot(k+2)+(k+1)>0
        \]
        which entails $-1 > |V_R|$, which is impossible.
        \end{proof}

        \begin{claim}\label{claim:vc_vrk}
            $|V_R|\leq k$.
        \end{claim}
        \begin{proof}
        We note that by plugging $|V_B| = |B|$ into \Cref{eqn:vc} we obtain that $(k+1)>|V_R|$, i.e., $|V_R|\leq k$, as claimed.
        \end{proof}
    Since $T$ is connected and each $v_b\in V_B$ has a neighbor in $V_R$, it follows that $N_{\widehat{G}}(V_R)\supseteq V_B=B$, thus $V_R$ is a vertex cover of size at most $k$ for $G$. This finishes the proof of \Cref{claim:vc_second_dir}.
    \end{proof}
    
    The correctness of the construction follows from \Cref{claim:vc_first_dir,claim:vc_second_dir}. Note that $G'\setminus \{g,\nu\}$ is planar, hence $G'$ is a $2$-apex graph. Note that the reduction can be computed in polynomial time. This finishes the proof of \Cref{thm:nph_twoapex}.
    
\end{proof}
}%toappendix

\begin{apprestatable}{theorem}{thmnphplanar}\label{thm:nph_planar}
    \MSTGcore is \NP-hard even when the underlying graph is planar.
\end{apprestatable}
\toappendix{
\sv{\thmnphplanar*}
\begin{proof}

    \begin{figure}
    \centering
    \begin{tikzpicture}[scale=1.2,
          % Styles for original graph G and set R (Red)
          vertexR/.style={circle, draw=red!80!black, thick, fill=white, inner sep=2pt, minimum size=6mm},
          vertexRsol/.style={circle, draw=red!80!black, thick, fill=red!30, inner sep=2pt, minimum size=6mm},
          % Styles for subdivision set B (Blue)
          vertexBsol/.style={circle, draw=blue!80!black, thick, fill=blue!20, inner sep=1pt, minimum size=4mm},
          % Styles for guard and supply (Gray)
          vertexGsol/.style={circle, draw=black!80!black, thick, fill=black!10, inner sep=2pt, minimum size=6mm},
          % Edge styles
          arrow/.style={->, ultra thick, line width=2pt, >=Stealth},
          solEdgeG/.style={line width=2.5pt, red!80!black},
          solEdgeGp/.style={line width=2.5pt, orange!80!black},
          nsEdge/.style={thick, black!30},
          nsEdgeG/.style={thick, black!40}
    ]

% LEFT SIDE: Graph G
\node[vertexR] (v1) at (0,0) {$v_1$};
\node[vertexRsol] (v2) at (2,0) {$v_2$};
\node[vertexRsol] (v3) at (2,2) {$v_3$};
\node[vertexRsol] (v4) at (0,2) {$v_4$};
\node[vertexR] (v5) at (1,3.5) {$v_5$};

\draw[nsEdgeG] (v1) -- (v2); 
\draw[nsEdgeG] (v1) -- (v3);
\draw[solEdgeG] (v2) -- (v3); 
\draw[solEdgeG] (v3) -- (v4); 
\draw[nsEdgeG] (v4) -- (v1); 
\draw[nsEdgeG] (v4) -- (v5);
\draw[nsEdgeG] (v5) -- (v3);

\node[font=\large\bfseries] at (1, -1.0) {$G$};

% ARROW: Construction / Reduction
\draw[arrow] (3.0, 1.5) -- node[above=2pt, font=\small\bfseries] {} (5.0, 1.5);

% RIGHT SIDE: Graph G' (with subset \widehat{G})
\begin{scope}[xshift=6.0cm]
    % Draw background dashed box for \widehat{G}
    \begin{scope}
        \draw[draw=gray!50, dashed, thick, rounded corners, fill=gray!2] 
            (-0.5,-0.6) rectangle (2.5,4.0);
        \node[gray!80!black, font=\small\bfseries, anchor=north west] at (-0.4, 3.9) {$\widehat{G}$};
    \end{scope}

    % Original vertices R
    \node[vertexR] (v1a) at (0,0) {$v_1$};
    \node[vertexRsol] (v2a) at (2,0) {$v_2$};
    \node[vertexRsol] (v3a) at (2,2) {$v_3$};
    \node[vertexRsol] (v4a) at (0,2) {$v_4$};
    \node[vertexR] (v5a) at (1,3.5) {$v_5$}; 
    
    % Subdivision vertices B
    \node[vertexBsol] (v12a) at (1,0) {};   
    \node[vertexBsol] (v23a) at (2,1) {};   
    \node[vertexBsol] (v34a) at (1,2) {};  
    \node[vertexBsol] (v41a) at (0,1) {};  
    \node[vertexBsol] (v45a) at (0.5, 2.75) {};
    \node[vertexBsol] (v53a) at (1.5, 2.75) {};
    \node[vertexBsol] (v13a) at (1,1) {};

    % Guard and Supply
    \node[vertexGsol] (g) at (3.3, 1) {$g$};
    \node[vertexGsol] (s) at (4.8, 1) {$s$};

    % Edges in \widehat{G}
    % v1 - v2
    \draw[nsEdge] (v1a) -- (v12a);
    \draw[solEdgeGp] (v12a) -- (v2a);
    % v1 - v3
    \draw[nsEdge] (v1a) -- (v13a);
    \draw[solEdgeGp] (v13a) -- (v3a);
    % v2 - v3
    \draw[solEdgeGp] (v2a) -- (v23a);
    \draw[solEdgeGp] (v23a) -- (v3a);
    % v3 - v4
    \draw[solEdgeGp] (v3a) -- (v34a);
    \draw[solEdgeGp] (v34a) -- (v4a);
    % v4 - v1
    \draw[solEdgeGp] (v4a) -- (v41a);
    \draw[nsEdge] (v41a) -- (v1a);
    % v4 - v5
    \draw[solEdgeGp] (v4a) -- (v45a);
    \draw[nsEdge] (v45a) -- (v5a);
    % v5 - v3
    \draw[nsEdge] (v5a) -- (v53a);
    \draw[solEdgeGp] (v53a) -- (v3a);

    % Guard and Supply Edges
    \draw[solEdgeGp] (g) -- (v2a);
    \draw[nsEdge] (g) -- (v3a);
    \draw[solEdgeGp] (s) -- (g);

    \node[font=\large\bfseries] at (2.2, -1.0) {$G'$};

\end{scope}

\end{tikzpicture}
\caption{Visualization of the reduction from \textsc{Connected Vertex Cover} (on the left) to \MSTGcore (on the right) from \Cref{thm:nph_planar}. On the left is $G$ with a connected vertex cover $S=\{v_2,v_3,v_4\}$ together with spanning tree witnessing the connectivity of $G[S]$ (in red).
On the right is $G'$, which consists of the subdivision graph $\widehat{G}$ (enclosed in the dashed gray box) along with the guard $g$ and the supply $s$.
Vertices in $R$ (original) are in red, while vertices in $B$ (subdivisions) are blue. The solution tree $T$ of $(G', s, x, w)$ is highlighted in orange. For the sake of visual clarity and readability, the allocation $x$ and edge weights $w$ are omitted from the illustration.
}
\label{fig:reduction_cvc_visualization}
\end{figure}
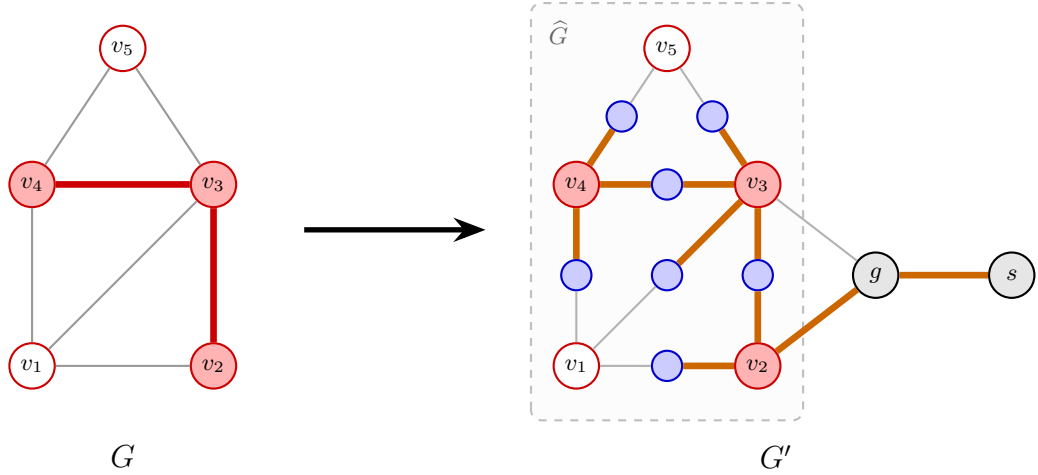

    Reduce from \textsc{Connected Vertex Cover} restricted to planar graphs of maximum degree $4$. Such variant is \NP-hard by the result of Garey and Johnson~\cite{garey1977rectilinear}. Let $(G,k)$ be an instance of \textsc{Connected Vertex Cover}. Assume that $G$ has no isolated vertices. We can remove them without changing the answer to $(G,k)$. If $G$ is disconnected and has no isolated vertices, then $(G,k)$ is trivially a no-instance, hence we return a trivial no-instance. Assume that $G$ is connected. We create an instance $(G',s,x,w)$ as follows. The construction is a modified construction from the proof of~\Cref{thm:nph_twoapex}. We refer the reader to~\Cref{fig:reduction_cvc_visualization} for visualization of the construction. Let $\widehat{G}$ be the graph obtained from $G$ by subdividing each edge exactly once. Let $R\subseteq V(\widehat{G})$ be the set of original vertices of $G$ and let $B=V(\widehat{G})\setminus R$ be the set of vertices corresponding to the subdivisions. Notice that the \textsc{Connected Vertex Cover} problem translates to finding a set $S\subseteq R,|S|\leq k$ such that $N_{\widehat{G}}(S)=B$ and moreover, $\widehat{G}[S\cup B]$ is connected subgraph of $\widehat{G}$. Let $e^*=\{u,v\}$ be any fixed edge of $G$. To create $G'$, add a vertex $g$ (the \emph{guard}) and connect it to exactly $u$ and $v$. Finally, add the supply vertex $s$ and connect it to $g$. We now specify the weights $w$ and the allocation $x$. All vertices outside $R\cup B$ satisfy $x(v)=0$. The vertices $v_r\in R$ satisfy $x(v_r)=k$, and $x(v_b)=2k+3$ for $v_b\in B$. The weights of the edges are as follows:
    \begin{align*}
        &w(\{s,g\})=|B|\cdot (k+2)-(k+1),\\
        &w(\{g,v\})=w(\{g,u\})=k+1,\\
        \forall \{v_r,v_b\}\in E(\widehat{G}):&w(\{v_r,v_b\})=k+1.
    \end{align*}
    This finishes the description of the instance $(G',s,x,w)$. 
    
    The difference from the previous construction in the proof of \Cref{thm:nph_twoapex} is in that we do not use the steiner vertex, and the only difference of $x$ and $w$ is that the weights $w$ on the edges between $R$ and $B$ and the values $x$ on $v_b\in B$ are decreased by $1$.

    \begin{claim}\label{claim:cvc_reduction_onedir}
        If $(G,k)$ is a yes-instance of \textsc{Connected Vertex Cover}, then $(G',s,x,w)$ is a yes-instance of \MSTGcore.
    \end{claim}
    \begin{proof}
        Suppose that $(G,k)$ is a yes-instance of \textsc{Connected Vertex Cover}. As observed above, a connected vertex cover $S$ of $G$ is the same as a set $S\subseteq R$ such that $N_{\widehat{G}}(S)=B$, and $\widehat{G}[S\cup B]$ is a connected subgraph of $\widehat{G}$. Let $\widehat{T}$ be an arbitrary spanning tree of $\widehat{G}[S\cup B]$. Create a solution $T$ of $(G',s,x,w)$ by first attaching $g$ as a leaf to $\widehat{T}$ and then attaching $s$ as a leaf via $g$. Note that since $S$ is a vertex cover of $G$, we have $u\in S\vee v \in S$, because $\{u,v\}\in E(G)$. Hence, we can attach $g$ to one of $u$ or $v$. We have
        \begin{align*}
        x(T)&=|B|\cdot (2k+3)+|S|\cdot k,\\
        w(T)&=(|B|+|S|)\cdot (k+1)+ |B|\cdot (k+2)-(k+1) .
        \end{align*}
        The formula for $w(T)$ comes from the fact that $T\setminus \{s\}$ is a subtree of $T$ with the vertex set $B\cup S\cup \{g\}$, hence it has $|B|+|S|+1$ vertices. All the edges of $T\setminus \{s\}$ have weight $k+1$, hence the total weight of this subtree is $(|B|+|S|)\cdot (k+1)$. Lastly, $|B|\cdot (k+2)-(k+1)$ comes from the edge $\{g,s\}$. It follows that $x(T)-w(T)=-|S|+(k+1)\geq 1>0$, because $|S|\leq k$. Thus, $(G',s,x,w)$ is a yes-instance.
    \end{proof}

    \begin{claim}\label{claim:cvc_reduction_twodir}
        If $(G',s,x,w)$ is a yes-instance of \MSTGcore, then $(G,k)$ is a yes-instance of \textsc{Connected Vertex Cover}.
    \end{claim}
    \begin{proof}
        Let $(G',s,x,w)$ be a yes-instance and let $T$ be a solution for $(G',s,x,w)$. Let $V_R=V(T)\cap R$ and $V_B=V(T)\cap B$. Note that $T\setminus \{s\}$ is a tree on $|V_B|+|V_R|+1$ vertices, and all edges of $T\setminus \{s\}$ have weight $k+1$. Hence we have
        \begin{align*}
            x(T)&=|V_B|\cdot (2k+3)+|V_R|\cdot k \\
            w(T)&=(|V_B|+|V_R|)\cdot (k+1)+|B|\cdot (k+2)-(k+1)
        \end{align*}
        The assumption $x(T)>w(T)$ is equivalent to $x(T)-w(T)>0$. Substituting the expressions for $x(T)$ and $y(T)$ into this inequality yields
        \begin{equation}\label{eqn:cvc}
            |V_B|\cdot (k+2)-|V_R|-|B|\cdot (k+2)+(k+1)>0
        \end{equation}
        Note that this is the same equation as \Cref{eqn:vc}. We thus have $|V_B|=|B|$ and $|V_R|\leq k$ by \Cref{claim:vc_vbb,claim:vc_vrk}. Since $T$ is connected and each $v_b\in V_B$ has a neighbor in $V_R$, it follows that $N_{\widehat{G}}(V_R)\supseteq V_B=B$, thus $V_R$ is vertex cover of size at most $k$ for $G$. It remains to argue about connectedness. We want to show that $\widehat{G}[V_R\cup B]$ is connected. Consider the tree $T$ without the vertices $g,s$, i.e. $T'=T[V_R\cup V_B]=T[V_R\cup B]$. If $\deg_T(g)=2$, then $T'$ is connected and thus witnessing the connectedness of $\widehat{G}[V_R\cup B]$, i.e., the connectedness of $G[V_R]$.

        If $\deg_T(g)=3$, then $T$ contains both $e_1=\{u,g\}$ and $e_2=\{v,g\}$. Let $v_{e^*}\in V(\widehat{G})$ be the vertex created by the subdivision of the edge $e^*=\{u,v\}\in E(G)$. Since $v_{e^*}\in V(T)$, one of the edges $e_3=\{u,v_{e^*}\},e_4=\{v,v_{e^*}\}$ has to be in $T$, because these are the only edges of $\widehat{G}$ incident to $v_{e^*}$. Note that if both $e_3,e_4$ were in $T$, then $e_1,e_2,e_3,e_4$ would create a $4$-cycle in $T$, which is impossible. Hence we can switch $e_2$ with the missing edge from $v_{e^*}$, creating a new tree $T''$ with $w(T'')=w(T')$. Now we have $\deg_{T''}(g)=2$ and the previous argument applies. This finishes the proof of \Cref{claim:cvc_reduction_twodir}.
    \end{proof}

    The correctness of the construction follows from \Cref{claim:cvc_reduction_onedir,claim:cvc_reduction_twodir}. Note that $\widehat{G}$ is planar as it is a subdivision of the planar graph $G$. Note that the vertices $g$ and $s$ including their incident edges can be drawn inside a face incident to the edge $e^*$ without introducing any new crossings. Hence, $G'$ is planar. Finally, the reduction can be computed in polynomial time. This finishes the proof of \Cref{thm:nph_planar}.

\end{proof}
}%toappendix

\subsection{\NP-hardness for maximum degree $3$ and binary weights}

\begin{observation}
    \MSTGcore remains \NP-hard even when the underlying graph has maximum degree $3$.
\end{observation}
\begin{proof}
    Reduce from \MSTGcore and perform binary tree expansion of every vertex $v$ with $\deg (v) > 3$. Correctness follows from~\Cref{lem:operations_correctness}.
\end{proof}

\begin{apprestatable}{theorem}{thmmstghardnesszeroone}\label{thm:mstg_hardness_01}
    \MSTGcore remains \NP-hard even in the following restricted cases.
    \begin{enumerate}
        \item $w(e)\in \{0,1\}$ for all edges $e$ and $x(v)\in \{0,1\}$ for all vertices $v$, or
        \item $w(e)=1$ for all edges $e$ and $x(v)\in\{0,2\}$ for all vertices $v$.
    \end{enumerate}
\end{apprestatable}
\sv{
\begin{proof}[Proof sketch]
    Reduce from \MSTGcore. For the first case use $(1,0,1)$-expansions and $(1,1,\ldots,1)$-subdivision. For the second case, use $(2,1,1)$-expansions and contractions.
\end{proof}
}%sv
\toappendix{
\sv{\thmmstghardnesszeroone*}
\begin{proof}
    We reduce from \MSTGcore which is strongly \NP-hard by the result of Faigle et al.~\cite{faigle97coNPhard}. For the first case, for any vertex $v$ with $x(v)>1$ perform $(1,0,1)$-expansion $x(v)-1$ times and for every edge $e$ with $w(e)>1$ perform a $(\underbrace{1,1,\ldots,1}_{w(e)})$-subdivision of $e$.

    For the second case, further contract every zero-weight edge. Now, for every vertex $v$ with $x(v)\notin \{0,2\}$ instead of $(1,0,1)$-expansion, perform $(2,1,1)$-expansion $x(v)-2$ times. 
    
    The correctness of both constructions follows from \Cref{lem:operations_correctness}.
\end{proof}
}%toappendix

Observe that the restrictions on the weights and allocation in \Cref{thm:mstg_hardness_01} are tight. If $x$ is restricted to attain only single value, then the problem becomes trivial. Similarly, if $\forall e\in E: w(e)=0$, then the instance is a yes-instance if and only if $x(v)\neq 0$ for some $v$ reachable from $s$. Even in the special case with $x(s)=0$ and $x(v)=1$ for $v\neq s$, the instance is a yes-instance if and only if there is an edge $e$ with $w(e)=0$ reachable from $s$. The last interesting case is when $w(e)=1$ for every edge and $x(v)\in \{0,1\}$ (i.e., the graph is unweighted). In this case, the instance is a yes-instance if and only if $x(s)=1$. We prove this formally in \Cref{obs:mstg_01_easy} \sv{in the appendix}.

\toappendix{
\begin{observation}\label{obs:mstg_01_easy}
    An instance of \MSTGcore where $x(s)=0$ and $x(v)=1$ for $v\neq s$, and $w(e)\in \{0,1\}$ is a yes-instance if and only if an edge $e$ with $w(e)=0$ is reachable from $s$.

    An instance of \MSTGcore where $w(e)=1$ for all edges $e$ and $x(v)\in\{0,1\}$ for all vertices $v$ is a yes-instance if and only if $x(s)=1$.

\end{observation}
\begin{proof}
    For the first part, suppose that $(G,s,x,w)$ is a yes-instance with a witnessing solution $T$. By the conditions on $x$ we have $x(T)=|T|-1$. If every edge of $T$ had weight $1$, then $w(T)=|T|-1$, contradicting the fact that $x(T)>w(T)$, hence in particular $T$ contains an edge with weight $0$, in particular, this edge is reachable from $s$. On the other hand, if there is an edge $e$ reachable from $s$, then the solution consists of the path from $s$ and including the edge $e$. Such solution $T$ satisfies $x(T)=|T|-1$ and $w(T)\leq |T|-2$ as there are exactly $|T|-1$ edges in $T$ and at least $1$ has weight $0$ (the edge $e$), hence at most $|T|-2$ of them have weight $1$.

    For the second part, suppose $(G,s,x,w)$ is a yes-instance with a witnessing solution $T$. Since $w(e)=1$ for every $e$, we have that $w(T)=|V(T)|-1$ and $x(T)\leq |V(T)|$, because $x(v)\leq 1$ for every vertex. On the other hand $x(T)>w(T)=|V(T)|-1$ because $T$ is a solution, thus $x(T)=|V(T)|$, thus, since $s\in V(T)$, we have $x(s)=1$. On the other hand, clearly if $x(s)=1$, then the tree $T$ consisting of a single vertex $s$ is a solution.
\end{proof}
}%toappendix

\appsection{Parameterization by the support of the allocation}{sec:support_fpt}
In this section, we provide an \FPT algorithm for \MSTGcore parameterized by the size of the support of the input allocation, i.e., the size of the set $\operatorname{supp}(x)=\{v\in V\mid x(v)\neq 0\}$. The idea is as follows. The `interesting' vertices are in the set $X:=\operatorname{supp}(x)\cup \{s\}$. We guess which vertices from $X$ are in the solution and then connect them using vertices outside $X$ in the cheapest possible way (in terms of weight of the edges). The cheapest way to connect vertices in $S\subseteq X$ is given by a Steiner tree for $S$ in the subgraph $G_S:=G\setminus (X\setminus S)$. We prove formally the correctness of this argument in \Cref{lem:fpt_by_supp_correctness}.

\begin{apprestatable}{lemma}{lemfptbysuppcorrectness}\label{lem:fpt_by_supp_correctness}
    Let $(G,s,x,w)$ be an instance of \MSTGcore. Let $X=\supp(x)\cup\{s\}$. The following are equivalent:
    \begin{enumerate}
        \item $(G,s,x,w)$ is yes-instance of \MSTGcore.
        \item $\exists S\subseteq X,S\ni s$ such that $x(T_S)>w(T_S)$ for some $T_S\in \STEINER_{G_S}(S)$, where $G_S=G\setminus (X\setminus S)$.
    \end{enumerate}
\end{apprestatable}

\toappendix{
\sv{\lemfptbysuppcorrectness*}
\begin{proof}
    $\Leftarrow$: Suppose there is a Steiner tree $T_S\in \STEINER_{G_S}(S)$ with $x(T_S)>w(T_S)$. Since $s\in S$, this is a solution for $(G,s,x,w)$.
    
    $\Rightarrow$: Suppose $(G,s,x,w)$ is a yes-instance of \MSTGcore and let $T$ be a solution. Let $S=V(T)\cap X$. Note that since $s\in V(T)$, we necessarily have $s\in S$. We show that there is a Steiner tree $T_S\in \STEINER_{G_S}(S)$ and $x(T_S)>w(T_S)$. Note that $T$ is a tree with $V(T)\supseteq S$ and $T$ is a subgraph of $G_S$, hence $\STEINER_{G_S}(S)\neq \emptyset$ and $w(\STEINER_{G_S}(S))\leq w(T)$. Let $T_S\in \STEINER_{G_S}(S)$ be arbitrary. It remains to prove the inequality $x(T_S)>w(T_S)$. Note that $x(T_S)=x(T)$ because $T$ and $T_S$ agree on the set $X$ containing the support of $x$. In other words $V(T)\cap X=V(T_S)\cap X$ and $\supp(x)\subseteq X$. We have
    \[
    x(T_S)=x(T)>w(T)\geq w(\STEINER_{G_S}(S))=w(T_S)
    \]
    and this finishes the proof.
\end{proof}
}%toappendix
To this end, we make use of the following result which is due to Fuchs et al.~\cite{FuchsKMRRW07_steiner_better}.
\begin{theorem}[\cite{FuchsKMRRW07_steiner_better}]\label{thm:fuchs_steiner}
    For any $\delta'>0$ it is possible, given a weighted graph $(G,w)$ and a set of terminals $S\subseteq V(G)$, to compute $T_S\in \STEINER_{G}(S)$ or report that $\STEINER_{G}(S)$ is empty in $(2+\delta')^{|S|}n^{O(1)}$ time.
\end{theorem}

\begin{apprestatable}{theorem}{thmMSTGcoreFPTsupport}\label{thm:fpt_by_s}
    For arbitrary $\delta > 0$, \MSTGcore is solvable in \mbox{$(3+\delta)^pn^{O(1)}$} time where $p=|\supp(x)|$.
\end{apprestatable}

\toappendix{
\sv{\thmMSTGcoreFPTsupport*}
\begin{proof}
    Consider the following algorithm for \MSTGcore. Let $(G,s,x,w)$ be the input instance. Let $X=\supp(x)\cup \{s\}$. Begin by enumerating all subsets $S\subseteq X$ containing $s$. Note that $|X|\leq p+1$, thus there are at most $2^p$ such subsets and they can be enumerated in $O(2^p)$ time. For each of them, use the algorithm from \Cref{thm:fuchs_steiner}  which for given $\delta'=\delta$ computes a Steiner tree $T_S\in \STEINER_{G_S}(S)$ of minimum weight in $(2+\delta)^{|S|}n^{O(1)}$ time (if one exists). If $x(T_S)>w(T_S)$, output \texttt{YES}. If for no $S\subseteq X$ a Steiner tree $T_S$ of minimum weight with $x(T_S)>w(T_S)$ is found, output \texttt{NO}. The correctness of the algorithm follows from \Cref{lem:fpt_by_supp_correctness} and the total running time is 
    \[
    \sum_{S\subseteq X}(2+\delta)^{|S|}n^{O(1)}=\sum_{i=1}^{|X|}\binom{|X|}{i}\cdot (2+\delta)^i \cdot n^{O(1)}\leq(3+\delta)^{|X|}\cdot n^{O(1)}\leq (3+\delta)^p\cdot n^{O(1)}
    \]
    where we used the binomial theorem in the second inequality. In the last inequality, we use $|X|\leq p+1$, and the one $3+\delta$ factor is hidden in the $n^{O(1)}$ term.
\end{proof}
}%toappendix
\appsection{Parameterization by treewidth}{sec:treewidth}
In this section we present an \FPT algorithm for \MSTGcore parameterized by the treewidth of the input graph $G$. We employ dynamic programming over a nice tree decomposition of $G$. Using \emph{nice} tree decomposition allows for a smoother description and analysis of the transitions between the dynamic programming states as the graph modifications are restricted to modify exactly one vertex or edge at a time.
\sv{Due to space constraints, we defer the formal definition of a nice tree decomposition along with related terminology and notation to the appendix, focusing here on a sketch of the dynamic programming table entries.}

\begin{apprestatable}{theorem}{thmTWDP}\label{thm:tw_dp}
    \MSTGcore is solvable in $|V(\mathcal{T})|\cdot k^{O(k)}$ time, assuming we have at hand a nice tree decomposition $(\mathcal{T},\beta)$ of the underlying graph $G$ of width $k$.
\end{apprestatable}
\sv{
\begin{proof}[Proof Sketch]
    The algorithm is a modification of the dynamic programming algorithm for \textsc{Steiner Tree} problem parameterized by treewidth from~\cite{cygan2015parameterized}.
    For a node $t\in V(\mathcal{T})$, a set $X\subseteq \beta(t)$ and partition $\mathcal{P}=\{P_1,\ldots,P_q\}$ of $X$, $\DP[t,X,\mathcal{P}]$ stores the maximum value of $x(F)-w(F)$ over all forests $F$ in the subgraph $G_t$ of $G$ associated to node $t$, such that $V(F)\cap \beta(t)=X$, and the connected components of $F$ induce the partition $\mathcal{P}$ on $X$.
    We ensure $s\in \beta(t)$ for every $t\in V(\mathcal{T})$, increasing the treewidth by at most one. Given $\beta(r)=\{s\}$ for the root $r$ of $\mathcal{T}$, $(G,s,x,w)$ is a yes-instance if and only if $\DP[r,\{s\},\{\{s\}\}]>0$. The full state transition rules and correctness proof are provided in the appendix.
\end{proof}
}%sv

Since it is possible, given a graph $G$, to compute a nice tree decomposition $(\mathcal{T},\beta)$ of width $O(\tw (G))$ in $2^{O(\tw(G))}n^{O(1)}$ time~\cite{bodlaender1993linear}, we obtain \Cref{thm:tw_fpt}.

\begin{theorem}\label{thm:tw_fpt}
    \MSTGcore is \FPT parameterized by the treewidth of the input graph.
\end{theorem}

\toappendix{
\sv{\thmTWDP*}
\begin{proof}[Proof of \Cref{thm:tw_dp}]
    We design a dynamic programming algorithm on the nice tree decomposition $(\mathcal{T},\beta)$. The idea is inspired by the dynamic programming algorithm for \textsc{Steiner Tree} problem parameterized by treewidth from~\cite{cygan2015parameterized}.
    
    For a node $t\in V(\mathcal{T})$, a set $X\subseteq \beta(t)$ and a partition $\mathcal{P}=\{P_1,\ldots,P_q\}\in \Pi(X)$ the semantics is as follows. $\DP[t,X,\mathcal{P}]$ is the maximum value of $x(F)-w(F)$, where
    \begin{enumerate}[(F1)]
        \item $F$ is a forest in $G_t$,\label{item:f1}
        \item $V(F)\cap \beta(t)=X$,\label{item:f2}
        \item $F$ has $q$ connected components,\label{item:f3}
        \item the intersections of connected components of $F$ in $X$ coincide with the parts of $\mathcal{P}$.\label{item:f4}
    \end{enumerate}
    If no such $F$ exists, $\DP[t,X,\mathcal{P}]$ should be $-\infty$. Before designing the computation of entries of \DP, we add the vertex $s$ to every bag of $(\mathcal{T},\beta)$. This increases the treewidth of $(\mathcal{T},\beta)$ by at most one but it will be helpful, since we require $s$ to be in the solution. Hence, the root node $r$ contains the bag $\beta(r)=\{s\}$ and since $G_r=G$, then if $\DP[r,\{s\},\{\{s\}\}]>0$, then the input instance $(G,s,x,w)$ is a yes-instance, otherwise it is a no-instance.

    We compute the entries of the table $\DP$ in bottom up manner along $\mathcal{T}$ from leaves to root. We describe the computation for each type of node in the nice tree decomposition and we then prove the correctness of these calculations. Whenever we are computing the value of $\DP[t,X,\mathcal{P}]$, we assume that the values for $\DP[t',X',\mathcal{P}']$ are computed for $t'$ being any child of $t$ in $\mathcal{T}$ and any valid $X',\mathcal{P}'$.

    \subsubsection*{Leaf node.} If $t$ is a leaf, then $\beta(t)=\{s\}$, thus $\DP[t,\emptyset,\emptyset]=0$ and $\DP[t,\{s\},\{\{s\}\}]=x(s)$.

    \subsubsection*{Introduce node.} Suppose that $t$ is a introduce node with child $t'$ and $v$ is the introduced vertex, i.e., $\beta(t)=\beta(t')\cup \{v\}$. We compute $\DP[t,X,\mathcal{P}]$ as follows. If $v\notin X$, then $\DP[t,X,\mathcal{P}]=\DP[t',X,\mathcal{P}]$. Suppose $v\in X$. Recall that we do not introduce any edges incident to $v$, thus unless $\{v\}\in \mathcal{P}$ there is no solution, since $v$ is an isolated vertex in $G_t$. Hence $\DP[t,X,\mathcal{P}]=-\infty$ if $\{v\}\notin \mathcal{P}$. Otherwise, $\DP[t,X,\mathcal{P}]= \DP[t',X \setminus\{v\},\mathcal{P} \setminus \{v\}] + x(v)$. Thus, the formula for $\DP[t,X,\mathcal{P}]$ is as follows: 
    \[
    \DP[t,X,\mathcal{P}]=
    \begin{cases}
            \DP[t',X,\mathcal{P}] & \text{ if $v\notin X$}
            \\
            x(v)+ \DP[t',X\setminus \{v\},\mathcal{P} \setminus \{\{v\}\}] & \text{ if $v\in X, \{v\} \in \mathcal{P}$}
            \\ 
            - \infty & \text{if $v \in X, \{v\} \notin \mathcal{P}$}
    \end{cases}
    \]

    \subsubsection*{Introduce edge node.} Suppose that $t$ is an introduce edge node that introduces an edge $e = \{u,v\}$ and let $t'$ be the child of $t$ such that $\beta(t) = \beta(t')$. For every set $X \subseteq \beta(t)$ and a partition $\mathcal{P} = \{P_{1},P_{2},\ldots,P_{q}\}$ of $X$ we consider the following cases. If $\{u,v\} \not\subseteq X$ then we cannot include the edge into the solution, so in that case we consider the subproblem $\DP[t',X,\mathcal{P}]$. The same will happen if $\{u,v\}\subseteq X$, but they would belong to distinct blocks of $\mathcal{P}$ (i.e $u \in \mathcal{P}_{i}$, $v\in \mathcal{P}_{j}$ and $i\neq j$). 

    Suppose now that $\{u,v\}\subseteq X$ and, moreover, $u$ and $v$ are contained in the same block of $\mathcal{P}$, without loss of generality, $\{u,v\}\subseteq P_1$. We consider all partitions of $X$ of the form $\mathcal{P}'=\{P_u,P_v,P_2,P_3,\ldots,P_q\}$ such that $u\in P_u$ and $v\in P_v$ and $P_u\cup P_v=P_1$ and we consider the best solution in the subproblem $\DP[t',X,\mathcal{P}']$ and we add the edge $e$ to the solution. We obtain the following formula for $\DP[t,X,\mathcal{P}]$:
    \[
    \DP[t,X,\mathcal{P}]= \max\{\DP[t',X,\mathcal{P}],\sup_{\mathcal{P'}}\{\DP[t',X,\mathcal{P'}]\}-w(e)\}.
    \]
    The inner $\sup$ goes over all partitions $\mathcal{P}'$ as specified above. Recall that $\sup \emptyset=-\infty$, hence if no such $\mathcal{P}'$ exists (in particular, if $\{u,v\}\not\subseteq X$ or if $u,v$ belong to distinct blocks of $\mathcal{P}$), then only the subproblem $\DP[t',X,\mathcal{P}]$ is considered, i.e., $\DP[t,X,\mathcal{P}]=\max\{\DP[t',X,\mathcal{P}],-\infty\}=\DP[t',X,\mathcal{P}]$.
    
    \subsubsection*{Forget node.} Suppose that $t$ is a forget node with child $t'$ and $v$ is the forgotten vertex, i.e., $\beta(t)=\beta(t')\setminus \{v\}$. We compute $\DP[t,X,\mathcal{P}]$ as follows. Suppose that $\mathcal{P}=\{P_1,\ldots,P_q\}$. For $i\in[q]$ let $\mathcal{P}_i=\{P_1,P_2,\ldots,P_i\cup\{v\},\ldots,P_q\}\in \Pi(X\cup \{v\})$. The formula for $\DP[t,X,\mathcal{P}]$ is
    \[
    \DP[t,X,\mathcal{P}]=\max\left\{ \DP[t',X,\mathcal{P}],\sup_{i\in[q]}\left\{\DP[t',X\cup \{v\},\mathcal{P}_i]\right\}\right\}.
    \]
    note that the only case the inner $\sup$ may go over an emptyset is when $q=0$, i.e.,$X=\emptyset$.

    \subsubsection*{Join node.}
    Suppose that $t$ is a join node with children $t_1,t_2$ with $\beta(t)=\beta(t_1)=\beta(t_2)$. Computing $\DP[t,X,\mathcal{P}]$ boils down to combining two partial solutions from $G_{t_1}$ and $G_{t_2}$. However, we need to ensure that the merge does not create cycles. This is basically the same idea as in the treewidth dynamic programming for \textsc{Steiner Tree} in~\cite{cygan2015parameterized}. Towards this, we define for a partition $\mathcal{P}$ of $X$, an auxiliary graph $G_\mathcal{P}$ which is a forest with vertex set $X$ whose connected components coincide with $\mathcal{P}$. We say that two partitions $\mathcal{P}_1,\mathcal{P}_2$ of $X$ form an \emph{acyclic merge} $\mathcal{P}$, if the multigraph with vertex set $X$ and the edges corresponding to both $G_{\mathcal{P}_1}$ and $G_{\mathcal{P}_2}$ is acyclic. Note that we must subtract the values $x(v)$ for vertices $v\in X$, since we added these twice. Recall that there are no edges in $G_t[\beta(t)]$. We obtain the following formula for $\DP[t,X,\mathcal{P}]$: 

    \[
    \DP[t,X,\mathcal{P}]=\max_{\mathcal{P}_{1},\mathcal{P}_{2}}\left\{ \DP[t_{1},X,\mathcal{P}_{1}]+ \DP[t_{2},X,\mathcal{P}_{2}]\right\} - \sum_{v \in X}x(v).
    \]
    This concludes the description of the computation, we now prove its correctness.
\begin{lemma}{\label{lem:treewidth_dp1}}
For every node $t \in V(\mathcal{T})$, $X \subseteq \beta(t)$ and a partition $\mathcal{P} \in \Pi(X)$, if $\operatorname{DP}[t,X,\mathcal{P}] \neq -\infty$, then there exists a forest $F$ in $G_{t}$ that satisfies $\ref{item:f1}$-$\ref{item:f4}$ and it holds that $x(F) -w(F) = \operatorname{DP}[t,X,\mathcal{P}]$.
\end{lemma}
\begin{proof}
We proceed by induction on the height of the subtree $\mathcal{T}_t$. As the base case, we have the leaf node.
\subsubsection*{Leaf node} If $t$ is a leaf node, we either have $X=\emptyset$ or $X=\{s\}$. In the first case $\DP[t,\emptyset,\emptyset]=0$ and clearly the empty forest satisfies the desired conditions \ref{item:f1}-\ref{item:f4}. If $X=\{s\}$, we have $\DP[t,\{s\},\{\{s\}\}]=x(s)$. The forest $F=(\{s\},\emptyset)$ satisfies the desired properties \ref{item:f1}-\ref{item:f4} as well, note that $x(F)-w(F)=x(s)-0=0$.
\subsubsection*{Introduce node}
If $t$ is an introduce node, and $t'$ is a child of $t$ then, if $v \notin X$ then $\DP[t,X,\mathcal{P}] = \DP[t',X,\mathcal{P}]$, then from the induction hypothesis there exists a forest $F$ satisfying $\ref{item:f1}$-$\ref{item:f4}$ and the $\DP[t,X,\mathcal{P}] = \DP[t',X,\mathcal{P}] = x(F) - w(F)$. For the second branch, we have $v \in X$ and $\{v\} \in \mathcal{P}$. Let $F'$ be the forest for $t',X\setminus \{v\},\mathcal{P}\setminus \{\{v\}\}$, which exists by the induction hypothesis. We let $F=F'+v$ and $x(F)-w(F)=x(F')-w(F') +x(v)=\DP[t',X\setminus \{v\},\mathcal{P}\setminus \{\{v\}\}]+x(v)$, where the last equality holds by the induction hypothesis that $\DP[t',X \setminus \{v\},\mathcal{P}\setminus \{v\}] = x(F') - w(F')$. To see that $F$ satisfies conditions $\ref{item:f1}$-$\ref{item:f4}$, first $F'$ was a forest in $G_{t'}$ and after the introduce of $v$ the $F = F' \cup \{v\}$ will remain a forest in $G_{t}$ and the number of connected components will be increased by 1, because the forest $F'$ was consisted of exactly $q-1$ connected components by the induction hypothesis and after we introduced $v$ in $F'$ the new forest $F$ will have $q$ connected components because the newly introduced vertex will be a singleton connected component, so the desired conditions $\ref{item:f1}$-$\ref{item:f4}$ are satisfied. We note also that the third branch of the $\DP[t,X,\mathcal{P}]$ will be never visited because we have assumed that the $\DP[t,X,\mathcal{P}] \neq -\infty$.
\subsubsection*{Introduce Edge node}
If $t$ is an introduce edge node that introduces $e = \{u, v\}$. Let $t'$ be the child of $t$. We either do not include one of the endpoints of the edge $e = \{u,v\}$ in $X$, that means $u \notin X$ or $v \notin X$ and then we have that $\DP[t,X,\mathcal{P}] = \DP[t',X,\mathcal{P}]$ and then from the induction hypothesis there exists a forest $F$ satisfying \ref{item:f1}-\ref{item:f4} and $\DP[t,X,\mathcal{P}] = \DP[t',X,\mathcal{P}] = x(F) - w(F)$. On the other hand if $\{u,v\} \subseteq X$ and both endpoints are in the same block of $\mathcal{P} = \{\mathcal{P}_{1},\mathcal{P}_{2},\ldots,\mathcal{P}_{q}\}$, without loss of generality, $\{u,v\} \subseteq P_{1}$, then we have $\DP[t,X,\mathcal{P}]=\DP[t',X,\mathcal{P}']-w(e)$ for some $\mathcal{P}'=\{P_u,P_v,P_2,P_3,\ldots,P_q\}$. Let $F'$ be the forest for $t',X,\mathcal{P}'$. We let $F = F'+e$ and now we get that $x(F) = x(F')$, because $F'$ already contains both $u$ and $v$. Also $w(F) = w(F') + w(e)$. We have $x(F)-w(F)=x(F')-w(F')-w(e)=\DP[t',X,\mathcal{P}']-w(e)=\DP[t,X,\mathcal{P}]$.
\subsubsection*{Forget node} If $t$ is a forget node with child $t'$, then $G_t=G_{t'}$, hence any valid forest for $t'$ is also a valid forest for $t$.
\subsubsection*{Join node}
If $t$ is a join node with children $t_1,t_2$, note that $\DP[t,X,\mathcal{P}]=\DP[t_1,X,\mathcal{P}_1]+\DP[t_2,X,\mathcal{P}_2] - \sum_{v \in X}x(v)$ for some $\mathcal{P}_1$ and $\mathcal{P}_2$, such that $\mathcal{P}$ is an acyclic merge of $\mathcal{P}_1$ and $\mathcal{P}_2$. Let $F_1$ and $F_2$ be the forests for the entries $\DP[t_1,X,\mathcal{P}_1]$ and $\DP[t_2,X,\mathcal{P}_2]$, respectively. Let $F=F_1\cup F_2$, more precisely $V(F)=V(F_1)\cup V(F_2)$ and $E(F)=E(F_1) \cup E(F_2)$. Note that $V(F_1)\cap V(F_2)=X$, hence $x(F)=x(F_1)+x(F_2)-x(X)$. Since $G_t[\beta(t)]$ is edgeless and there are no other vertices common to both $F_1$ and $F_2$ other than those in $X$, we have $E(F_1)\cap E(F_2)=\emptyset$. Thus $w(F)=w(F_1)+w(F_2)$. Then we have the following,

\begin{align*}
    x(F)-w(F)&=x(F_1)+x(F_2)-x(X)-w(F_1)-w(F_2)=\\&=\left(x(F_1)-w(F_1)\right)+\left(x(F_2)-w(F_2)\right)-x(X)=\\&=\DP[t_1,X,\mathcal{P}_1]+\DP[t_2,X,\mathcal{P}_2] - \sum_{v\in X}x(v)=\\&=\DP[t,X,\mathcal{P}]    
\end{align*}

It remains to argue that $F$ is indeed a forest and satisfies the desired properties. Consider contracting every vertex in $F_1$ except those in $X$. This creates a forest with connected components coinciding with those in $G_{\mathcal{P}_1}$. Similarly for $F_2$. Note that any path in $F$ corresponds to some path in $G_\mathcal{P}$, but since $\mathcal{P}$ is acyclic merge of $\mathcal{P}_1$ and $\mathcal{P}_2$, $F$ is indeed a forest and the connected components coincide with $\mathcal{P}$.

\end{proof}

\begin{lemma}{\label{lem: treewidth_dp2}}
For every node $t \in V(\mathcal{T})$, $X \subseteq\beta(t)$ and a partition $\mathcal{P} \in \Pi(X)$, if there exists a forest $F$ in $G_{t}$ satisfying $\ref{item:f1}$-$\ref{item:f4}$ then $\operatorname{DP}[t,X,\mathcal{P}] \geq x(F) - w(F)$.
\end{lemma}
\begin{proof}
We proceed by induction on the height of the subtree $\mathcal{T}_t$. As the base case, we have the leaf node.
\subsubsection*{Leaf node} If $t$ is a leaf node and $(F,X,\mathcal{P})$ are given, then there are two possibilities. Either $F$ is an empty forest, i.e., $F=(\emptyset,\emptyset)$, then $x(F)-w(F)=0$ and $\DP[t,\emptyset,\emptyset]=0$, so the claim holds. Or $F=(\{s\},\emptyset)$ and then we have $\DP[t,\{s\},\{\{s\}\}]=x(s)=x(F)-w(F)$, so the claim holds.

\subsubsection*{Introduce node}
Let $t$ be an introduce node and let $v$ be the introduced vertex. Suppose $(F,X,\mathcal{P})$ are given. If $v\notin X$, then clearly $F$ also satisfies the assumptions for $t',X,\mathcal{P}$, hence $\DP[t,X,\mathcal{P}]=\DP[t',X,\mathcal{P}]\geq x(F)-w(F)$ where the inequality follows from the induction hypothesis. Otherwise we have $v\in X$.In that case, since $v$ is an isolated vertex in $G_t$, it has to be a singleton connected component in $F$, in other words, $\{v\}\in \mathcal{P}$. Hence the forest $F':=F-v$ has components coinciding with $\mathcal{P}\setminus \{\{v\}\}$ and $V(F)\cap \beta(t')=X\setminus v$. By induction hypothesis, we know that $\DP[t',X\setminus \{v\},\mathcal{P}\setminus \{\{v\}\}]\geq x(F')-w(F')$. Note that $x(F')=x(F)-x(v)$ and $w(F')=w(F)$. Hence $x(F')-w(F')=x(F)-w(F)-x(v)$. Thus
\begin{align*}
    \DP[t,X,\mathcal{P}]&=x(v)+\DP[t',X\setminus \{v\},\mathcal{P}\setminus \{\{v\}\}]\geq\\&\geq x(v)+x(F')-w(F')=\\&= x(v)+x(F)-w(F)-x(v)=\\&=x(F)-w(F)
\end{align*}
\subsubsection*{Introduce Edge node} If $t$ is an introduce edge node and let also $e = \{u,v\}$ be the introduced edge. We suppose that $(F,X,\mathcal{P})$ are given. If the edge $e \notin F$, then we cannot include this edge in the solution and we consider the problem in the child $t'$ where the forest $F$ in that case satisfies the assumptions for $(t',X,\mathcal{P})$, so $\DP[t,X,\mathcal{P}] = \DP[t',X,\mathcal{P}] \geq x(F) - w(F)$, where the last inequality follows from the induction hypothesis. On the other hand if $e \in F$, then we consider both endpoints to be part of some block of $\mathcal{P}$. Hence we consider the forest $F' = F-e$  and the partition $\mathcal{P'} = \{C_{u},C_{v}, P_{2},\ldots,P_{q}\}$, where $C_{u},C_{v}$ are the connected components that were created after the removal of the edge $e$ and it holds that $u \in V(C_{u}),v \in V(C_{v})$. Now if we consider, 
\begin{align*}
\sup_{i\in[q]}\{\DP[t',X\cup \{v\},\mathcal{P}_i]\} \geq \DP[t',X\cup \{v\},\mathcal{P}_i] \geq x(F') - w(F') = x(F)-w(F) +w(e).
\end{align*}
The second inequality comes from the application of the induction hypothesis and the rest from the definition of the forest $F'$.
\subsubsection*{Forget Node} If $t$ is a forget node with child $t'$ and suppose also that $v$ is the forgotten vertex and that $(F,X,P)$ are given. If $v \notin X$, then clearly $F$ also satisfies the assumptions for $t',X,\mathcal{P}$, hence $\DP[t,X,\mathcal{P}]=\DP[t',X,\mathcal{P}]\geq x(F)-w(F)$ where the inequality follows from the induction hypothesis. Otherwise, $v \in X$, then we consider the partition $\mathcal{P} = \{P_{1},P_{2},\ldots,P_{q}\}$ and that $v \in V(F)$, note that the vertex $v$ can belong to the connected component $\beta(t')$ and we will have the following: 
\begin{align*}
\sup_{i\in[q]}\{\DP[t',X\cup \{v\},\mathcal{P}_i]\} \geq \DP[t',X\cup \{v\},\mathcal{P}_i] \geq x(F) -w(F)
\end{align*}
where the last inequality comes from the induction hypothesis.
\subsubsection*{Join Node}
If $t$ is a join node with children $t_{1}$ and $t_{2}$ and suppose that $(F_{1},F_{2},X,\mathcal{P}_{1},\mathcal{P}_{2})$ are given, then if we consider the acyclic merge of $F_{1}$ and $F_{2}$, we know that from the induction hypothesis that $F_{1}$ is a forest for the entry $\DP[t_{1},X,\mathcal{P}_{1}]$ and $F_{2}$ is a forest for the entry $\DP[t_{2},X,\mathcal{P}_{2}]$ then if we consider $F':=F_{1} \cup F_{2}$ then $x(F) = x(F_{1} \cup F_{2}) = x(F_{1}) + x(F_{2}) - x(X)$ and $w(F) = w(F_{1} \cup F_{2}) = w(F_{1}) + w(F_{2}) - w(F_{1} \cap F_{2})$ then if we consider the difference we have the following: 
\begin{align*} 
x(F) - w(F) &= (x(F_{1}) - w(F_{1})) + (x(F_{2}) - w(F_{2})) - x(X) \leq \\& 
\DP[t_{1},X,\mathcal{P}_{1}]+ \DP[t_{2},X,\mathcal{P}_{2}] - x(X) \leq \\&
\max\{\DP[t_{1},X,\mathcal{P}_{1}]+ \DP[t_{2},X,\mathcal{P}_{2}]\} - x(X).
\end{align*}
here we denote with $x(X) = \sum_{v \in X} v(x)$ 

\end{proof}

The correctness of the algorithm follows from \Cref{lem:treewidth_dp1,lem: treewidth_dp2}. It remains to argue about the time complexity. Note that the size of the table is upper bounded by $|V(\mathcal{T})|\cdot 2^{k+1}\cdot B_{k+1}$, where $B_k$ is the $k$-th Bell number, representing the number of partitions of a $k$-element set. It is not hard to convince one that $|B_k|\leq k^k$, hence the size of the table is upper bounded by $|V(\mathcal{T})|\cdot (k+1)^{2k+2}$. Note that the value of each entry can be computed in $O(B_{k+1}^2)\leq (k+1)^{2(k+1)}$ time (the main bottleneck is the join node, where we are iterating over all pairs of partitions $\mathcal{P}_1,\mathcal{P}_2$). Hence, the total time needed to fill the table is upper bounded by $|V(\mathcal{T})|\cdot k^{O(k)}$, as promised and this finishes the proof of \Cref{thm:tw_dp}.

\end{proof}
}%toappendix

\appsection{Signed Neighborhood Diversity}{sec:snd}
In this section we parameterize \MSTGcore by the (signed) neighborhood diversity of the underlying weighted graph $(G,w)$. The problem is easily seen to be \NP-hard on cliques, as we can simply add any missing edges with a large enough weight to ensure they are unusable in any feasible solution. Because cliques have standard neighborhood diversity of $1$, it follows that the problem is \NP-hard on graphs with neighborhood diversity $1$. This motivates our focus on signed neighborhood diversity ($\operatorname{snd}$).

The main result of this section is a simple single-exponential \FPT algorithm parameterized by $\operatorname{snd}$~(\Cref{thm:snd_running_time,thm:snd_fpt}). As a natural extension, we also show that \MSTGcore admits a polynomial kernel parameterized by $\operatorname{snd}$~(\Cref{thm:snd_kernel_bound,thm:snd_kernel_ft}), although this does not yield a better running time of the algorithm.

We assume that along with the input instance $(G,s,x,w)$, we are given a $w$-uniform partition $V(G)=V_1\cup V_2\cup \cdots \cup V_d$. Moreover, we assume that the partition is \emph{nice}, meaning that $V_i=\{s\}$ for some $i\in[d]$. Note that any $w$-uniform partition of $V(G)$ can be transformed into a nice one by isolating the supply vertex~$s$ into its own singleton part. Because this operation increases the size of the partition by at most $1$, for any choice of $s$, any weighted graph $(G,w)$ admits a nice $w$-uniform partition of size at most $\operatorname{snd}(G)+1$.

Let $V_i^{\max}=\{v_i\in V_i\mid \forall v \in V_i:x(v_i)\geq x(v)\}$ be the set of vertices $v_i\in V_i$ maximizing $x(v_i)$. Let $T$ be a solution for $(G,s,x,w)$. The part $V_i$ is \emph{active} with respect to $T$ if $V(T)\cap V_i\neq \emptyset$.
We call $T$ \emph{pretty} if it maximizes $x(T)-w(T)$, among those minimizes $|V(T)|$ and, moreover, contains at most one non-leaf vertex of $T$ from each active part $V_i$.

\appsubsection{\FPT algorithm for Signed Neighborhood Diversity}{subsec:snd_fpt}

We first show in \Cref{lem:nd_making_allbutone_leaf} that a pretty solution always exists (assuming that some solution exists). We then show in \Cref{lem:nd_first_lemma} that a pretty solution contains some vertex $v_i^*\in V_i^{\max}$ from every active part $V_i$. Moreover, any vertex may be chosen as the candidate non-leaf vertex of an active part. Finally, we characterize in~\Cref{lem:nd_leaves_onedir,lem:nd_leaves_twodir} which vertices $v\in V_i\setminus \{v_i^*\}$ should be included in a pretty solution based on the weight of its incident edges connecting $v$ to some active part of the solution.

\begin{apprestatable}{lemma}{lemndmakingallbutoneleaf}
\label{lem:nd_making_allbutone_leaf}
    Let $(G,s,x,w)$ be a yes-instance of \MSTGcore, then $(G,s,x,w)$ admits a pretty solution. Moreover, the candidate non-leaf vertex of every active part may be chosen arbitrarily.
\end{apprestatable}
\toappendix{
\sv{\lemndmakingallbutoneleaf*}
\begin{proof}
Let $T$ be a solution of $(G,s,x,w)$ maximizing $x(T)-w(T)$ and among those minimizing $|V(T)|$. Let $V_i$ be an active part and let $r_i\in V(T)\cap V_i$ be chosen arbitrarily. We make every other vertex $v\in (V(T)\cap V_i)\setminus \{r_i\}$ a leaf. 
We modify $T$ by a local replacement around $V_i$. Imagine that the edges of $T$ are oriented away from $s$. Hence, every vertex except $s$ has exactly one incoming edge in this sense.
If $v\in (V(T)\cap V_i)\setminus \{r_i\}$ is not a leaf of $T$, we replace all but one edge of the form $\{v,v'\}$ by $\{r_i,v'\}$, where $v'\in V_j$ (possibly $i=j$), thus $v$ becomes a leaf. The one edge that we do not replace is the edge that is pointing to $v$ in the orientation mentioned above. This ensures that the resulting graph remains a tree. Denote the newly created tree by $T'$. By the definition of signed neighborhood diversity $w(\{v,v'\})=w(\{r_i,v'\})$ because both $v$ and $r_i$ belong to $V_i$. Thus $w(T)=w(T')$. Note that the vertex set didn't change, hence $T'$ is a pretty solution for $(G,s,x,w)$.
\end{proof}
}%toappendix

\begin{apprestatable}{lemma}{lemndfirstlemma}\label{lem:nd_first_lemma}
    Let $T$ be a pretty solution of $(G,s,x,w)$ with $V(T)\cap V_i\neq \emptyset$. Then $V(T)\cap V_i^{\max}\neq \emptyset$.
\end{apprestatable}
\toappendix{
\sv{\lemndfirstlemma*}
\begin{proof}
    Suppose for the sake of contradiction that $V(T)\cap V_i^{\max}=\emptyset$. Let $v\in V(T)\cap V_i$. Create a new solution $T'$ by removing $v$ and adding $v_i\in V_i^{\max}$ and replacing all the edges of the form $\{v,v'\}$ by the edges $\{v_i,v'\}$ for every $v'\in N_T(v)$. Note that $w(T')=w(T)$ by the definition of signed neighborhood diversity, because both $v$ and $v_i$ are from $V_i$, hence $w(\{v,v'\})=w(\{v_i,v'\})$ for any $v'\in V(G)$. Note that $x(T')=x(T)-x(v)+x(v_i)$, but by the choice of $v_i$ and $v$, we have $x(v_i)>x(v)$ (because $v\notin V_i^{\max}$, thus $x(v_i)-x(v)>0$, thus $x(T')-w(T')>x(T)-w(T)$, which contradicts the assumption that $T$ is pretty.
\end{proof}
}%toappendix

\begin{apprestatable}{lemma}{lemndleavesonedir}\label{lem:nd_leaves_onedir}
    Let $T$ be a pretty solution of $(G,s,x,w)$. Suppose that $u\in V(T)$ and let $e=\{u,v\}$ for some $v\in V(G)$. If $x(v)>w(e)$, then $v\in V(T)$.
\end{apprestatable}
\toappendix{
\sv{\lemndleavesonedir*}
\begin{proof}
    This is immediate as otherwise $T$ is not pretty. If $v\notin V(T)$, then we could attach $v$ as a leaf via $e$ and increase $x(T)-w(T)$.
\end{proof}
}%toappendix

\begin{apprestatable}{lemma}{lemndleavestwodir}\label{lem:nd_leaves_twodir}
    Let $T$ be a pretty solution and let $V_i$ be an active part of $T$. Suppose that $r_i\in V(T)\cap V_i^{\max}$ and let $v\in V_i\setminus \{r_i\}$. If an edge $\widehat{e}\ni v$ minimizes the weight $w(\widehat{e})$ among all edges $e\ni v$ connecting $v$ to some active part of $T$, then if $x(v)\leq w(\widehat{e})$, then $v\notin V(T)$. 
\end{apprestatable}
\toappendix{
\sv{\lemndleavestwodir*}
\begin{proof}
    Observe first that the assumption $r_i\in V(T)\cap V_i^{\max}$ is valid due to \Cref{lem:nd_first_lemma}. 
    Let $v$ be connected to $v_j\in V_j$, where $V_j$ is an active part of $T$ (possibly $V_i=V_j$).
    Suppose for the sake of contradiction that $v\in V(T)$. Since $r_i$ is also in~$T$, by \Cref{lem:nd_making_allbutone_leaf}, we may assume without loss of generality that everyone in $(V(T)\cap V_i)\setminus \{r_i\}$ is a leaf, including~$v$. Since $v$ is a leaf it is connected to $T$ via an edge $e'\in E(T)$. By the definition of $\widehat{e}$, we get $w(e')\geq w(\widehat{e})\geq x(v)$. Let $T'=T\setminus \{v\}$. Since $v$ was a leaf of $T$, it follows that $T'$ is a tree. If $x(v)<w(e')$, we get $x(T')-w(T')>x(T)-w(T)$. If $x(v)=w(e')$, then $x(T')-w(T')=x(T)-w(T)$, but $|V(T')|<|V(T)|$. This is contradicting the choice of $T$ as a pretty solution in both cases.
\end{proof}
}%toappendix

With \Cref{lem:nd_first_lemma,lem:nd_making_allbutone_leaf,lem:nd_leaves_onedir,lem:nd_leaves_twodir} at hand, we may present the main idea of the algorithm (see \Cref{alg:nd} for a compact pseudocode). We guess the active parts $\mathcal{X}\subseteq \{V_1,\ldots,V_d\}$ of the solution. Note that $\{s\}$ must always be an active part. We then arbitrarily choose one of the vertices $v_i^*\in V_i^{\max}$ for every $V_i\in \mathcal{X}$. Let $M_{\mathcal{X}}=\{v_i^*\mid V_i\in \mathcal{X}\}$. These vertices will be the candidate non-leaf vertices of the guessed active parts. We then compute the minimum spanning tree using a standard algorithm for the subgraph $G[M_{\mathcal{X}}]$. The remaining vertices from other active parts are either attached greedily as leaves or not touched by the solution at all according to \Cref{lem:nd_leaves_onedir,lem:nd_leaves_twodir}.

\algrenewcommand\algorithmicrequire{\textbf{Input:}}
\algrenewcommand\algorithmicensure{\textbf{Output:}}

\begin{algorithm}[bt]
\caption{\FPT Algorithm parameterized by the signed neighborhood diversity}
\label{alg:nd}
\begin{algorithmic}[1]
\Require Instance $(G,s,x,w)$ of \MSTGcore, nice $w$-uniform partition $V_1,\ldots,V_d$ of $V(G)$
\Ensure \texttt{YES} if $(G,s,x,w)$ is a yes-instance, \texttt{NO} otherwise.
\State Guess a subset $\mathcal{X} \subseteq \{V_1, V_2, \ldots, V_d\}$ such that $\{s\} \in \mathcal{X}$ 

\For{each $V_i \in \mathcal{X}$}
    \State Let $v_i^{*} \in V_i$ such that $\forall v \in V_i: x(v_i^{*}) \ge x(v)$\Comment{If there are multiple such $v_i^{*}$'s, pick one arbitrarily.}
\EndFor
\State $M_\mathcal{X} \gets \{v_i^{*} \mid V_i \in \mathcal{X}\}$
\If{$G[M_\mathcal{X}]$ is not connected}
    \State \Return \texttt{NO}
\EndIf
\State Compute a minimum spanning tree $T^*$ of $G[M_\mathcal{X}]$ 

\For{each $V_i \in \mathcal{X}$}
    \For{each $v\in V_i \setminus M_{\mathcal{X}}$}
        \State Let $\widehat{e}$ be a minimum-weight edge connecting $v$ to some $v_j \in V_j$ with $V_j \in \mathcal{X}$
        \If{$x(v) > w(\widehat{e})$} 
            \State $T^* \gets T^* + \widehat{e}$ \Comment{Add $v_{i}$ as a leaf to $T^*$ via $\widehat{e}$.}
        \EndIf
    \EndFor
\EndFor

\State Let $T^*$ be the resulting tree
\If{$x(T^*) > w(T^*)$}
    \State \Return \texttt{YES}
\EndIf
\State \Return \texttt{NO}
\end{algorithmic}
\end{algorithm}

It is clear that if \Cref{alg:nd} returns \texttt{YES}, then the input instance $(G,s,x,w)$ is a yes-instance. Note that $s\in V(T^*)$ for every tree $T^*$ considered by the algorihtm, because $\{s\}\in \mathcal{X}$ (hence $s\in M_{\mathcal{X}}$) for every choice of $\mathcal{X}$. It remains to argue that if $(G,s,x,w)$ is a yes-instance, then \Cref{alg:nd} outputs \texttt{YES} for some guess $\mathcal{X}\subseteq \{V_1,V_2,\ldots, V_d\}$. We prove this in \Cref{lem:alg_snd_correctness}.

\begin{apprestatable}{lemma}{lemalgsndcorrectness}\label{lem:alg_snd_correctness}
    If $(G,s,x,w)$ is a yes-instance, then \Cref{alg:nd} outputs \texttt{YES} for some $\mathcal{X}\subseteq \{V_1,V_2,\ldots, V_d\}$.
\end{apprestatable}

\toappendix{
\sv{\lemalgsndcorrectness*}
\begin{proof}
    Suppose that $(G,s,x,w)$ is a yes-instance and let $T$ be a pretty solution for $(G,s,x,w)$. Let $\mathcal{X}=\{V_i\mid V(T)\cap V_i\neq \emptyset\}$ be the set of active parts of $T$. We claim that for this choice of $\mathcal{X}$, \Cref{alg:nd} outputs a solution $T^*$ with $x(T^*)=x(T)$ and $w(T^*)=w(T)$. Since $T$ is pretty, by \Cref{lem:nd_first_lemma}, for every $V_i\in \mathcal{X}$, $V(T)\cap V_i^{\max}\neq \emptyset$. Note that we can also assume that $v_i^*$ picked by the algorithm is in $V(T)\cap V_i^{\max}$ as otherwise we can make a local replacement as in the proof of \Cref{lem:nd_first_lemma} without changing $w(T)$ or $x(T)$. We now apply \Cref{lem:nd_making_allbutone_leaf} for $r_i=v_i^*$. Hence, we assume that all but $v_i^*$ in $V(T)\cap V_i$ are leaves. Note that this still doesn't change $x(T)$ or $w(T)$. Recall that $M_\mathcal{X} = \{v_i^{*} \mid V_i \in \mathcal{X}\}$. Consider the `spine' $T[M_\mathcal{X}]$. In other words, $T[M_\mathcal{X}]$ is created from $T$ by removing all the leaves that are outside $M_{\mathcal{X}}$, keeping only $v_i^*$ for every $V_i\in \mathcal{X}$.
    
    \begin{claim}
        $T[M_{\mathcal{X}}]$ is a minimum spanning tree of $G[M_{\mathcal{X}}]$.
    \end{claim}
    \begin{proof}
        Suppose for the sake of contradiction that there is a spanning tree $T'$ with $w(T')<w(T[M_{\mathcal{X}}])$. Let $\widehat{T}$ be a tree that is created from $T'$ by attaching all vertices in $V(T)\setminus M_{\mathcal{X}}$ as leaves to $T'$. This is a contradiction with the assumption that $T$ is pretty because $x(T)=x(\widehat{T})$ while $w(\widehat{T})<w(T)$, hence $x(\widehat{T})-w(\widehat{T})>x(T)-w(T)$.
    \end{proof}
     Note that the presence of vertices $v_i\in V_i\setminus M_{\mathcal{X}}$ in $V(T)$ is determined by \Cref{lem:nd_leaves_onedir,lem:nd_leaves_twodir}. For every $v\in V_i\setminus \{v_i^*\}$, $v\in V(T)$ if and only if the minimum weight edge $\widehat{e}$ connecting $v$ to some $v_j\in V_j$ such that $V(T)\cap V_{i} \neq \emptyset$ satisfies $w(\widehat{e})<x(v)$. Since \Cref{alg:nd} computes a minimum spanning tree of $G[M_{\mathcal{X}}]$ on line $7$ and the for-loops on lines $8$-$12$ in \Cref{alg:nd} pick exactly the vertices and edges to $T^*$ according to the above characterization, we have $x(T^*)=x(T)$ and $w(T^*)=w(T)$ and this finishes the proof of \Cref{lem:alg_snd_correctness}.
\end{proof}
}%toappendix

\begin{apprestatable}{theorem}{thmsndrunningtime}\label{thm:snd_running_time}
\MSTGcore is solvable in $O(2^{d}\log (d)(n+m))$ time, assuming we have at hand a nice $w$-uniform partition of the underlying weighted graph $(G,w)$ of size $d$.
\end{apprestatable}
\toappendix{
\sv{\thmsndrunningtime*}
\begin{proof}
We run \Cref{alg:nd}. Its correctness follows from \Cref{lem:alg_snd_correctness}. It remains to analyse the running time.

The algorithm starts in line $1$ by enumerating all possible subsets $\mathcal{X} \subseteq \{V_1,V_2,\ldots,V_{d}\}$ such that $\{s\}\in \mathcal{X}$. Note that there are at most $2^{d-1}$ such subsets and they can be enumerated in $O(2^d)$ time. In lines $2$-$3$, the algorithm searches for such vertices that are maximized among all the others in each $V_i$ and create $M_{\mathcal{X}}$. This takes $O(n)$ time in total. We proceed to the computation of the minimum spanning tree $T^*$ on line $7$.

\begin{claim}
    The computation of the minimum spanning tree of $G[M_{\mathcal{X}}]$ can be done in $O(m\log d)$ time.
\end{claim}
\begin{proof}  
    We use a standard textbook implementation of Jarník's~\cite{jarnik30} (also known as Prim's~\cite{prim57}) algorithm for the minimum spanning tree with binary heap. We start with a spanning tree containing arbitrary vertex~$v_0$. Vertices not yet added to the tree are kept in a binary heap. They key of a vertex $v$ is the shortest edge from the so-far constructed tree to $v$. Note that the size of the heap is at most $|M_{\mathcal{X}}|\leq d$ hence every heap operation takes $O(\log d)$ time. Every new vertex $v$ added to the spanning tree under construction performs $O(\deg v)$ operations with the heap. Thus, in total, we get the running time of $O(m\log d)$, as claimed.
\end{proof} 

Note that we can incorporate the check whether $G[M_{\mathcal{X}}]$ is not connected inside Jarník's algorithm to possibly terminate in the if condition on lines $5$-$6$. Finally, the lines $8$-$12$ take $O(n+m)$ time. The two for loops on lines $8$-$9$ take $O(n)$ time in total and the minimum on line $10$ can be found by simply iterating over the adjacency list of $v$ considering only the edges connected to some $v_j\in V_j$ with $V_j\in\mathcal{X}$. Note that we can even store the tree $T^*$ using e.g. parent pointers for every vertex (this is essentially a consequence of using Jarník's algorithm) and everything else is attached as a leaf to existing vertices of $T^*$, so even attaching leaves on line $12$ takes constant time. Hence, the total running time for every guess of $\mathcal{X}$ is $O((m+n)\log d)$ hence the total running time is $O(2^{d}\log (d)(n+m))$, as promised in the theorem.
\end{proof}

}%toappendix
Since a $w$-uniform partition of size $\operatorname{snd}$ can be computed in polynomial time (and hence a nice one of size at most $\operatorname{snd}+1$), we obtain the following theorem.
\begin{theorem}\label{thm:snd_fpt}
    \MSTGcore is \FPT parameterized by the signed neighborhood diversity of the underlying weighted graph.
\end{theorem}

\appsubsection{Polynomial kernel for Signed Neighborhood Diversity}{subsec:snd_kernel}

To obtain the polynomial kernel, we introduce a reduction rule that merges vertices of the same `type' within each part of a nice $w$-uniform partition $\{V_1,V_2,\ldots, V_d\}$. Let $T$ be a pretty solution and let $V_i$ be an active part of $T$.\footnote{All variables introduced from this point should have additional superscript $i$ to signify their dependence on the part $V_i$. Nonetheless, for the sake of readability, we omit it.} Recall that there is at most one vertex $v\in V_i\cap V(T)$ that is not a leaf of $T$. If there is one, we may assume that this vertex is $v_i^*\in V_i^{\max}$, and any other $v\in (V_i\cap V(T))\setminus \{v_i^*\}$ is a leaf of $T$ connected to some active part of $T$. Specifically, the behavior of $v\in V_i\setminus \{v_i^*\}$ depends only on the set of unique weights in the set $W=\{w(e)\mid e\ni v\}=\{w_1,w_2,\ldots,w_q\}$ (where $w_1<w_2<\cdots<w_q$, $q\leq d$). Note that since $V_i$ is a part of a $w$-uniform partition, $W$ does not depend on the choice of $v$. The set $W$ partitions the nonnegative real line into $q+1$ intervals: $I_0=[0,w_1]$, $I_j=(w_j,w_{j+1}]$ for $j\in[q-1]$, and $I_q=(w_q,+\infty)$. Let $w_{\min}^T=\displaystyle\min_{\text{$V_j$ is active part of $T$}}w(\{v,v_j\})$, where $v_j\in V_j$ is chosen arbitrarily. Note that $w_{\min}^T$ does not depend on the choice of $v\in V_i$ nor the choice of $v_j\in V_j$. Observe that $w_{\min}^T\in W$. By \Cref{lem:nd_leaves_onedir,lem:nd_leaves_twodir} $v\in V_i\setminus \{v_i^*\}$ is included in a pretty solution $T$ if and only if $x(v)>w_{\min}^T$. If $x(v)\in I_j$, this condition is equivalent to $w_{\min}^T\in \{w_1,w_2,\ldots,w_j\}$. Crucially, this condition is identical for all vertices whose value falls to the same interval $I_j$. Thus, regardless of active parts of a pretty solution, either all vertices in $V_{i,j}=\{v\in V_i\setminus \{v_i^*\}\mid x(v)\in I_j \}$ are included as leaves, or none of them are touched by the solution. This motivates the definition of \Cref{rrule:snd_rule}.

%Recall that in a pretty solution $T$, there is at most one $v\in V_i\cap V(T)$ which is a non-leaf. We may assume that this vertex is $v_i^*\in V_i^{\max}$ by~\Cref{lem:nd_first_lemma}, any other vertex $v\in V_i\cap V(T)\setminus \{v_i^*\}$ is a leaf of $T$ connected to some active part of $T$. Specifically, the behavior of $v\in V_i\setminus \{v_i^*\}$ depends only on the set of unique weights in the set $W_i=\{w(e)\mid e\ni v\}$. Note that since $V_i$ is a part of a $w$-uniform partition, $W_i$ does not depend on the choice of $v$. Let $W_i=\{w_1,\ldots,w_p\}$ where $w_1<w_2<\cdots<w_p$. Note that $p\leq d$. These weights partition the nonnegative real line into $p+1$ intervals: $I_0=[0,w_1],I_j=(w_j,w_{j+1})$ for $j\in[p-1]$, and $I_p=(w_p,+\infty)$. Let $w_{\min}^T=\displaystyle\min_{\text{$V_j$ is active part of $T$}}w(\{v_i,v_j\})$ where $v_i\in V_i$ and $v_j\in V_j$ is chosen arbitrarily. Note that $w(\{v_i,v_j\})$ is the same for any choice of $v_i,v_j$, and $w_{\min}^T\in W_i$. By \Cref{lem:nd_leaves_onedir,lem:nd_leaves_twodir} $v\in V_i\setminus \{v_i^*\}$ is included in the pretty solution $T$ if and only if $x(v)>w_{\min}^T$. If $x(v)\in I_j$, this condition is equivalent to $w_{\min}\leq w_j$. Crucially, this condition is identical for all vertices whose value falls to the same interval $I_j$. Thus, regardless of active parts of a pretty solution, either all vertices in $V_{i,j}=\{v\in V_i\setminus \{v_i^*\}\mid x(v)\in I_j \}$ are included as leaves, or none of them are touched by the solution. This motivates the definition of \Cref{rrule:snd_rule}.

\begin{rrule}\label{rrule:snd_rule}
    Let $(G,s,x,w)$ be an instance of \MSTGcore, and let $V_1,\ldots, V_d$ be a nice $w$-uniform partition of $V(G)$. If for some $i,j$ we have $|V_{i,j}|>1$ (as defined above), merge all vertices in $V_{i,j}$ into a new vertex $v_{\Sigma}$, creating a new graph $G'$. The new allocation is $x'(v_\Sigma)=\sum_{v\in V_{i,j}}x(v)$. For each $u\in N_G(v)$, we set $w'(\{v_{\Sigma},u\})=\sum_{u\in N_G(v)}w(\{v,u\})=|V_{i,j}|\cdot w(\{v,u\})$, where $v\in V_i$ is chosen arbitrarily. All other vertices and edges have their values and weights unchanged. Output the instance $(G',s,x',w')$.
\end{rrule}

\begin{apprestatable}{lemma}{lemrrsndkernel}\label{lem:rr_snd_kernel}
    \Cref{rrule:snd_rule} is correct.
\end{apprestatable}
\toappendix{
\sv{\lemrrsndkernel*}
\begin{proof}
    Let $(G,s,x,w)$ be the original instance of \MSTGcore, and let $(G',s,x',w')$ be the instance obtained by applying \Cref{rrule:snd_rule}. We show that $(G,s,x,w)$ is a yes-instance if and only if $(G',s,x',w')$ is a yes-instance.

    Observe first that since the $w$-uniform partition is nice, the vertex $s$ is untouched by the reduction rule. If $V_i=\{s\}$, then $v_i^*=s$ and thus $V_{i,j}=\emptyset$ for every $j$.

    $\Rightarrow$: Let $T$ be a pretty solution of $(G,s,x,w)$. By \Cref{lem:nd_making_allbutone_leaf}, we can assume that all vertices in $V_{i,j}$ are leaves of $T$ (or not in $T$). Since all vertices in $V_{i,j}$ fall into the same interval $I_j$, by \Cref{lem:nd_leaves_onedir,lem:nd_leaves_twodir} they either all belong to $V(T)$ or none do. If they do not belong to $V(T)$, then $T'=T$ is a valid solution for $(G',s,x',w')$. If they all belong to $V(T)$, they are connected as leaves of $T$ to some active parts. Since they all have the same neighbors and weights, they are all connected to some vertex $u$ via the edge $\widehat{e}$ achieving the minimum weight $w_{\min}^T$. We construct a solution $T'$ for $(G',s,x',w')$ by replacing the leaves in $V_{i,j}$ with the single leaf $v_{\Sigma}$ connected to $u$. The weight of the edge $\{v_\Sigma,u\}$ in $T'$ is $|V_{i,j}|\cdot w_{\min}$, which is exactly the sum of the weights of the edges incident to $V_{i,j}$ in $T$. Hence, $w'(T')=w(T)$. Similarly, $x'(T')=x(T)+x'(v_\Sigma)-\sum_{v\in V_{i,j}}x(v)$, hence $x'(T')=x(T)$. It follows that $(G',s,x',w')$ is a yes-instance.

    $\Leftarrow$: Suppose that $(G',s,x',w')$ is a yes-instance, witnessed by any solution $T'$ of $(G',s,x',w')$. If $v_{\Sigma}\notin V(T')$, then $T'$ is also a valid solution for $(G,s,x,w)$. If $v_{\Sigma}\in V(T')$, let $\widehat{T}$ result from $T'$ by replacing $v_{\Sigma}$ by the vertices of $V_{i,j}$ and the original edges of $G$. Note that $x(\widehat{T})=x'(T')$ and $w(\widehat{T})=w'(T')$. If $\widehat{T}$ is not acyclic (i.e., if $\deg_{T'}(v_{\Sigma})>1$), then let $T$ be a spanning tree of $\widehat{T}$. Since the weights are nonnegative, then $w(T)\leq w(\widehat{T})$, hence $x(T)-w(T)\geq x(\widehat{T})-w(\widehat{T})=  x'(T')-w'(T')>0$. Thus, $(G,s,x,w)$ is a yes-instance.

\end{proof}
}%toappendix

\begin{apprestatable}{theorem}{thmsndkernelbound}\label{thm:snd_kernel_bound}
    \MSTGcore admits a kernel with $O(\operatorname{snd}(G)^2)$ vertices and $O(\operatorname{snd}(G)^4)$ edges.
\end{apprestatable}
\toappendix{
\sv{\thmsndkernelbound*}
\begin{proof}
    Let $(G,s,x,w)$ be the input instance and let $V_1,V_2,\ldots, V_{d}$ be a nice $w$-uniform partition of $V(G)$ of size at most $\operatorname{snd}(G)+1$. Recall that such partition is computable in polynomial time. We apply \Cref{rrule:snd_rule} exhaustively. Let $(G',s,x',w')$ be the resulting instance. Note that each application of the rule might increase the signed neighborhood diversity and also change the partition. More precisely, if $v'$ is a vertex created by merging some vertices $v_1,v_2,\ldots,v_{\ell}$ in $V_i$, then $(V_i\setminus \{v_1,v_2,\ldots,v_{\ell}\})\cup \{v'\}$ might not be a part of a $w$-uniform partition. However, each original part of the partition creates at most $d+1$ new vertices and hence $G'$ has at most $d(d+1)\leq (\operatorname{snd}(G)+1)(\operatorname{snd}(G)+2)=O(\operatorname{snd}(G)^2)$ vertices. Since $\operatorname{snd}(G')\leq |V(G')|$, it follows that $\operatorname{snd}(G')\leq O(\operatorname{snd}(G)^2)$. Hence $(G',s,x',w')$ is the desired kernel. Note that the number of edges is at most quadratic in any graph, thus $|E(G')|\leq |V(G')|^2\leq O(\operatorname{snd}(G)^4)$.
\end{proof}
}%toappendix

\begin{theorem}\label{thm:snd_kernel_ft}
    \MSTGcore admits a polynomial kernel parameterized by the signed neighborhood diversity of the underlying weighted graph.
\end{theorem}
\begin{proof}
    Let $(G,s,x,w)$ be the input instance. Let $(G',s,x',w')$ be the kernel from \Cref{thm:snd_kernel_bound}. By applying the algorithm from \Cref{lem:alg_reducing_weights}, we obtain an equivalent instance $(G',s,x'',w'')$ with total bit-size $O(\operatorname{snd}(G)^{16})$.
\end{proof}

\appsection{Feedback Edge Number}{sec:fen}
In this section we parameterize \MSTGcore by the feedback edge number (\fen) of the underlying graph $G$. The main result of this section is \Cref{thm:fen_kernel}.
\begin{restatable}{theorem}{thmfenkernel}\label{thm:fen_kernel}
    \MSTGcore admits a polynomial kernel parameterized by the feedback edge number of the input graph.
\end{restatable}

In order to prove \Cref{thm:fen_kernel} we design two reduction rules. \Cref{rrule:leaves} deals with degree $1$ vertices. Intuitively, whenever the neighbor of a leaf $\ell$ is in a solution, then it is reasonable to attach $\ell$ into the solution if and only if the weight of the edge connecting $\ell$ is cheaper than $x(\ell)$. \Cref{rrule:deg2} shortens long paths of degree-2 vertices, compressing the weight $w$ and values $x$ along the path into constantly many meaningful configurations. These reduction rules alone produce an instance with $O(\fen)$ vertices and edges~(\Cref{lem:bound_fen}). We then apply the framework of Frank and Tardos~\cite{FrankT87} to reduce the input weights and obtain a truly polynomial kernel parameterized by $\fen$, proving \Cref{thm:fen_kernel}. As a corollary of our construction we obtain a very simple linear-time algorithm for \MSTGcore when the input graph is a tree~(\Cref{thm:linear_trees}).

\begin{rrule}\label{rrule:leaves}
    Let $(G,s,x,w)$ be an instance of \MSTGcore and let $\ell$ be a leaf of $G$, $\ell \neq s$. Let $e=\{u,\ell\}$ be the unique edge incident to $\ell$. Output the instance $(G',s,x',w)$, where $G'=G\setminus \ell$ and the new allocation $x'$ is as follows. $x'(v)=x(v)$ for $v\neq u$, and $x'(u)=x(u)+\max\{x(\ell)-w(e),0\}$.
\end{rrule}

\begin{apprestatable}{lemma}{rruleonecorrectness}
    \Cref{rrule:leaves} is correct.
\end{apprestatable}

\toappendix{
\sv{\rruleonecorrectness*}
\begin{proof}
    We show the equivalence $(G,s,x,w)$ is a yes-instance of \MSTGcore if and only if $(G',s,x',w)$ is a yes-instance of \MSTGcore.
    
    $\Rightarrow$: Let $T$ be a solution for $(G,s,x,w)$ and moreover assume that $|V(T)|$ is minimum. We split the analysis into two cases based on whether $\ell \in V(T)$.
    \begin{description}
        \item[Case 1] $\ell \notin V(T):$
        \begin{description}
            \item[Case 1.1] $w(e)>x(\ell)$: In this case, $x'$ and $x$ agree on $V(T)$, thus $x'(T)>w(T)$.
            \item[Case 1.2] $w(e)\leq x(\ell)$: In this case, $T$  also witnesses $x'(T)>w(T)$, because $x'(T)=x(T)+x(\ell)-w(e)\geq x(T)>w(T)$.
        \end{description}
        \item[Case 2] $\ell \in V(T)$: Since $\ell \neq s$, we necessarily have $u\in V(T)$, because $T$ is connected and $u$ is the only neighbor of $\ell$. Let $T'=T\setminus \{\ell\}$.
        \begin{description}
            \item[Case 2.1] $w(e)>x(\ell)$: We show that $x(T')>w(T)$, contradicting the minimality of $|V(T)|$. To see this, note that $x(T')=x(T)-x(\ell)>w(T)-w(e)= w(T')$. Hence this case cannot happen.
            \item[Case 2.2] $w(e)\leq x(\ell)$: We show that $T'$ is a solution to $(G',s,x',w)$. We have $x'(T')=x(T)-x(\ell)>w(T)-w(e)=w(T')$.
        \end{description}
    \end{description}
    
    $\Leftarrow$: Let $T'$ be a solution for $(G's,x',w)$. If $u\notin V(T')$ or $x(\ell)<w(e)$, then $T'$ is also solution for $(G,s,x,w)$ because $x'$ agrees with $x$ on $V(T')$ in such case. The only case left to consider is $u\in V(T')$ and $w(e)\leq x(\ell)$. Let $T=T'+e$. Observe that $x(T)=x(T')+x(\ell)= (x(T')+x(\ell)-w(e))+w(e) = x'(T')+w(e)>w(T')+w(e)=w(T)$ and this finishes the proof.
\end{proof}
}%toappendix

\begin{restatable}{theorem}{thmlineartrees}\label{thm:linear_trees}
    \MSTGcore is solvable in linear time when the input graph is a tree.
\end{restatable}
\begin{proof}
    Apply \Cref{rrule:leaves} exhaustively until the one-vertex graph with $s$ remains and output \texttt{YES} if $x(s)>0$, otherwise output \texttt{NO}.
\end{proof}

\begin{rrule}\label{rrule:deg2}
    Let $(G,s,x,w)$ be an instance of \MSTGcore reduced with respect to \Cref{rrule:leaves}. Let $v_1,v_2,v_3,\ldots,v_k$ ($k\geq 3$) be a path avoiding $s$ consisting solely of vertices of degree $2$ in $G$. Let $v_0\in N_G(v_1)\setminus \{v_2\}$ be the unique neighbor of $v_1$ other than $v_2$ and similarly let $v_{k+1}\in N_G(v_k)\setminus \{v_{k-1}\}$ be the unique neighbor of $v_k$ other than $v_{k-1}$. Denote $e_i=\{v_{i-1},v_{i}\}$ for $i \geq 1$. For $i\in[k]$ compute the following prefix and suffix sums: $\prefix_i=\sum_{\ell=1}^i(x(v_{\ell})-w(e_{\ell}))$ and $\suffix_i=\sum_{\ell=i}^k(x(v_{\ell})-w(e_{\ell+1}))$ respectively. We define the four numbers referred to as \emph{total} $(\total)$, \emph{prefix} $(\prefix)$, \emph{suffix} $(\suffix)$, and \emph{combined} $(\combined)$:
   \begin{align*}
    \total = \sum_{\ell=1}^kx(v_{\ell})-\sum_{\ell=1}^{k+1}w(e_{\ell}), \quad  \prefix = \max_{i\in[k]}\{\prefix_i\}, \quad
    \suffix = \max_{i\in[k]}\{\suffix_i\}, \quad \combined = \max_{i<j}\{\prefix_i+\suffix_j\}.
\end{align*}

    Replace the vertices $v_1,\ldots,v_k$ by two new vertices $v_1',v_2'$ and add the new edges $e_1'=\{v_0,v_1'\},e_2'=\{v_1',v_2'\},e_3'=\{v_2',v_{k+1}\}$. Denote the new graph $G'$ and let the new weight function $w'$ and allocation $x'$ be as follows. All vertices and edges from the original instance have their values and weight unchanged. For the newly created vertices and edges set:
    \begin{align*}
        w'(e_1')=\suffix-\total, \quad 
        w'(e_2')=\combined-\total, \quad 
        w'(e_3')= \prefix-\total, \quad
    x'(v_1')=x'(v_2')=\frac{\prefix+\suffix+\combined-2\total}{2}.
    \end{align*}
    Output the instance $(G',s,x',w')$.
\end{rrule}

We refer the reader to \Cref{fig:reduction_rule_deg2,fig:reduction_rule_deg2_example} \sv{in the appendix} for an illustration of application of \Cref{rrule:deg2}.

\lv{We stress out the fact that it is possible in \Cref{rrule:deg2} that $s\in \{v_0,v_{k+1}\}$ or that $v_0=v_{k+1}$. In such case $v_0,v_1,\ldots,v_k,v_{k+1}$ is a cycle connected to the rest of $G$ only via the cut-vertex~$v_0$. On the other hand, note that $v_{k+1}\neq v_1$ (and $v_0\neq v_k$) as otherwise if $G$ is connected, then $G$ is isomorphic to a cycle, hence one of the vertices is $s$, contradicting the applicability of the reduction rule.}

\lv{We remark that the rule does not need a maximal such path in order to work. One may also reduce a long path piecewise, taking always 3 vertices and replacing it by 2. However, this may introduce larger fractions in the intermediate instances when applying the rule `piecewise'. On the other hand, applying it directly to the maximal paths, the maximum denominator in any reduced fraction in the resulting instance is at most~$2$.}

\toappendix{
\begin{figure}
    \centering
\begin{tikzpicture}[
	vertex/.style={circle,draw,inner sep=2pt},
	endpoint/.style={rectangle,draw,inner sep=3pt}
]%first scope
\begin{scope}[yshift=0cm]
%endpoints
\node[endpoint] (v0) at (0.000000,0.000000){};
\node[endpoint] (v5) at (9.000000,0.000000){};
%internal vertices
\node[vertex] (v1) at (1.80,0.00){};
\node[vertex] (v2) at (3.60,0.00){};
\node[vertex] (v3) at (5.40,0.00){};
\node[vertex] (v4) at (7.20,0.00){};
%edges
\draw (v0)--(v1);
\draw (v1)--(v2);
\draw (v2)--(v3);
\draw (v3)--(v4);
\draw (v4)--(v5);
%bottom edge labels (default is e_1,...,e_{k+1})
\node[below,font=\small] at (0.90,0.00) {$e_{1}$};
\node[below,font=\small] at (2.70,0.00) {$e_{2}$};
\node[below,font=\small] at (4.50,0.00) {$e_{3}$};
\node[below,font=\small] at (6.30,0.00) {$e_{4}$};
\node[below,font=\small] at (8.10,0.00) {$e_{5}$};
%top edge labels (default is empty)
\node[above,font=\small] at (0.90,0.00) {};
\node[above,font=\small] at (2.70,0.00) {};
\node[above,font=\small] at (4.50,0.00) {};
\node[above,font=\small] at (6.30,0.00) {};
\node[above,font=\small] at (8.10,0.00) {};
%bottom vertex labels (default is v0,...,v{k+1})
\node[below=2pt,font=\small] at (0.00,0.00) {$v_{0}$};
\node[below=2pt,font=\small] at (1.80,0.00) {$v_{1}$};
\node[below=2pt,font=\small] at (3.60,0.00) {$v_{2}$};
\node[below=2pt,font=\small] at (5.40,0.00) {$v_{3}$};
\node[below=2pt,font=\small] at (7.20,0.00) {$v_{4}$};
\node[below=2pt,font=\small] at (9.00,0.00) {$v_{5}$};
%top vertex labels (default is empty)
\node[above=2pt,font=\small] at (0.00,0.00) {};
\node[above=2pt,font=\small] at (1.80,0.00) {};
\node[above=2pt,font=\small] at (3.60,0.00) {};
\node[above=2pt,font=\small] at (5.40,0.00) {};
\node[above=2pt,font=\small] at (7.20,0.00) {};
\node[above=2pt,font=\small] at (9.00,0.00) {};
\end{scope}
\draw[->, thick] (4.500000,-1.000000) to (4.500000,-2.000000);
%second scope
\begin{scope}[yshift=-3cm]
%endpoints
\node[endpoint] (v0) at (0.000000,0.000000){};
\node[endpoint] (v3) at (9.000000,0.000000){};
%internal vertices
\node[vertex] (v1) at (3.00,0.00){};
\node[vertex] (v2) at (6.00,0.00){};
%edges
\draw (v0)--(v1);
\draw (v1)--(v2);
\draw (v2)--(v3);
%bottom edge labels (default is e_1,...,e_{k+1})
\node[below,font=\small] at (1.50,0.00) {$e'_{1}$};
\node[below,font=\small] at (4.50,0.00) {$e'_{2}$};
\node[below,font=\small] at (7.50,0.00) {$e'_{3}$};
%top edge labels (default is empty)
\node[above,font=\small] at (1.50,0.00) {\textcolor{red}{$\suffix-\total$}};
\node[above,font=\small] at (4.50,0.00) {\textcolor{red}{$\combined-\total$}};
\node[above,font=\small] at (7.50,0.00) {\textcolor{red}{$\prefix-\total$}};
%bottom vertex labels (default is v0,...,v{k+1})
\node[below=2pt,font=\small] at (0.00,0.00) {$v_{0}$};
\node[below=2pt,font=\small] at (3.00,0.00) {$v'_{1}$};
\node[below=2pt,font=\small] at (6.00,0.00) {$v'_{2}$};
\node[below=2pt,font=\small] at (9.00,0.00) {$v_{5}$};
%top vertex labels (default is empty)
\node[above=2pt,font=\small] at (0.00,0.00) {};
\node[above=2pt,font=\small] at (3.00,0.00) {\textcolor{blue}{$\frac{\prefix+\suffix+\combined-2\total}{2}$}};
\node[above=2pt,font=\small] at (6.00,0.00) {\textcolor{blue}{$\frac{\prefix+\suffix+\combined-2\total}{2}$}};
\node[above=2pt,font=\small] at (9.00,0.00) {};
\end{scope}
\end{tikzpicture}

        \caption{Visualization of application of \Cref{rrule:deg2} for $k=4$. On the top is (a part of) the original instance $(G,s,x,w)$, while on the bottom is the reduced one $(G',s,x',w')$. The circular vertices are the vertices of degree $2$ replaced by the reduction rule. The square vertices are untouched by the reduction rule. The new edge \textcolor{red}{weights} and vertex \textcolor{blue}{values} are displayed above the edges (resp. vertices) in the reduced instance.
        }
        \label{fig:reduction_rule_deg2}
\end{figure}
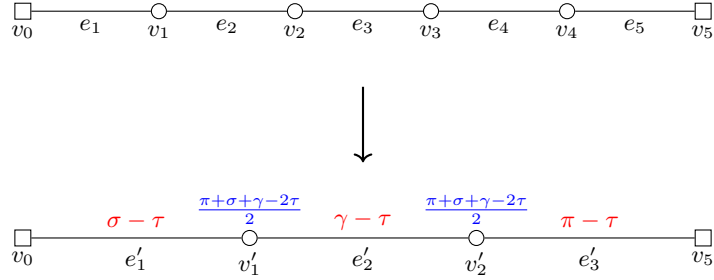
}%toappendix
%example_2
\toappendix{
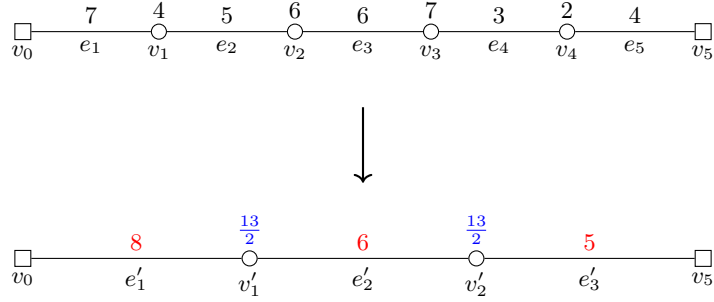
\begin{figure}
        \centering
        \begin{tikzpicture}[
	vertex/.style={circle,draw,inner sep=2pt},
	endpoint/.style={rectangle,draw,inner sep=3pt}
]%first scope
\begin{scope}[yshift=0cm]
%endpoints
\node[endpoint] (v0) at (0.000000,0.000000){};
\node[endpoint] (v5) at (9.000000,0.000000){};
%internal vertices
\node[vertex] (v1) at (1.80,0.00){};
\node[vertex] (v2) at (3.60,0.00){};
\node[vertex] (v3) at (5.40,0.00){};
\node[vertex] (v4) at (7.20,0.00){};
%edges
\draw (v0)--(v1);
\draw (v1)--(v2);
\draw (v2)--(v3);
\draw (v3)--(v4);
\draw (v4)--(v5);
%bottom edge labels (default is e_1,...,e_{k+1})
\node[below,font=\small] at (0.90,0.00) {$e_{1}$};
\node[below,font=\small] at (2.70,0.00) {$e_{2}$};
\node[below,font=\small] at (4.50,0.00) {$e_{3}$};
\node[below,font=\small] at (6.30,0.00) {$e_{4}$};
\node[below,font=\small] at (8.10,0.00) {$e_{5}$};
%top edge labels (default is empty)
\node[above,font=\small] at (0.90,0.00) {$7$};
\node[above,font=\small] at (2.70,0.00) {$5$};
\node[above,font=\small] at (4.50,0.00) {$6$};
\node[above,font=\small] at (6.30,0.00) {$3$};
\node[above,font=\small] at (8.10,0.00) {$4$};
%bottom vertex labels (default is v0,...,v{k+1})
\node[below=2pt,font=\small] at (0.00,0.00) {$v_{0}$};
\node[below=2pt,font=\small] at (1.80,0.00) {$v_{1}$};
\node[below=2pt,font=\small] at (3.60,0.00) {$v_{2}$};
\node[below=2pt,font=\small] at (5.40,0.00) {$v_{3}$};
\node[below=2pt,font=\small] at (7.20,0.00) {$v_{4}$};
\node[below=2pt,font=\small] at (9.00,0.00) {$v_{5}$};
%top vertex labels (default is empty)
\node[above=2pt,font=\small] at (0.00,0.00) {};
\node[above=2pt,font=\small] at (1.80,0.00) {$4$};
\node[above=2pt,font=\small] at (3.60,0.00) {$6$};
\node[above=2pt,font=\small] at (5.40,0.00) {$7$};
\node[above=2pt,font=\small] at (7.20,0.00) {$2$};
\node[above=2pt,font=\small] at (9.00,0.00) {};
\end{scope}
\draw[->, thick] (4.500000,-1.000000) to (4.500000,-2.000000);
%second scope
\begin{scope}[yshift=-3cm]
%endpoints
\node[endpoint] (v0) at (0.000000,0.000000){};
\node[endpoint] (v3) at (9.000000,0.000000){};
%internal vertices
\node[vertex] (v1) at (3.00,0.00){};
\node[vertex] (v2) at (6.00,0.00){};
%edges
\draw (v0)--(v1);
\draw (v1)--(v2);
\draw (v2)--(v3);
%bottom edge labels (default is e_1,...,e_{k+1})
\node[below,font=\small] at (1.50,0.00) {$e'_{1}$};
\node[below,font=\small] at (4.50,0.00) {$e'_{2}$};
\node[below,font=\small] at (7.50,0.00) {$e'_{3}$};
%top edge labels (default is empty)
\node[above,font=\small] at (1.50,0.00) {\textcolor{red}{$8$}};
\node[above,font=\small] at (4.50,0.00) {\textcolor{red}{$6$}};
\node[above,font=\small] at (7.50,0.00) {\textcolor{red}{$5$}};
%bottom vertex labels (default is v0,...,v{k+1})
\node[below=2pt,font=\small] at (0.00,0.00) {$v_{0}$};
\node[below=2pt,font=\small] at (3.00,0.00) {$v'_{1}$};
\node[below=2pt,font=\small] at (6.00,0.00) {$v'_{2}$};
\node[below=2pt,font=\small] at (9.00,0.00) {$v_{5}$};
%top vertex labels (default is empty)
\node[above=2pt,font=\small] at (0.00,0.00) {};
\node[above=2pt,font=\small] at (3.00,0.00) {\textcolor{blue}{$\frac{13}{2}$}};
\node[above=2pt,font=\small] at (6.00,0.00) {\textcolor{blue}{$\frac{13}{2}$}};
\node[above=2pt,font=\small] at (9.00,0.00) {};
\end{scope}
\end{tikzpicture}

        \caption{Example of application of  \Cref{rrule:deg2} for $k=4$ (same as in \Cref{fig:reduction_rule_deg2}) with specific numbers. The weights $w$ and allocation $x$ are displayed above the edges and vertices. The prefixes of the original path are $(\prefix_{1},\prefix_{2},\prefix_{3},\prefix_{4}) = (-3,-2,-1,-2)$, and the suffixes are $(\suffix_{1},\suffix_{2},\suffix_{3},\suffix_{4}) = (1,2,2,-2)$, the combined number is $\combined = \prefix_2+\suffix_3=0$ and the total number is $\total = (4+6+7+2)-(7+5+6+3+4)=-6$. By plugging into the equations in \Cref{rrule:deg2} we obtain  $x'(v_1')=x'(v_2')=\frac{\prefix+\suffix+\combined+2\total}{2}=\frac{-1+2+0-2\cdot(-6)}{2}=\textcolor{blue}{\frac{13}{2}}$, $w'(e_1')=\suffix-\total=2-(-6)=\textcolor{red}{8}$, $w'(e_2')=\combined-\total=0-(-6)=\textcolor{red}{6}$, and $w'(e_3')=\prefix-\total=-1-(-6)=\textcolor{red}{5}$.}
        \label{fig:reduction_rule_deg2_example}
\end{figure}
}%toappendix

\begin{apprestatable}{lemma}{lemrruletwocor}\label{lem:rrule2_cor}
 \Cref{rrule:deg2} is correct.
\end{apprestatable}
\sv{
\begin{proof}[Proof Sketch]
    The main idea is that the long path will be used in at most $4$ distinct structural ways (see \Cref{fig:forward_proof_visualization}). In the original instance $(G,s,x,w)$, either the solution will take $v_0$ and some part from left, or $v_k$ and some part from the right, or it will take the whole path, or it will take both $v_0$ and $v_k$ and some part from left and some part from right, but not the whole path. The four quantities $\prefix,\suffix,\total,\combined$ correspond to the optimal choices and the values and weights in the reduced instance are set in such a way to reflect these quantities when taking the corresponding compressed solution.
    Full details and analysis of all cases can be found in the appendix.
\end{proof}

\begin{figure}[bt]
\centering
\begin{tikzpicture}[
    circle node/.style={draw, minimum size=2mm, inner sep=0pt,circle},
    square node/.style={draw, minimum size=2mm, inner sep=0pt,rectangle},
    x=1cm, y=0.6cm
]

\tikzset{
  redsnake/.style={
    red,
    thick,
    decorate,
    decoration={snake, amplitude=2pt, segment length=10pt}
  }
}

\node(g) at (2.5,1.1){$(G,s,x,w)$};
\node(gp) at (1.5+8,1.1){$(G',s,x',w')$};

%labels, first row
\foreach \i in {0,...,5}{
    \node(v\i) at (\i,0.5){$v_\i$};    
}
\foreach \i in {1,...,2}
{
    \node(vp\i) at (\i+8,0.5){$v_\i'$};
}
\node(vp0) at (0+8,0.5){$v_0$};
\node(vp3) at (3+8,0.5){$v_5$};

%cycle over rows, draws just plain paths (black)
\foreach \row in {0,1,2,3} {

    % LEFT COLUMN: paths with 6 vertices
    \foreach \i in {0,...,5} {
        \ifnum\i=0
            \node[square node] (L\row-\i) at (\i, -\row) {};
        \else\ifnum\i=5
            \node[square node] (L\row-\i) at (\i, -\row) {};
        \else
            \node[circle node] (L\row-\i) at (\i, -\row) {};
        \fi\fi
    }

    % edges for left paths
    \foreach \i in {0,...,4} {
        \draw (L\row-\i) -- (L\row-\the\numexpr\i+1\relax);
    }

    % RIGHT COLUMN: paths with 4 vertices
    \foreach \i in {0,...,3} {
        \ifnum\i=0
            \node[square node] (R\row-\i) at (\i+8, -\row) {};
        \else\ifnum\i=3
            \node[square node] (R\row-\i) at (\i+8, -\row) {};
        \else
            \node[circle node] (R\row-\i) at (\i+8, -\row) {};
        \fi\fi
    }

    % edges for right paths
    \foreach \i in {0,...,2} {
        \draw (R\row-\i) -- (R\row-\the\numexpr\i+1\relax);
    }
}

%case1
\foreach \i in {0,...,2}{
    \draw[red,thick, ] (L0-\i) -- (L0-\the\numexpr\i+1\relax);
}

\foreach \i in {0,...,1}{
    \draw[red,thick, ] (R0-\i) -- (R0-\the\numexpr\i+1\relax);
}

\foreach \i in {0,...,3}{
    \ifnum\i=0
        \node[square node,red,fill] (L0r-\i) at (\i,0){};
    \else
        \node[circle node,red,fill] (L0r-\i) at (\i,0){};
    \fi
}

\foreach \i in {0,...,2} {
        \ifnum\i=0
            \node[square node,red,fill] (R0r-\i) at (\i+8, 0) {};
        \else\ifnum\i=3
            \node[square node,red,fill] (R0r-\i) at (\i+8, 0) {};
        \else
            \node[circle node,red,fill] (R0r-\i) at (\i+8, 0) {};
        \fi\fi
    }

%case2
\foreach \i in {2,...,4}{
    \draw[red,thick, ] (L1-\i) -- (L1-\the\numexpr\i+1\relax);
}
\foreach \i in {1,...,2}{
    \draw[red,thick, ] (R1-\i) -- (R1-\the\numexpr\i+1\relax);
}

\foreach \i in {2,...,5}{
    \ifnum\i=5
        \node[square node,red,fill] (L1r-\i) at (\i,-1){};
    \else
        \node[circle node,red,fill] (L1r-\i) at (\i,-1){};
    \fi
}

\foreach \i in {1,...,3} {
        \ifnum\i=0
            \node[square node,red,fill] (R1r-\i) at (\i+8, -1) {};
        \else\ifnum\i=3
            \node[square node,red,fill] (R1r-\i) at (\i+8, -1) {};
        \else
            \node[circle node,red,fill] (R1r-\i) at (\i+8, -1) {};
        \fi\fi
    }

%case3
\foreach \i in {0,...,4}{
    \draw[red,thick, ] (L2-\i) -- (L2-\the\numexpr\i+1\relax);
}
\foreach \i in {0,...,2}{
    \draw[red,thick, ] (R2-\i) -- (R2-\the\numexpr\i+1\relax);
}

\foreach \i in {0,...,5}{
    \ifnum\i=5
        \node[square node,red,fill] (L2r-\i) at (\i,-2){};
    \else\ifnum\i=0
        \node[square node,red,fill] (L2r-\i) at (\i,-2){};
    \else
        \node[circle node,red,fill] (L2r-\i) at (\i,-2){};
    \fi\fi
}

\foreach \i in {0,...,3} {
        \ifnum\i=0
            \node[square node,red,fill] (R2r-\i) at (\i+8, -2) {};
        \else\ifnum\i=3
            \node[square node,red,fill] (R2r-\i) at (\i+8, -2) {};
        \else
            \node[circle node,red,fill] (R2r-\i) at (\i+8, -2) {};
        \fi\fi
    }

%case4
\foreach \i in {0,...,1}{
    \draw[red,thick] (L3-\i) -- (L3-\the\numexpr\i+1\relax);
}
\foreach \i in {3,...,4}{
    \draw[red,thick] (L3-\i) -- (L3-\the\numexpr\i+1\relax);
}

\draw[red,thick] (R3-0) -- (R3-1);
\draw[red,thick] (R3-2) -- (R3-3);

\foreach \i in {0,...,2}{
    \ifnum\i=0
        \node[square node,red,fill] (L3r-\i) at (\i,-3){};
    \else
        \node[circle node,red,fill] (L3r-\i) at (\i,-3){};
    \fi
}

\foreach \i in {3,...,5}{
    \ifnum\i=5
        \node[square node,red,fill] (L3r-\i) at (\i,-3){};
    \else
        \node[circle node,red,fill] (L3r-\i) at (\i,-3){};
    \fi
}

\foreach \i in {0,...,1} {
        \ifnum\i=0
            \node[square node,red,fill] (R3r-\i) at (\i+8, -3) {};
        \else\ifnum\i=3
            \node[square node,red,fill] (R3r-\i) at (\i+8, -3) {};
        \else
            \node[circle node,red,fill] (R3r-\i) at (\i+8, -3) {};
        \fi\fi
    }

\foreach \i in {2,...,3} {
        \ifnum\i=0
            \node[square node,red,fill] (R3r-\i) at (\i+8, -3) {};
        \else\ifnum\i=3
            \node[square node,red,fill] (R3r-\i) at (\i+8, -3) {};
        \else
            \node[circle node,red,fill] (R3r-\i) at (\i+8, -3) {};
        \fi\fi
    }

\end{tikzpicture}

\caption{The figure represents the sketch of the main idea of the correctness of \Cref{rrule:deg2} for $k=4$. Each line represents one possible behaviour of a solution (in red). The left part depicts the four possible cases in the original instance $(G,s,x,w)$ and the right part depicts the corresponding situation in the reduced instance $(G',s,x',w')$. The case visualized corresponds to the best prefix $\prefix$ being attained for $\prefix_3$, the best suffix ($\suffix$) for $\suffix_2$ and the best combined ($\combined$) for $\prefix_2+\sigma_3$.}
\label{fig:forward_proof_visualization}
\end{figure}

}%

\toappendix{
\sv{\lemrruletwocor*}
The proof of \Cref{lem:rrule2_cor} is covered by three claims~(\Cref{claim:batmostpps,claim:rrule2_nonnegative,claim:rr2_ineq}) and two lemmas (\Cref{lem:rrule2_cor_forward,lem:rrule2_cor_backward}).  We first observe several numerical properties of the values $\prefix,\suffix,\total,\combined$ in \Cref{claim:batmostpps,claim:rrule2_nonnegative,claim:rr2_ineq} and we then prove the safeness of the reduction in \Cref{lem:rrule2_cor_forward,lem:rrule2_cor_backward}. For the sake of readability and compactness, denote just for this proof $x_i:=x'(v_i')$ and $w_i:=w'(e_i')$.
\begin{claim}\label{claim:batmostpps}
    $\combined\leq \prefix+\suffix$.
\end{claim}
\begin{proof}
    Notice that $
    \combined=\max_{i<j}\{\prefix_i+\suffix_j\}\leq \max_{i,j \in [k]}\{\prefix_i+\suffix_j\}=\max_{i \in [k]}\{\prefix_i\}+\max_{j \in [k]}\{\suffix_j\}=\prefix+\suffix$.
\end{proof}

\begin{claim}\label{claim:rrule2_nonnegative}
    If in the original instance all the weights $w(e_{i})$ and values $x(v_{i})$ were non-negative, then the new weights $w_i$ and values $x_i$ in the reduced instance remain non-negative.
\end{claim}
\begin{proof}
    Observe that $\total=\prefix_k-w(e_{k+1})\leq \prefix_k\leq \max_{i\in[k]}\{\prefix_i\} = \prefix$, thus $\prefix-\total\geq 0$. Similarly, $\total=\suffix_1-w(e_1)\leq \suffix_1 \leq \suffix$, thus $\suffix-\total\geq 0$. Analogously, $\total=\prefix_1+\suffix_2-w(e_2)\leq \prefix_1+\suffix_2\leq \combined$, thus $\combined-\total\geq 0$. For the inequality $x(v_1')\geq 0$, notice that $\combined\leq \prefix+\suffix$ by \Cref{claim:batmostpps}, so $\frac{\prefix+\suffix+\combined-2\total}{2}\geq \frac{2\combined-2\total}{2}=\combined-\total\geq 0$.
\end{proof}

\begin{claim}\label{claim:rr2_ineq}
    The values $x_1,x_2,w_1,w_2,w_3$ satisfy the following six (in)equalities:
    \begin{align}
        x_1+x_2-w_1-w_3&=\combined \\
        x_1+x_2-w_1-w_2&=\prefix \\
        x_1+x_2-w_2-w_3&=\suffix \\
        x_1+x_2-w_1-w_2-w_3&=\total\\
        w_2&\leq x_1 \\
        w_2&\leq x_2
    \end{align}
\end{claim}
\begin{proof}
    The first four equalities are straightforward to check. For the inequalities, recall that $x_1=x_2$ and we have:
    \begin{align*}
        w_2&\leq x_1 &\Leftrightarrow \\
        \combined-\total&\leq \frac{\prefix+\suffix+\combined-2\total}{2} &\Leftrightarrow\\
        \combined&\leq \prefix + \suffix.&
    \end{align*}
    The last inequality is true by~\Cref{claim:batmostpps} and this finishes the proof.
\end{proof}

\begin{lemma}\label{lem:rrule2_cor_forward}
    Let $(G,s,x,w)$ be an instance of \MSTGcore and let $(G',s,x',w')$ result from $(G,s,x,w)$ by application of \Cref{rrule:deg2}. If $(G,s,x,w)$ is a yes-instance, then $(G',s,x',w')$ is a yes-instance.
\end{lemma}

\begin{proof}
    Let $T$ be a solution for $(G,s,x,w)$. If $T$ avoids $v_1,\ldots,v_k$, then it is clearly a solution for $(G's,x',w')$. Suppose that $V(T)\cap \{v_1,\ldots v_k\}\neq \emptyset$.
    Since $T$ is connected and  $s\notin \{v_1,\ldots,v_k\}$, there are four cases to consider based on the shape of the intersection $E(T)\cap \{e_1,\ldots,e_{k+1}\}$:
    \begin{enumerate}[label=\arabic*.]
        \item $E(T)\cap \{e_1,\ldots,e_{k+1}\} = \{e_1,e_2,\ldots e_i\} \text{ for some $i\in[k]$}$
        \item $E(T)\cap \{e_1,\ldots,e_{k+1}\} = \{e_{i+1},\ldots,e_{k+1}\} \text{ for some $i\in[k]$}$
        \item $E(T)\cap \{e_1,\ldots,e_{k+1}\} = \{e_1,\ldots,e_{k+1}\}$
        \item $E(T)\cap \{e_1,\ldots, e_{k+1}\}=\{e_1,e_2,\ldots,e_i\}\cup \{e_{j+1},\ldots,e_{k+1}\} \text{ for some $i<j$ and $i,j \in [k]$}$
    \end{enumerate}

    We describe in each case how to construct a solution $T'$ for $(G',s,x',w')$.

    \begin{description}
        \item[Case 1] $E(T)\cap \{e_1,\ldots,e_{k+1}\} = \{e_1,e_2,\ldots e_i\} \text{ for some $i\in[k]$}$: Let $T':=T|_{G'}+ \{e_1',e_2'\}$. See \Cref{fig:reduction_rule_deg2_case1} for illustration. We have:
        \begin{align*}
        x'(T')-w'(T')&=\left(x(T)-\sum_{\ell=1}^ix(v_{\ell})+x_1+x_2\right)-\left(w(T)-\sum_{\ell=1}^iw(e_{\ell})+w_1+w_2\right)= \\
        &= x(T)-w(T)-\prefix_i+x_1+x_2-w_1-w_2=\\&=x(T)-w(T)-\prefix_i+\prefix>0
        \end{align*}
        The last equality is due to \Cref{claim:rr2_ineq}. Note that $x(T)-w(T)>0$ because $T$ is a solution and $\prefix-\prefix_i\geq 0$ by the definition of $\prefix$. \\ 
%figure for case 1--similar the others.
\begin{figure}[H]
    \centering
    \begin{tikzpicture}[
        vertex/.style={circle,draw,inner sep=2pt},
        endpoint/.style={rectangle,draw,inner sep=3pt}
]%first scope
\begin{scope}[yshift=0cm]
%endpoints
\node[fill=red,endpoint] (v0) at (0.000000,0.000000){};
\node[endpoint] (v5) at (9.000000,0.000000){};
%internal vertices
\node[fill=red,vertex] (v1) at (1.80,0.00){};
\node[fill=red,vertex] (v2) at (3.60,0.00){};
\node[vertex] (v3) at (5.40,0.00){};
\node[vertex] (v4) at (7.20,0.00){};
%edges
\draw[red] (v0)--(v1);
\draw[red] (v1)--(v2);
\draw (v2)--(v3);
\draw (v3)--(v4);
\draw (v4)--(v5);
%bottom edge labels (default is e_1,...,e_{k+1})
\node[below,font=\small] at (0.90,0.00) {$e_{1}$};
\node[below,font=\small] at (2.70,0.00) {$e_{2}$};
\node[below,font=\small] at (4.50,0.00) {$e_{3}$};
\node[below,font=\small] at (6.30,0.00) {$e_{4}$};
\node[below,font=\small] at (8.10,0.00) {$e_{5}$};
%top edge labels (default is empty)
\node[above,font=\small] at (0.90,0.00) {};
\node[above,font=\small] at (2.70,0.00) {};
\node[above,font=\small] at (4.50,0.00) {};
\node[above,font=\small] at (6.30,0.00) {};
\node[above,font=\small] at (8.10,0.00) {};
%bottom vertex labels (default is v0,...,v{k+1})
\node[below=2pt,font=\small] at (0.00,0.00) {$v_{0}$};
\node[below=2pt,font=\small] at (1.80,0.00) {$v_{1}$};
\node[below=2pt,font=\small] at (3.60,0.00) {$v_{2}$};
\node[below=2pt,font=\small] at (5.40,0.00) {$v_{3}$};
\node[below=2pt,font=\small] at (7.20,0.00) {$v_{4}$};
\node[below=2pt,font=\small] at (9.00,0.00) {$v_{5}$};
%top vertex labels (default is empty)
\node[above=2pt,font=\small] at (0.00,0.00) {};
\node[above=2pt,font=\small] at (1.80,0.00) {};
\node[above=2pt,font=\small] at (3.60,0.00) {};
\node[above=2pt,font=\small] at (5.40,0.00) {};
\node[above=2pt,font=\small] at (7.20,0.00) {};
\node[above=2pt,font=\small] at (9.00,0.00) {};
\end{scope}
\draw[->, thick] (4.500000,-1.000000) to (4.500000,-2.000000);
%second scope
\begin{scope}[yshift=-3cm]
%endpoints
\node[fill=red,endpoint] (v0) at (0.000000,0.000000){};
\node[endpoint] (v3) at (9.000000,0.000000){};
%internal vertices
\node[fill=red,vertex] (v1) at (3.00,0.00){};
\node[fill=red,vertex] (v2) at (6.00,0.00){};
%edges
\draw[red] (v0)--(v1);
\draw[red] (v1)--(v2);
\draw (v2)--(v3);
%bottom edge labels (default is e_1,...,e_{k+1})
\node[below,font=\small] at (1.50,0.00) {$e'_{1}$};
\node[below,font=\small] at (4.50,0.00) {$e'_{2}$};
\node[below,font=\small] at (7.50,0.00) {$e'_{3}$};
%top edge labels (default is empty)
\node[above,font=\small] at (1.50,0.00) {};
\node[above,font=\small] at (4.50,0.00) {};
\node[above,font=\small] at (7.50,0.00) {};
%bottom vertex labels (default is v0,...,v{k+1})
\node[below=2pt,font=\small] at (0.00,0.00) {$v_{0}$};
\node[below=2pt,font=\small] at (3.00,0.00) {$v'_{1}$};
\node[below=2pt,font=\small] at (6.00,0.00) {$v'_{2}$};
\node[below=2pt,font=\small] at (9.00,0.00) {$v_{5}$};
%top vertex labels (default is empty)
\node[above=2pt,font=\small] at (0.00,0.00) {};
\node[above=2pt,font=\small] at (3.00,0.00) {};
\node[above=2pt,font=\small] at (6.00,0.00) {};
\node[above=2pt,font=\small] at (9.00,0.00) {};
\end{scope}
\end{tikzpicture}
\caption{Visualization of Case 1 from the proof of \Cref{lem:rrule2_cor_forward} for $i=2$. In the reduced instance $(G',s,x',w')$ (bottom), the solution uses the vertices $v_1',v_2'$ and the edges $e_1', e_2'$. The vertices and edges highlighted in red are touched by the considered solutions. The vertex $v_5$ may or may not be in the solution.}
        \label{fig:reduction_rule_deg2_case1}
\end{figure}

        \item[Case 2] $E(T)\cap \{e_1,\ldots,e_{k+1}\} = \{e_{i+1},\ldots,e_{k+1}\} \text{ for some $i\in[k]$}$: Let $T':=T|_{G'}+ \{e_2',e_3'\}$. See \Cref{fig:reduction_rule_deg2_case2} for illustration. We have:

        \begin{align*}
            x'(T')-w'(T')&=\left(x(T)-\sum_{\ell=i}^kx(v_{\ell})+x_1+x_2\right) - \left(w(T)-\sum_{\ell=i}^kw(e_{\ell+1})+w_2+w_3\right)=\\&=x(T)-w(T)-\suffix_i+x_1+x_2-w_2-w_3=\\&=x(T)-w(T)-\suffix_i+\suffix >0
        \end{align*}
        Note that $x(T)-w(T)>0$ because $T$ is a solution and $\suffix-\suffix_i\geq 0$ by the definition of $\suffix$.\\

%Figure for case 2
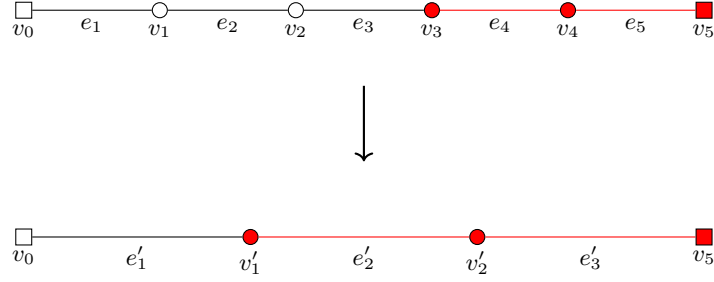
\begin{figure}[H]
    \centering
  \begin{tikzpicture}[
        vertex/.style={circle,draw,inner sep=2pt},
        endpoint/.style={rectangle,draw,inner sep=3pt}
]%first scope
\begin{scope}[yshift=0cm]
%endpoints
\node[endpoint] (v0) at (0.000000,0.000000){};
\node[fill=red,endpoint] (v5) at (9.000000,0.000000){};
%internal vertices
\node[vertex] (v1) at (1.80,0.00){};
\node[vertex] (v2) at (3.60,0.00){};
\node[fill=red,vertex] (v3) at (5.40,0.00){};
\node[fill=red,vertex] (v4) at (7.20,0.00){};
%edges
\draw (v0)--(v1);
\draw (v1)--(v2);
\draw (v2)--(v3);
\draw[red] (v3)--(v4);
\draw[red] (v4)--(v5);
%bottom edge labels (default is e_1,...,e_{k+1})
\node[below,font=\small] at (0.90,0.00) {$e_{1}$};
\node[below,font=\small] at (2.70,0.00) {$e_{2}$};
\node[below,font=\small] at (4.50,0.00) {$e_{3}$};
\node[below,font=\small] at (6.30,0.00) {$e_{4}$};
\node[below,font=\small] at (8.10,0.00) {$e_{5}$};
%top edge labels (default is empty)
\node[above,font=\small] at (0.90,0.00) {};
\node[above,font=\small] at (2.70,0.00) {};
\node[above,font=\small] at (4.50,0.00) {};
\node[above,font=\small] at (6.30,0.00) {};
\node[above,font=\small] at (8.10,0.00) {};
%bottom vertex labels (default is v0,...,v{k+1})
\node[below=2pt,font=\small] at (0.00,0.00) {$v_{0}$};
\node[below=2pt,font=\small] at (1.80,0.00) {$v_{1}$};
\node[below=2pt,font=\small] at (3.60,0.00) {$v_{2}$};
\node[below=2pt,font=\small] at (5.40,0.00) {$v_{3}$};
\node[below=2pt,font=\small] at (7.20,0.00) {$v_{4}$};
\node[below=2pt,font=\small] at (9.00,0.00) {$v_{5}$};
%top vertex labels (default is empty)
\node[above=2pt,font=\small] at (0.00,0.00) {};
\node[above=2pt,font=\small] at (1.80,0.00) {};
\node[above=2pt,font=\small] at (3.60,0.00) {};
\node[above=2pt,font=\small] at (5.40,0.00) {};
\node[above=2pt,font=\small] at (7.20,0.00) {};
\node[above=2pt,font=\small] at (9.00,0.00) {};
\end{scope}
\draw[->, thick] (4.500000,-1.000000) to (4.500000,-2.000000);
%second scope
\begin{scope}[yshift=-3cm]
%endpoints
\node[endpoint] (v0) at (0.000000,0.000000){};
\node[fill=red,endpoint] (v3) at (9.000000,0.000000){};
%internal vertices
\node[fill=red,vertex] (v1) at (3.00,0.00){};
\node[fill=red,vertex] (v2) at (6.00,0.00){};
%edges
\draw (v0)--(v1);
\draw[red] (v1)--(v2);
\draw[red] (v2)--(v3);
%bottom edge labels (default is e_1,...,e_{k+1})
\node[below,font=\small] at (1.50,0.00) {$e'_{1}$};
\node[below,font=\small] at (4.50,0.00) {$e'_{2}$};
\node[below,font=\small] at (7.50,0.00) {$e'_{3}$};
%top edge labels (default is empty)
\node[above,font=\small] at (1.50,0.00) {};
\node[above,font=\small] at (4.50,0.00) {};
\node[above,font=\small] at (7.50,0.00) {};
%bottom vertex labels (default is v0,...,v{k+1})
\node[below=2pt,font=\small] at (0.00,0.00) {$v_{0}$};
\node[below=2pt,font=\small] at (3.00,0.00) {$v'_{1}$};
\node[below=2pt,font=\small] at (6.00,0.00) {$v'_{2}$};
\node[below=2pt,font=\small] at (9.00,0.00) {$v_{5}$};
%top vertex labels (default is empty)
\node[above=2pt,font=\small] at (0.00,0.00) {};
\node[above=2pt,font=\small] at (3.00,0.00) {};
\node[above=2pt,font=\small] at (6.00,0.00) {};
\node[above=2pt,font=\small] at (9.00,0.00) {};
\end{scope}
\end{tikzpicture}
\caption{Visualization of Case 2 from the proof of \Cref{lem:rrule2_cor_forward} for $i=2$. Note that for the suffixes the indices are shifted by $1$. In the reduced instance $(G',s,x',w')$ (bottom), the solution uses the vertices $v_1',v_2'$ and the edges $e_2', e_3'$. The vertices and edges highlighted in red are touched by the considered solutions. The vertex $v_0$ may or may not be in the solution.}
\label{fig:reduction_rule_deg2_case2}
\end{figure}

        \item[Case 3]  $E(T)\cap \{e_1,\ldots,e_{k+1}\} = \{e_1,\ldots,e_{k+1}\}$: Let $T':=T|_{G'}+\{e_1',e_2',e_3'\}$. See \Cref{fig:reduction_rule_deg2_case3} for illustration. We have:
        \begin{align*}
            x'(T')-w'(T')&=\left(x(T)-\sum_{\ell=1}^kx(v_{\ell})+x_1+x_2\right)-\left(w(T)-\sum_{\ell=1}^{k+1}w(e_{\ell})+w_1+w_2+w_3\right)=\\&=x(T)-w(T)-\total+x_1+x_2-w_1-w_2-w_3 >0
        \end{align*}
        Note that $x(T)-w(T)>0$ because $T$ is a solution and $-\total+x_1+x_2-w_1-w_2-w_3 = 0$ by \Cref{claim:rr2_ineq}. \\

%figure 3 entire path 
\begin{figure}[H]
    \centering
    \begin{tikzpicture}[
        vertex/.style={circle,draw,inner sep=2pt},
        endpoint/.style={rectangle,draw,inner sep=3pt}
]%first scope
\begin{scope}[yshift=0cm]
%endpoints
\node[fill=red,endpoint] (v0) at (0.000000,0.000000){};
\node[fill=red,endpoint] (v5) at (9.000000,0.000000){};
%internal vertices
\node[fill=red,vertex] (v1) at (1.80,0.00){};
\node[fill=red,vertex] (v2) at (3.60,0.00){};
\node[fill=red,vertex] (v3) at (5.40,0.00){};
\node[fill=red,vertex] (v4) at (7.20,0.00){};
%edges
\draw[red] (v0)--(v1);
\draw[red] (v1)--(v2);
\draw[red] (v2)--(v3);
\draw[red] (v3)--(v4);
\draw[red] (v4)--(v5);
%bottom edge labels (default is e_1,...,e_{k+1})
\node[below,font=\small] at (0.90,0.00) {$e_{1}$};
\node[below,font=\small] at (2.70,0.00) {$e_{2}$};
\node[below,font=\small] at (4.50,0.00) {$e_{3}$};
\node[below,font=\small] at (6.30,0.00) {$e_{4}$};
\node[below,font=\small] at (8.10,0.00) {$e_{5}$};
%top edge labels (default is empty)
\node[above,font=\small] at (0.90,0.00) {};
\node[above,font=\small] at (2.70,0.00) {};
\node[above,font=\small] at (4.50,0.00) {};
\node[above,font=\small] at (6.30,0.00) {};
\node[above,font=\small] at (8.10,0.00) {};
%bottom vertex labels (default is v0,...,v{k+1})
\node[below=2pt,font=\small] at (0.00,0.00) {$v_{0}$};
\node[below=2pt,font=\small] at (1.80,0.00) {$v_{1}$};
\node[below=2pt,font=\small] at (3.60,0.00) {$v_{2}$};
\node[below=2pt,font=\small] at (5.40,0.00) {$v_{3}$};
\node[below=2pt,font=\small] at (7.20,0.00) {$v_{4}$};
\node[below=2pt,font=\small] at (9.00,0.00) {$v_{5}$};
%top vertex labels (default is empty)
\node[above=2pt,font=\small] at (0.00,0.00) {};
\node[above=2pt,font=\small] at (1.80,0.00) {};
\node[above=2pt,font=\small] at (3.60,0.00) {};
\node[above=2pt,font=\small] at (5.40,0.00) {};
\node[above=2pt,font=\small] at (7.20,0.00) {};
\node[above=2pt,font=\small] at (9.00,0.00) {};
\end{scope}
\draw[->, thick] (4.500000,-1.000000) to (4.500000,-2.000000);
%second scope
\begin{scope}[yshift=-3cm]
%endpoints
\node[fill=red,endpoint] (v0) at (0.000000,0.000000){};
\node[fill=red,endpoint] (v3) at (9.000000,0.000000){};
%internal vertices
\node[fill=red,vertex] (v1) at (3.00,0.00){};
\node[fill=red,vertex] (v2) at (6.00,0.00){};
%edges
\draw[red] (v0)--(v1);
\draw[red] (v1)--(v2);
\draw[red] (v2)--(v3);
%bottom edge labels (default is e_1,...,e_{k+1})
\node[below,font=\small] at (1.50,0.00) {$e'_{1}$};
\node[below,font=\small] at (4.50,0.00) {$e'_{2}$};
\node[below,font=\small] at (7.50,0.00) {$e'_{3}$};
%top edge labels (default is empty)
\node[above,font=\small] at (1.50,0.00) {};
\node[above,font=\small] at (4.50,0.00) {};
\node[above,font=\small] at (7.50,0.00) {};
%bottom vertex labels (default is v0,...,v{k+1})
\node[below=2pt,font=\small] at (0.00,0.00) {$v_{0}$};
\node[below=2pt,font=\small] at (3.00,0.00) {$v'_{1}$};
\node[below=2pt,font=\small] at (6.00,0.00) {$v'_{2}$};
\node[below=2pt,font=\small] at (9.00,0.00) {$v_{5}$};
%top vertex labels (default is empty)
\node[above=2pt,font=\small] at (0.00,0.00) {};
\node[above=2pt,font=\small] at (3.00,0.00) {};
\node[above=2pt,font=\small] at (6.00,0.00) {};
\node[above=2pt,font=\small] at (9.00,0.00) {};
\end{scope}
\end{tikzpicture}

\caption{Visualization of Case 3 from the proof of \Cref{lem:rrule2_cor_forward}. In the reduced instance $(G',s,x',w')$ (bottom), the solution uses the vertices $v_1',v_2'$ and the edges $e_1', e_2', e_3'$. The vertices and edges highlighted in red are touched by the considered solutions.} 
\label{fig:reduction_rule_deg2_case3}
\end{figure}

        \item[Case 4] $E(T)\cap \{e_1,\ldots, e_{k+1}\}=\{e_1,e_2,\ldots,e_i\}\cup \{e_{j+1},\ldots,e_{k+1}\} \text{ for some $i<j$}$: Let $T':=T|_{G'}+ \{e_1',e_3'\}$. See \Cref{fig:reduction_rule_deg2_case4} for illustration. We have:
        \begin{align*}
            x'(T')-w'(T')&=\left(x(T)-\sum_{\ell=1}^{i}x(v_{\ell})-\sum_{\ell=j}^kx(v_\ell)+x_1+x_2\right)-
            \\&-\left(w(T)-\sum_{\ell=1}^iw(e_\ell)-\sum_{\ell=j}^kw(e_{\ell+1})+w_1+w_3\right)=\\&=x(T)-w(T)-\prefix_i-\suffix_j+x_1+x_2-w_1-w_3=\\&=x(T)-w(T)-\prefix_i-\suffix_j+\combined >0
        \end{align*}
        Note that $x(T)-w(T)>0$ because $T$ is a solution, $x_1+x_2-w_1-w_3 = \combined$ by \Cref{claim:rr2_ineq}   and $\combined\geq \prefix_i+\suffix_j$ because $i<j$, thus $-\prefix_i-\suffix_j+\combined\geq 0$.
        
    \end{description}
%Figure_4
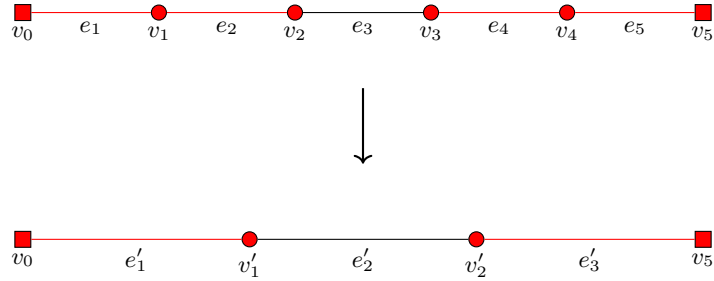
\begin{figure}[H]
    \centering
    \begin{tikzpicture}[
        vertex/.style={circle,draw,inner sep=2pt},
        endpoint/.style={rectangle,draw,inner sep=3pt}
]%first scope
\begin{scope}[yshift=0cm]
%endpoints
\node[fill=red,endpoint] (v0) at (0.000000,0.000000){};
\node[fill=red,endpoint] (v5) at (9.000000,0.000000){};
%internal vertices
\node[fill=red,vertex] (v1) at (1.80,0.00){};
\node[fill=red,vertex] (v2) at (3.60,0.00){};
\node[fill=red,vertex] (v3) at (5.40,0.00){};
\node[fill=red,vertex] (v4) at (7.20,0.00){};
%edges
\draw[red] (v0)--(v1);
\draw[red] (v1)--(v2);
\draw (v2)--(v3);
\draw[red] (v3)--(v4);
\draw[red] (v4)--(v5);
%bottom edge labels (default is e_1,...,e_{k+1})
\node[below,font=\small] at (0.90,0.00) {$e_{1}$};
\node[below,font=\small] at (2.70,0.00) {$e_{2}$};
\node[below,font=\small] at (4.50,0.00) {$e_{3}$};
\node[below,font=\small] at (6.30,0.00) {$e_{4}$};
\node[below,font=\small] at (8.10,0.00) {$e_{5}$};
%top edge labels (default is empty)
\node[above,font=\small] at (0.90,0.00) {};
\node[above,font=\small] at (2.70,0.00) {};
\node[above,font=\small] at (4.50,0.00) {};
\node[above,font=\small] at (6.30,0.00) {};
\node[above,font=\small] at (8.10,0.00) {};
%bottom vertex labels (default is v0,...,v{k+1})
\node[below=2pt,font=\small] at (0.00,0.00) {$v_{0}$};
\node[below=2pt,font=\small] at (1.80,0.00) {$v_{1}$};
\node[below=2pt,font=\small] at (3.60,0.00) {$v_{2}$};
\node[below=2pt,font=\small] at (5.40,0.00) {$v_{3}$};
\node[below=2pt,font=\small] at (7.20,0.00) {$v_{4}$};
\node[below=2pt,font=\small] at (9.00,0.00) {$v_{5}$};
%top vertex labels (default is empty)
\node[above=2pt,font=\small] at (0.00,0.00) {};
\node[above=2pt,font=\small] at (1.80,0.00) {};
\node[above=2pt,font=\small] at (3.60,0.00) {};
\node[above=2pt,font=\small] at (5.40,0.00) {};
\node[above=2pt,font=\small] at (7.20,0.00) {};
\node[above=2pt,font=\small] at (9.00,0.00) {};
\end{scope}
\draw[->, thick] (4.500000,-1.000000) to (4.500000,-2.000000);
%second scope
\begin{scope}[yshift=-3cm]
%endpoints
\node[fill=red,endpoint] (v0) at (0.000000,0.000000){};
\node[fill=red,endpoint] (v3) at (9.000000,0.000000){};
%internal vertices
\node[fill=red,vertex] (v1) at (3.00,0.00){};
\node[fill=red,vertex] (v2) at (6.00,0.00){};
%edges
\draw[red] (v0)--(v1);
\draw (v1)--(v2);
\draw[red] (v2)--(v3);
%bottom edge labels (default is e_1,...,e_{k+1})
\node[below,font=\small] at (1.50,0.00) {$e'_{1}$};
\node[below,font=\small] at (4.50,0.00) {$e'_{2}$};
\node[below,font=\small] at (7.50,0.00) {$e'_{3}$};
%top edge labels (default is empty)
\node[above,font=\small] at (1.50,0.00) {};
\node[above,font=\small] at (4.50,0.00) {};
\node[above,font=\small] at (7.50,0.00) {};
%bottom vertex labels (default is v0,...,v{k+1})
\node[below=2pt,font=\small] at (0.00,0.00) {$v_{0}$};
\node[below=2pt,font=\small] at (3.00,0.00) {$v'_{1}$};
\node[below=2pt,font=\small] at (6.00,0.00) {$v'_{2}$};
\node[below=2pt,font=\small] at (9.00,0.00) {$v_{5}$};
%top vertex labels (default is empty)
\node[above=2pt,font=\small] at (0.00,0.00) {};
\node[above=2pt,font=\small] at (3.00,0.00) {};
\node[above=2pt,font=\small] at (6.00,0.00) {};
\node[above=2pt,font=\small] at (9.00,0.00) {};
\end{scope}
\end{tikzpicture}

\caption{Visualization of Case 4 from the proof of \Cref{lem:rrule2_cor_forward} for $i=2,j=3$. In the reduced instance $(G',s,x',w')$ (bottom), the solution uses the vertices $v_1',v_2'$ and the edges $e_1', e_3'$. The vertices and edges highlighted in red are touched by the considered solutions.}
\label{fig:reduction_rule_deg2_case4}
\end{figure}
This finishes the proof of \Cref{lem:rrule2_cor_forward}.
\end{proof} %lemma =>

\begin{lemma}\label{lem:rrule2_cor_backward}
    Let $(G,s,x,w)$ be an instance of \MSTGcore and let $(G',s,x',w')$ result from $(G,s,x,w)$ by application of \Cref{rrule:deg2}. If $(G',s,x',w')$ is a yes-instance, then $(G,s,x,w)$ is a yes-instance.
\end{lemma}

\begin{proof}
        Let $T'$ be a solution for $(G',s,x',w')$. Similarly as in the proof of \Cref{lem:rrule2_cor_forward}, if $T'$ avoids the newly created vertices $v_1'$ and $v_2'$, then clearly $T:=T'$ is a solution for $(G,s,x,w)$. Thus, let us assume that $V(T')\cap \{v_1',v_2'\}\neq \emptyset$ and moreover, assume that $|V(T')|$ is maximum.

        \begin{claim}\label{claim:both_v1_v2_inT}
            $\{v_1',v_2'\}\subseteq V(T')$
        \end{claim}
        \begin{proof} 
         For the sake of contradiction and without loss of generality, suppose that $V(T')\cap \{v_1',v_2'\}=\{v_1'\}$ (the other case is symmetric). We claim that in this case $T'':=T'+e_2'$ is also a solution to $(G',s,x',w')$, thus contradicting the maximality of $T'$. Observe that
        \[
        x'(T'')-w'(T'')=x'(S')+x_2-w'(T')-w_2>0
        \]
        because $x'(S')-w'(T')>0$ and $x_2-w_2\geq 0$ by \Cref{claim:rr2_ineq}. 
        \end{proof}
        
        Thus, we may assume that $V(T')\cap \{v_1',v_2'\}=\{v_1',v_2'\}$. We split the analysis based on the intersection $E(T')\cap \{e_1',e_2',e_3'\}$.

        \begin{claim}\label{claim:atleast_two_edges_inT}
            $|E(T')\cap \{e_1',e_2',e_3'\}|\geq 2$
        \end{claim}
        \begin{proof}
            We know by \Cref{claim:both_v1_v2_inT} that both $v_1'$ and $v_2'$ are necessarily in $T'$. Let $E'=E(T')\cap \{e_1',e_2',e_3'\}$. If there are no edges in $E'$ it would imply that both $v_1'$ and $v_2'$ are isolated. If $|E'|=1$, then in the case $E'=\{e_2'\}$, $\{v_1',v_2'\}$ form a single connected component, thus $T'$ is disconnected. Finally, if $E'=\{e_1'\}$ or $E'=\{e_3'\}$, then either $v_2'$ or $v_1'$ are isolated in $T'$, thus $T'$ is disconnected. Note that since $T'$ is a tree, it cannot be disconnected.
        \end{proof}

        By \Cref{claim:atleast_two_edges_inT} we can assume $|E(T')\cap \{e_1',e_2',e_3'\}|\geq 2$. There are thus $4$ cases to consider. Before we delve into them, let us introduce some notation. Let $i_\prefix\in[k]$ be an arbitrary index such that $\prefix_{i_\prefix}=\prefix$, similarly, let $i_\suffix\in[k]$ be an arbitrary index such that $\suffix_{i_{\suffix}}=\suffix$ and $(i_\combined,j_\combined)\in \{(i,j)\mid i<j,i,j\in [k]\}$ such that $\prefix_{i_\combined}+\suffix_{j_{\combined}}=\combined$.
            \begin{description}
                \item[Case 1] $E(T')\cap \{e_1',e_2',e_3'\}=\{e_1',e_2'\}$: Let $T=T'|_{G}+ \{e_1,e_2,\ldots,e_{i_\prefix}\}$. See \Cref{fig:label_backwardirr_case1} for illustration. We have
                \begin{align*}
                x(T)-w(T)&=\left(x(T')-x_1-x_2+\sum_{\ell=1}^{i_\prefix}x(v_\ell)\right)-\left(w(T')-w_1-w_2+\sum_{\ell=1}^{i_\prefix}w(e_\ell)\right)=\\&=x(T')-w(T')>0    
                \end{align*}
                because $-x_1-x_2+w_1+w_2=-\prefix$ by \Cref{claim:rr2_ineq} and $\prefix=\prefix_{i_\prefix}$. 
%Figure_case_1_claim_
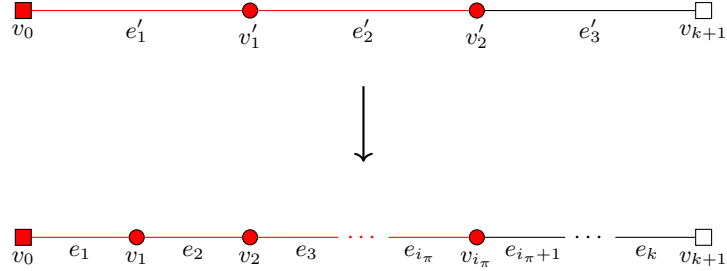
\begin{figure}[H]
    \centering
    \begin{tikzpicture}[
	vertex/.style={circle,draw,inner sep=2pt},
	endpoint/.style={rectangle,draw,inner sep=3pt}
]%first scope
\begin{scope}[yshift=0cm]
%endpoints
\node[fill=red,endpoint] (v0) at (0.000000,0.000000){};
\node[endpoint] (v3) at (9.000000,0.000000){};
%internal vertices
\node[fill=red,vertex] (v1) at (3.00,0.00){};
\node[fill=red,vertex] (v2) at (6.00,0.00){};
%edges
\draw[red] (v0)--(v1);
\draw[red] (v1)--(v2);
\draw (v2)--(v3);
%bottom edge labels (default is e_1,...,e_{k+1})
\node[below,font=\small] at (1.50,0.00) {$e'_{1}$};
\node[below,font=\small] at (4.50,0.00) {$e'_{2}$};
\node[below,font=\small] at (7.50,0.00) {$e'_{3}$};
%top edge labels (default is empty)
\node[above,font=\small] at (1.50,0.00) {};
\node[above,font=\small] at (4.50,0.00) {};
\node[above,font=\small] at (7.50,0.00) {};
%bottom vertex labels (default is v0,...,v{k+1})
\node[below=2pt,font=\small] at (0.00,0.00) {$v_{0}$};
\node[below=2pt,font=\small] at (3.00,0.00) {$v'_{1}$};
\node[below=2pt,font=\small] at (6.00,0.00) {$v'_{2}$};
\node[below=2pt,font=\small] at (9.00,0.00) {$v_{k+1}$};
%top vertex labels (default is empty)
\node[above=2pt,font=\small] at (0.00,0.00) {};
\node[above=2pt,font=\small] at (3.00,0.00) {};
\node[above=2pt,font=\small] at (6.00,0.00) {};
\node[above=2pt,font=\small] at (9.00,0.00) {};
\end{scope}
\draw[->, thick] (4.500000,-1.000000) to (4.500000,-2.000000);
%second scope
\begin{scope}[yshift=-3cm]
%endpoints
\node[fill=red,endpoint] (v0) at (0.000000,0.000000){};
\node[endpoint] (v6) at (9.000000,0.000000){};
%internal vertices
\node[fill=red,vertex] (v1) at (1.50,0.00){};
\node[fill=red,vertex] (v2) at (3.00,0.00){};
\node[] (v3) at (4.50,0.00){\textcolor{red}{$\cdots$}};
\node[fill=red,vertex] (v4) at (6.00,0.00){};
\node[] (v5) at (7.50,0.00){$\cdots$};
%edges
\draw[red] (v0)--(v1);
\draw[red] (v1)--(v2);
\draw[red] (v2)--(v3);
\draw[red] (v3)--(v4);
\draw (v4)--(v5);
\draw (v5)--(v6);
%bottom edge labels (default is e_1,...,e_{k+1})
\node[below,font=\small] at (0.75,0.00) {$e_{1}$};
\node[below,font=\small] at (2.25,0.00) {$e_{2}$};
\node[below,font=\small] at (3.75,0.00) {$e_{3}$};
\node[below,font=\small] at (5.25,0.00) {$e_{{i_{\prefix}}}$};
\node[below,font=\small] at (6.75,0.00) {$e_{i_{\prefix}+1}$};
\node[below,font=\small] at (8.25,0.00) {$e_{k}$};
%top edge labels (default is empty)
\node[above,font=\small] at (0.75,0.00) {};
\node[above,font=\small] at (2.25,0.00) {};
\node[above,font=\small] at (3.75,0.00) {};
\node[above,font=\small] at (5.25,0.00) {};
\node[above,font=\small] at (6.75,0.00) {};
\node[above,font=\small] at (8.25,0.00) {};
%bottom vertex labels (default is v0,...,v{k+1})
\node[below=2pt,font=\small] at (0.00,0.00) {$v_{0}$};
\node[below=2pt,font=\small] at (1.50,0.00) {$v_{1}$};
\node[below=2pt,font=\small] at (3.00,0.00) {$v_{2}$};
\node[below=2pt,font=\small] at (4.50,0.00) {};
\node[below=2pt,font=\small] at (6.00,0.00) {$v_{i_{\prefix}}$};
\node[below=2pt,font=\small] at (7.50,0.00) {};
\node[below=2pt,font=\small] at (9.00,0.00) {$v_{k+1}$};
%top vertex labels (default is empty)
\node[above=2pt,font=\small] at (0.00,0.00) {};
\node[above=2pt,font=\small] at (1.50,0.00) {};
\node[above=2pt,font=\small] at (3.00,0.00) {};
\node[above=2pt,font=\small] at (4.50,0.00) {};
\node[above=2pt,font=\small] at (6.00,0.00) {};
\node[above=2pt,font=\small] at (7.50,0.00) {};
\node[above=2pt,font=\small] at (9.00,0.00) {};
\end{scope}
\end{tikzpicture}
\caption{Visualization of Case 1 from the proof of \Cref{lem:rrule2_cor_backward}. In the original instance $(G,s,x,w)$ (bottom), the solution uses the best prefix of the path. The vertices and edges highlighted in red are touched by the solutions. Note that the vertex $v_{k+1}$ may or may not be in the solutions.}
\label{fig:label_backwardirr_case1}
\end{figure}
                \item[Case 2] $E(T')\cap \{e_1',e_2',e_3'\}=\{e_2',e_3'\}$:
                Let $T=T'|_{G}+ \{e_{i_\suffix+1},e_{i_\suffix+2},\ldots,e_{k+1}\}$. See \Cref{fig:label_backwardirr_case2} for illustration. We have
                \begin{align*}
                        x(T)-w(T)&=\left(x(T')-x_1-x_2+\sum_{\ell=i_\suffix}^kx(v_\ell)\right)-\left(w(T')-w_2-w_3+\sum_{\ell = i_\suffix}^kw(e_{\ell+1})\right)=\\&=x(T')-w(T')>0
                \end{align*}
                because $-x_1-x_2+w_2+w_3=-\suffix$ by \Cref{claim:rr2_ineq} and $\suffix=\suffix_{i_\suffix}$. 
%Figure_case_2_claim_12
\begin{figure}[H]
    \centering
    \begin{tikzpicture}[
	vertex/.style={circle,draw,inner sep=2pt},
	endpoint/.style={rectangle,draw,inner sep=3pt}
]%first scope
\begin{scope}[yshift=0cm]
%endpoints
\node[endpoint] (v0) at (0.000000,0.000000){};
\node[endpoint,fill=red] (v3) at (9.000000,0.000000){};
%internal vertices
\node[vertex,fill=red] (v1) at (3.00,0.00){};
\node[vertex,fill=red] (v2) at (6.00,0.00){};
%edges
\draw (v0)--(v1);
\draw[red] (v1)--(v2);
\draw[red] (v2)--(v3);
%bottom edge labels (default is e_1,...,e_{k+1})
\node[below,font=\small] at (1.50,0.00) {$e'_{1}$};
\node[below,font=\small] at (4.50,0.00) {$e'_{2}$};
\node[below,font=\small] at (7.50,0.00) {$e'_{3}$};
%top edge labels (default is empty)
\node[above,font=\small] at (1.50,0.00) {};
\node[above,font=\small] at (4.50,0.00) {};
\node[above,font=\small] at (7.50,0.00) {};
%bottom vertex labels (default is v0,...,v{k+1})
\node[below=2pt,font=\small] at (0.00,0.00) {$v_{0}$};
\node[below=2pt,font=\small] at (3.00,0.00) {$v'_{1}$};
\node[below=2pt,font=\small] at (6.00,0.00) {$v'_{2}$};
\node[below=2pt,font=\small] at (9.00,0.00) {$v_{k+1}$};
%top vertex labels (default is empty)
\node[above=2pt,font=\small] at (0.00,0.00) {};
\node[above=2pt,font=\small] at (3.00,0.00) {};
\node[above=2pt,font=\small] at (6.00,0.00) {};
\node[above=2pt,font=\small] at (9.00,0.00) {};
\end{scope}
\draw[->, thick] (4.500000,-1.000000) to (4.500000,-2.000000);
%second scope
\begin{scope}[yshift=-3cm]
%endpoints
\node[endpoint] (v0) at (0.000000,0.000000){};
\node[endpoint,fill=red] (v6) at (9.000000,0.000000){};
%internal vertices
\node[] (v1) at (1.50,0.00){$\cdots$};
\node[vertex,fill=red] (v2) at (3.00,0.00){};
\node[] (v3) at (4.50,0.00){\textcolor{red}{$\cdots$}};
\node[vertex,fill=red] (v4) at (6.00,0.00){};
\node[vertex,fill=red] (v5) at (7.50,0.00){};
%edges
\draw (v0)--(v1);
\draw (v1)--(v2);
\draw[red] (v2)--(v3);
\draw[red] (v3)--(v4);
\draw[red] (v4)--(v5);
\draw[red] (v5)--(v6);
%bottom edge labels (default is e_1,...,e_{k+1})
\node[below,font=\small] at (0.75,0.00) {$e_{1}$};
\node[below,font=\small] at (2.25,0.00) {$e_{i_{\suffix}}$};
\node[below,font=\small] at (3.75,0.00) {$e_{i_{\suffix}+1}$};
\node[below,font=\small] at (5.25,0.00) {$e_{k-1}$};
\node[below,font=\small] at (6.75,0.00) {$e_{k}$};
\node[below,font=\small] at (8.25,0.00) {$e_{k+1}$};
%top edge labels (default is empty)
\node[above,font=\small] at (0.75,0.00) {};
\node[above,font=\small] at (2.25,0.00) {};
\node[above,font=\small] at (3.75,0.00) {};
\node[above,font=\small] at (5.25,0.00) {};
\node[above,font=\small] at (6.75,0.00) {};
\node[above,font=\small] at (8.25,0.00) {};
%bottom vertex labels (default is v0,...,v{k+1})
\node[below=2pt,font=\small] at (0.00,0.00) {$v_{0}$};
\node[below=2pt,font=\small] at (1.50,0.00) {};
\node[below=2pt,font=\small] at (3.00,0.00) {$v_{i_{\suffix}}$};
\node[below=2pt,font=\small] at (4.50,0.00) {};
\node[below=2pt,font=\small] at (6.00,0.00) {$v_{k-1}$};
\node[below=2pt,font=\small] at (7.50,0.00) {$v_{k}$};
\node[below=2pt,font=\small] at (9.00,0.00) {$v_{k+1}$};
%top vertex labels (default is empty)
\node[above=2pt,font=\small] at (0.00,0.00) {};
\node[above=2pt,font=\small] at (1.50,0.00) {};
\node[above=2pt,font=\small] at (3.00,0.00) {};
\node[above=2pt,font=\small] at (4.50,0.00) {};
\node[above=2pt,font=\small] at (6.00,0.00) {};
\node[above=2pt,font=\small] at (7.50,0.00) {};
\node[above=2pt,font=\small] at (9.00,0.00) {};
\end{scope}
\end{tikzpicture}

\caption{Visualization of Case 2 from the proof of \Cref{lem:rrule2_cor_backward}. In the original instance $(G,s,x,w)$ (bottom), the solution uses the best suffix of the path. The vertices and edges highlighted in red are touched by the solutions. The vertex $v_0$ may or may not be touched by the solutions.}
\label{fig:label_backwardirr_case2}
\end{figure}
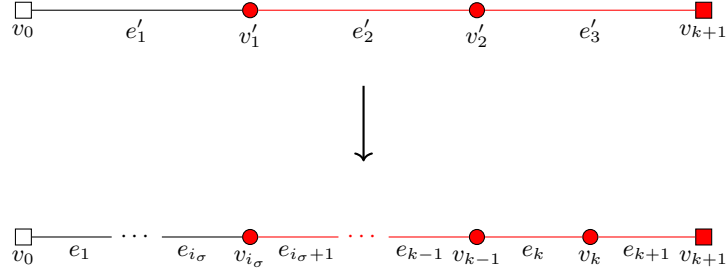

                \item[Case 3] $E(T')\cap \{e_1',e_2',e_3'\}=\{e_1',e_3'\}$: Let $T=T'|_G+ \{e_1,e_2,\ldots,e_{i_\combined}\}+\{e_{j_\combined+1},e_{j_\combined+2},\ldots,e_{k+1}\}$.
                Recall that $i_\combined < j_\combined$. See \Cref{fig:label_backwardirr_case3} for illustration.
                We have
                \begin{align*}
                    x(T)-w(T)&=\left(x(T')-x_1-x_2+\sum_{\ell=1}^{i_\combined}x(v_{\ell})+\sum_{\ell=j_\combined}^kx(v_\ell)\right)-\\&-\left(w(T')-w_1-w_3+\sum_{\ell=1}^{i_\combined}w(e_{\ell})+\sum_{\ell=j_\combined}^kw(e_{\ell+1})\right)=\\&=x(T')-w(T')>0
                \end{align*}
                because $-x_1-x_2+w_1+w_3=-\combined$ by \Cref{claim:rr2_ineq} and $\combined=\prefix_{i_\combined}+\suffix_{j_\combined}$.
%Figure_case_3_claim12
\begin{figure}[H]
    \centering
     \begin{tikzpicture}[
	vertex/.style={circle,draw,inner sep=2pt},
	endpoint/.style={rectangle,draw,inner sep=3pt}
]%first scope
\begin{scope}[yshift=0cm]
%endpoints
\node[endpoint,fill=red] (v0) at (0.000000,0.000000){};
\node[endpoint,fill=red] (v3) at (9.000000,0.000000){};
%internal vertices
\node[vertex,fill=red] (v1) at (3.00,0.00){};
\node[vertex,fill=red] (v2) at (6.00,0.00){};
%edges
\draw[red] (v0)--(v1);
\draw (v1)--(v2);
\draw[red] (v2)--(v3);
%bottom edge labels (default is e_1,...,e_{k+1})
\node[below,font=\small] at (1.50,0.00) {$e'_{1}$};
\node[below,font=\small] at (4.50,0.00) {$e'_{2}$};
\node[below,font=\small] at (7.50,0.00) {$e'_{3}$};
%top edge labels (default is empty)
\node[above,font=\small] at (1.50,0.00) {};
\node[above,font=\small] at (4.50,0.00) {};
\node[above,font=\small] at (7.50,0.00) {};
%bottom vertex labels (default is v0,...,v{k+1})
\node[below=2pt,font=\small] at (0.00,0.00) {$v_{0}$};
\node[below=2pt,font=\small] at (3.00,0.00) {$v'_{1}$};
\node[below=2pt,font=\small] at (6.00,0.00) {$v'_{2}$};
\node[below=2pt,font=\small] at (9.00,0.00) {$v_{k+1}$};
%top vertex labels (default is empty)
\node[above=2pt,font=\small] at (0.00,0.00) {};
\node[above=2pt,font=\small] at (3.00,0.00) {};
\node[above=2pt,font=\small] at (6.00,0.00) {};
\node[above=2pt,font=\small] at (9.00,0.00) {};
\end{scope}
\draw[->, thick] (4.500000,-1.000000) to (4.500000,-2.000000);
%second scope
\begin{scope}[yshift=-3cm]
%endpoints
\node[endpoint,fill=red] (v0) at (0.000000,0.000000){};
\node[endpoint,fill=red] (v6) at (9.000000,0.000000){};
%internal vertices
\node[] (v1) at (1.50,0.00){\textcolor{red}{$\cdots$}};
\node[vertex,fill=red] (v2) at (3.00,0.00){};
\node[] (v3) at (4.50,0.00){$\cdots$};
\node[vertex,fill=red] (v4) at (6.00,0.00){};
\node[] (v5) at (7.50,0.00){\textcolor{red}{$\cdots$}};
%edges
\draw[red] (v0)--(v1);
\draw[red] (v1)--(v2);
\draw (v2)--(v3);
\draw (v3)--(v4);
\draw[red] (v4)--(v5);
\draw[red] (v5)--(v6);
%bottom edge labels (default is e_1,...,e_{k+1})
\node[below,font=\small] at (0.75,0.00) {$e_{1}$};
\node[below,font=\small] at (2.25,0.00) {$e_{i_{\combined}}$};
\node[below,font=\small] at (3.75,0.00) {$e_{i_{\combined}+1}$};
\node[below,font=\small] at (5.25,0.00) {$e_{j_{\combined}}$};
\node[below,font=\small] at (6.75,0.00) {$e_{j_{\combined}+1}$};
\node[below,font=\small] at (8.25,0.00) {$e_{k+1}$};
%top edge labels (default is empty)
\node[above,font=\small] at (0.75,0.00) {};
\node[above,font=\small] at (2.25,0.00) {};
\node[above,font=\small] at (3.75,0.00) {};
\node[above,font=\small] at (5.25,0.00) {};
\node[above,font=\small] at (6.75,0.00) {};
\node[above,font=\small] at (8.25,0.00) {};
%bottom vertex labels (default is v0,...,v{k+1})
\node[below=2pt,font=\small] at (0.00,0.00) {$v_{0}$};
\node[below=2pt,font=\small] at (1.50,0.00) {};
\node[below=2pt,font=\small] at (3.00,0.00) {$v_{i_{\combined}}$};
\node[below=2pt,font=\small] at (4.50,0.00) {};
\node[below=2pt,font=\small] at (6.00,0.00) {$v_{j_{\combined}}$};
\node[below=2pt,font=\small] at (7.50,0.00) {};
\node[below=2pt,font=\small] at (9.00,0.00) {$v_{k+1}$};
%top vertex labels (default is empty)
\node[above=2pt,font=\small] at (0.00,0.00) {};
\node[above=2pt,font=\small] at (1.50,0.00) {};
\node[above=2pt,font=\small] at (3.00,0.00) {};
\node[above=2pt,font=\small] at (4.50,0.00) {};
\node[above=2pt,font=\small] at (6.00,0.00) {};
\node[above=2pt,font=\small] at (7.50,0.00) {};
\node[above=2pt,font=\small] at (9.00,0.00) {};
\end{scope}
\end{tikzpicture}

\caption{Visualization of Case 3 from the proof of \Cref{lem:rrule2_cor_backward}. In the original instance $(G,s,x,w)$ (bottom), the solution uses the best prefix and suffix combined according to the indices $i_{\combined}, j_{\combined}$. The vertices and edges highlighted in red are touched by the considered solutions.
}
\label{fig:label_backwardirr_case3}
\end{figure}
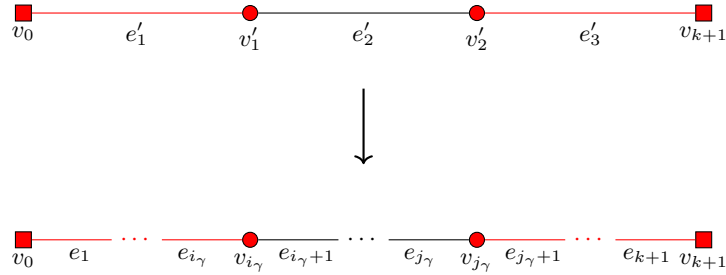

                \item[Case 4] $E(T')\cap \{e_1',e_2',e_3'\}=\{e_1',e_2',e_3'\}$: Let $T=T'|_G+ \{e_1,e_2,\ldots,e_{k+1}\}$. See \Cref{fig:label_backwardirr_case4} for illustration. We have
                \begin{align*}
                    x(T)-w(T)&=\left(x(T')-x_1-x_2+\sum_{\ell=1}^kx(v_\ell)\right)-\left(w(T')-w_1-w_2-w_3+\sum_{\ell=1}^{k+1}w(e_{\ell})\right)=\\&=x(T')-w(T')>0
                \end{align*}
                because $\tau = \sum_{\ell=1}^kx(v_\ell)-\sum_{\ell=1}^{k+1}w(e_{\ell})$ and $-x_1-x_2+w_1+w_2+w_3=-\total$ by \Cref{claim:rr2_ineq}. 
            \end{description}
%Figure_case_4_claim_12

\begin{figure}[H]
    \centering
     \begin{tikzpicture}[
	vertex/.style={circle,draw,inner sep=2pt},
	endpoint/.style={rectangle,draw,inner sep=3pt}
]%first scope
\begin{scope}[yshift=0cm]
%endpoints
\node[endpoint,fill=red] (v0) at (0.000000,0.000000){};
\node[endpoint,fill=red] (v3) at (9.000000,0.000000){};
%internal vertices
\node[vertex,fill=red] (v1) at (3.00,0.00){};
\node[vertex,fill=red] (v2) at (6.00,0.00){};
%edges
\draw[red] (v0)--(v1);
\draw[red] (v1)--(v2);
\draw[red] (v2)--(v3);
%bottom edge labels (default is e_1,...,e_{k+1})
\node[below,font=\small] at (1.50,0.00) {$e'_{1}$};
\node[below,font=\small] at (4.50,0.00) {$e'_{2}$};
\node[below,font=\small] at (7.50,0.00) {$e'_{3}$};
%top edge labels (default is empty)
\node[above,font=\small] at (1.50,0.00) {};
\node[above,font=\small] at (4.50,0.00) {};
\node[above,font=\small] at (7.50,0.00) {};
%bottom vertex labels (default is v0,...,v{k+1})
\node[below=2pt,font=\small] at (0.00,0.00) {$v_{0}$};
\node[below=2pt,font=\small] at (3.00,0.00) {$v'_{1}$};
\node[below=2pt,font=\small] at (6.00,0.00) {$v'_{2}$};
\node[below=2pt,font=\small] at (9.00,0.00) {$v_{k+1}$};
%top vertex labels (default is empty)
\node[above=2pt,font=\small] at (0.00,0.00) {};
\node[above=2pt,font=\small] at (3.00,0.00) {};
\node[above=2pt,font=\small] at (6.00,0.00) {};
\node[above=2pt,font=\small] at (9.00,0.00) {};
\end{scope}
\draw[->, thick] (4.500000,-1.000000) to (4.500000,-2.000000);
%second scope
\begin{scope}[yshift=-3cm]
%endpoints
\node[endpoint,fill=red] (v0) at (0.000000,0.000000){};
\node[endpoint,fill=red] (v6) at (9.000000,0.000000){};
%internal vertices
\node[vertex,fill=red] (v1) at (1.50,0.00){};
\node[vertex,fill=red] (v2) at (3.00,0.00){};
\node (v3) at (4.50,0.00){\textcolor{red}{$\cdots$}};
\node[vertex,fill=red] (v4) at (6.00,0.00){};
\node[vertex,fill=red] (v5) at (7.50,0.00){};
%edges
\draw[red] (v0)--(v1);
\draw[red] (v1)--(v2);
\draw[red] (v2)--(v3);
\draw[red] (v3)--(v4);
\draw[red] (v4)--(v5);
\draw[red] (v5)--(v6);
%bottom edge labels (default is e_1,...,e_{k+1})
\node[below,font=\small] at (0.75,0.00) {$e_{1}$};
\node[below,font=\small] at (2.25,0.00) {$e_{2}$};
\node[below,font=\small] at (3.75,0.00) {$e_{3}$};
\node[below,font=\small] at (5.25,0.00) {$e_{k-1}$};
\node[below,font=\small] at (6.75,0.00) {$e_{k}$};
\node[below,font=\small] at (8.25,0.00) {$e_{k+1}$};
%top edge labels (default is empty)
\node[above,font=\small] at (0.75,0.00) {};
\node[above,font=\small] at (2.25,0.00) {};
\node[above,font=\small] at (3.75,0.00) {};
\node[above,font=\small] at (5.25,0.00) {};
\node[above,font=\small] at (6.75,0.00) {};
\node[above,font=\small] at (8.25,0.00) {};
%bottom vertex labels (default is v0,...,v{k+1})
\node[below=2pt,font=\small] at (0.00,0.00) {$v_{0}$};
\node[below=2pt,font=\small] at (1.50,0.00) {$v_{1}$};
\node[below=2pt,font=\small] at (3.00,0.00) {$v_{2}$};
\node[below=2pt,font=\small] at (4.50,0.00) {};
\node[below=2pt,font=\small] at (6.00,0.00) {$v_{k-1}$};
\node[below=2pt,font=\small] at (7.50,0.00) {$v_{k}$};
\node[below=2pt,font=\small] at (9.00,0.00) {$v_{k+1}$};
%top vertex labels (default is empty)
\node[above=2pt,font=\small] at (0.00,0.00) {};
\node[above=2pt,font=\small] at (1.50,0.00) {};
\node[above=2pt,font=\small] at (3.00,0.00) {};
\node[above=2pt,font=\small] at (4.50,0.00) {};
\node[above=2pt,font=\small] at (6.00,0.00) {};
\node[above=2pt,font=\small] at (7.50,0.00) {};
\node[above=2pt,font=\small] at (9.00,0.00) {};
\end{scope}
\end{tikzpicture}

\caption{Visualization of Case 4 from the proof of \Cref{lem:rrule2_cor_backward}. In the original instance $(G,s,x,w)$ (bottom), the solution uses the whole path. The vertices and edges highlighted in red are touched by the considered solutions.}
\label{fig:label_backwardirr_case4}
\end{figure}
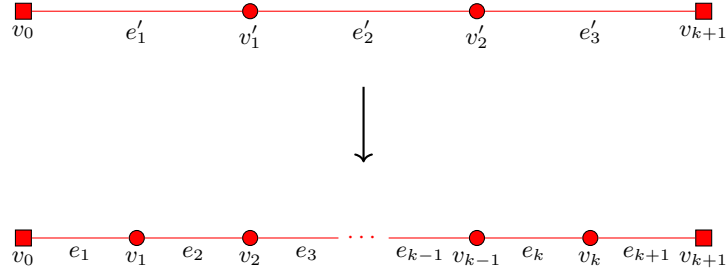
        This finishes the proof of \Cref{lem:rrule2_cor_backward}.
    \end{proof}%lemma <=

}%toappendix, lemma correctness rrule 2 (contains sub-claim and two sub-lemmas)

\begin{apprestatable}{lemma}{lemboundfen}\label{lem:bound_fen}
    Let $(G,s,x,w)$ be an instance of \MSTGcore reduced with respect to \Cref{rrule:leaves,rrule:deg2}. Then $G$ has at most $12\fen(G)$ vertices and at most $13\fen(G)$ edges.
\end{apprestatable}

\toappendix{
In the proof of \Cref{lem:bound_fen} we utilize a folklore bound on the number of vertices of degree at least $3$ in terms of the number of leaves of a tree.

\begin{observation}[\cite{DvorakKOPSS25_ARXIV}]\label{obs:bound_on_vertices_of_atleast3}
    Let $H$ be a tree and $\ell$ the number of leaves of $H$. Then the number of vertices of degree at least $3$ in $H$ is at most $\ell-2$.
\end{observation}

\sv{\lemboundfen*}
\begin{proof}[Proof of \Cref{lem:bound_fen}]
The proof follows the ideas of Nichterlein et al.~\cite[Theorem 10]{NichterleinNUW13}.
Let $F$ be a feedback edge set for $G$ and let $f:=|F|=\fen(G)$. By definition, $H=G-F$ is a tree, since we assume $G$ to be connected. Let $H^*$ result from $H$ by bypassing all vertices having degree $2$ in both $G$ and $H$. Here, \emph{bypassing a vertex} $v$ of degree $2$ with neighbors $N_G(v)=\{u,w\}$ means deleting $v$ and adding an edge $\{u,w\}$. Notice that $\deg_H(v)=\deg_{H^*}(v)$ for all vertices $v$, i.e., all vertices have the same degree in $H$ and $H^*$. We analyze the number of vertices in $H^*$. Since we applied \Cref{rrule:leaves} exhaustively, every vertex of degree $1$ in $H^*$ is either the supply vertex $s$ or is incident to an edge of $F$. All degree $2$ vertices of $H$ were bypassed in $H^*$, thus each vertex of degree $2$ in $H^*$ is incident to an edge of $F$. In total, the number of vertices of degree at most $2$ in $H^*$ is thus $|\bigcup F \cup \{s\}| \leq |\bigcup F|+1\leq 2f+1$. In particular, $H^*$ has at most $2f+1$ leaves, thus by \Cref{obs:bound_on_vertices_of_atleast3} the number of vertices of degree at least $3$ in $H^*$ is at most $2f-1$. Hence $|V(H^*)|\leq 2f+1+2f-1=4f$. We proceed to bound the number of vertices of $H$. Since $(G,s,x,w)$ was reduced with respect to \Cref{rrule:deg2}, each edge $\{u,v\}\in E(H^*)$ resulted as a bypass of a $u$-$v$ path of length at most $3$ in $H$, thus each edge of $H^*$ corresponds to at most $2$ vertices of $H$. Since $H^*$ is a forest, we have $|E(H^*)|\leq |V(H^*)|$, thus $|V(G)|=|V(H)|\leq |V(H^*)|+2\cdot |E(H^*)|\leq 4f+2\cdot |E(H^*)| \leq 4f+2\cdot 4f= 12f$. To bound the number of edges of the graph $G$, note that $f\geq |E(G)| - |V(G)|$, thus $|E(G)|\leq f+|V(G)|\leq 13f$ and this finishes the proof of the lemma.
\end{proof}
}%toappendix

By applying the same approach as in \Cref{thm:snd_kernel_ft}, we obtain \Cref{thm:fen_kernel}.
\toappendix{
\thmfenkernel*
\begin{proof}
    Let $(G,s,x,w)$ be the input instance. We apply exhaustively \Cref{rrule:leaves,rrule:deg2} and obtain an instance $(G',s,x',w')$ with $O(\fen(G'))$ vertices and edges. This follows from~\Cref{lem:bound_fen}. Observe that $\fen(G')=\fen(G)$ because the reduction rules do not change the parameter. We then apply the algorithm from \Cref{lem:alg_reducing_weights} to obtain an equivalent instance with bit-size $O(\fen(G)^4)$.
\end{proof}
}

\appsection{Vertex Cover Number}{sec:vertex_cover} 
In this section, we study the parameter vertex cover number. Note that \MSTGcore is \FPT parameterized by the vertex cover number, because it is \FPT by treewidth~(\Cref{thm:tw_fpt}). In this section we focus on kernelization. 

From the negative side, we show that in general graphs, the problem does not admit a polynomial compression (hence a polynomial kernel) unless $\NP\subseteq \coNP/_{\poly}$~(\Cref{thm:no_vc_kernel}). This result is obtained in \Cref{subsec:kernel_lb} using a polynomial parameter transformation from the \textsc{Red-Blue Dominating Set} problem.

From the positive side, we show that there is a polynomial kernel in the class of planar graphs~(\Cref{thm:planar_vc_kernel}). To obtain the kernel we utilize \Cref{rrule:leaves} from \Cref{sec:fen}. The structural property that allows for a kernel parameterized by the vertex cover number in planar graphs is the fact that planar graphs do not contain the complete bipartite graph $K_{3,3}$ as a subgraph. This entails that the number of vertices of degree $\geq 3$ outside the vertex cover is roughly $O(\vcn^3)$. While \Cref{rrule:leaves} easily disposes of degree-$1$ vertices, vertices of degree $2$ present a greater challenge. We introduce \Cref{rrule:deg2_planar_vc}, which bounds the number of degree-$2$ vertices sharing the same pair of endpoints. Ultimately, this reduction reduces the number of remaining degree-$2$ vertices outside the vertex cover to $O(\vcn^2)$.

\appsubsection{Kernelization lower bound}{subsec:kernel_lb}
In the \textsc{Red-Blue Dominating Set}, the input is a bipartite graph $G=(R\cup B,E)$ and a positive integer~$k$. The task is to find a set $S\subseteq R,|S|\leq k$ such that $N_G(S)=B$.
%\begin{center}
%\begin{tabular}{|r|l|}
%\hline
%                    & \textsc{Red-Blue Dominating Set} \\\hline
%     \textsc{Input:} & Bipartite graph $G=(R\cup B,E)$, positive integer $k$ \\\hline
%     \textsc{Question:}& Is there $S\subseteq R$,$|S|\leq k$ such that $N_G(S)=B$? \\\hline
%\end{tabular}
%\end{center}
\begin{theorem}[\cite{DomLS09_no_kernel_rbds}]\label{thm:rbds_no_kernel}
    Unless $\NP\subseteq \coNP/_{\poly}$,
    {\sc Red-Blue Dominating Set} parameterized by $|B|$ does not admit a polynomial compression.    
\end{theorem}

\begin{apprestatable}{lemma}{lemrbdomsetmstg}{\label{lem:rbdomset_mstg}}
    There is a polynomial parameter transformation from \textsc{Red-Blue Dominating Set} parameterized by $|B|$ to \MSTGcore parameterized by the vertex cover number. The resulting instance $(G',s,x,w)$ of \MSTGcore satisfies $x(G')=w(\MST(G'))$.
\end{apprestatable}

\toappendix{
\sv{\lemrbdomsetmstg*}
\begin{proof}
\definecolor{setblue}{HTML}{BBDEFB}
\definecolor{setred}{HTML}{FFCDD2}
\definecolor{setgreen}{HTML}{C8E6C9}
\definecolor{setgray}{HTML}{E0E0E0}
\definecolor{darkgreen}{HTML}{2E7D32}
\definecolor{darkpurple}{HTML}{6A1B9A}
\begin{figure}
\centering
\begin{tikzpicture}[
    vertex/.style={circle, draw, minimum size=7mm, thick},
    font=\small,
]

    % --- Vertices ---
    \node[vertex, fill=setgreen] (s) at (0,3) {$s$};
    \node[vertex, fill=setgray] (g) [below=of s] {$g$};
    \node[vertex, fill=setgreen!40] (nu) [below=of g] {$\nu$};
    
    % Set R
    \node[vertex, fill=setred] (r2) [below=of nu] {$v_{r_{2}}$};
    \node[vertex, fill=setred] (r1) [left=of r2] {$v_{r_{1}}$};
    \node[vertex, fill=setred] (r3) [right=of r2] {$v_{r_{3}}$};
    
    % Set B
    \node[vertex, fill=setblue] (b2) [below=of r2] {$v_{b_{2}}$};
    \node[vertex, fill=setblue] (b1) [left=of b2] {$v_{b_{1}}$};
    \node[vertex, fill=setblue] (b3) [right=of b2] {$v_{b_{3}}$};

     % Subgraph G
    \path[fill=setgray!55, fill opacity=0.35, rounded corners=2mm]
        ($(r1.north west)+(-0.45,0.75)$)
        rectangle
        ($(b3.south east)+(0.45,-0.35)$);
    \draw[gray, thick, dashed, rounded corners=2mm]
        ($(r1.north west)+(-0.45,0.75)$)
        rectangle
        ($(b3.south east)+(0.45,-0.35)$);
    \node[anchor=north west, color=gray!70!black]
        at ($(r1.north west)+(-0.25,0.55)$) {$G$};

    % --- Left Column: Node Definitions (Aligned with Nodes) ---
    \node[left=6cm of s, anchor=west] {supply node: $x(s) = 0$};
    \node[left=6cm of g, anchor=west] {guard vertex: $x(g) = 0$};
    \node[left=6cm of nu, anchor=west] {steiner vertex: $x(\nu) = 0$};
    \node[left=6cm of r2, anchor=west] {set $R$: $x(v) = k$};
    \node[left=6cm of b2, anchor=west] {set $B$: $x(v) = 2k+4$};

    % --- Right Column: Edge Weights (Aligned with Edge Gaps) ---

    \node[](gnu) at($(g)!0.5!(nu)$) {};
    % Midpoint g to nu
    \node[right=2.43cm of gnu, anchor=west, color=darkpurple]
        {$g \leftrightarrow R \cup \{\nu\}$: $w = k + 1$};
        
    % Midpoint nu to R
    \node[](nuR) at ($(nu)!0.5!(r2)$){};
    \node[right=2.43cm of nuR, anchor=west, color=darkgreen]
        {$\nu \leftrightarrow R$: $w = k$};
        
    % Midpoint R to B
    \node[](RB) at ($(r2)!0.5!(b2)$){};
    \node[right=2.43cm of RB, anchor=west] 
        {$R \leftrightarrow B$: $w = k + 2$};

    %Midpoint s to g
    \node[](sg) at($(s)!0.5!(g)$){};
    \node[right=2.43cm of sg, anchor=west]{$s\leftrightarrow g$: $w = |B|(k+2) - (k+1)$};

    % --- Edges ---
    % s to g (ultra-thick black)
    \draw[line width=2.5pt] (s) -- (g);

    % g to nu (purple, weight k+1)
    \draw[thick, darkpurple] (g) -- (nu);

    % g to R (purple, weight k+1)
    \draw[thick, darkpurple] (g) to[bend right=30] (r1);
    \draw[thick, darkpurple] (g) to[bend right=30] (r2); 
    \draw[thick, darkpurple] (g) to[bend left=30] (r3);

    % nu to R (green, weight k)
    \foreach \i in {1,2,3} {
        \draw[thick, darkgreen] (nu) -- (r\i);
    }

    % B to R (black, weight k+2)
    \draw[thick] (b1) -- (r1);
    \draw[thick] (b1) -- (r2);
    
    \draw[thick] (b2) -- (r1);
    \draw[thick] (b2) -- (r3);
    
    \draw[thick] (b2) -- (r2);
    \draw[thick] (b3) -- (r3);

\end{tikzpicture}

    \caption{Construction of the \MSTGcore instance from a \textsc{Red-Blue Dominating Set} instance (\Cref{lem:rbdomset_mstg}). Vertex $x$-values are listed on the left, while edge weights $w$ are indicated by color-coded labels on the right. The black edges between sets $R$ and $B$ correspond to the edges of the graph $G$ from the instance of \textsc{Red-Blue Dominating Set}.}
    \label{fig:rb_ds_reduction}
\end{figure}
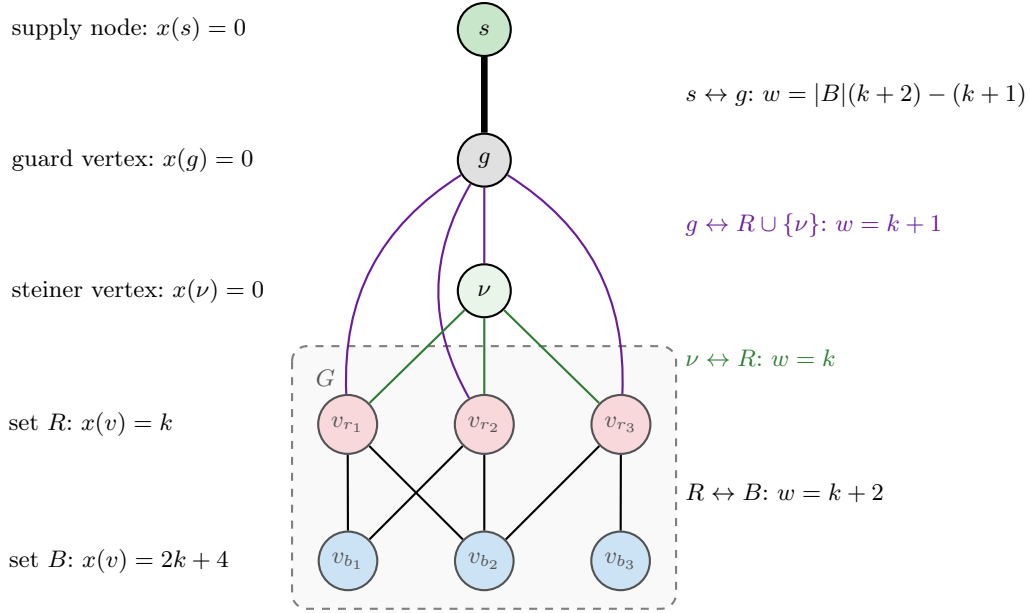

    Let $(G=(R\cup B,E),k)$ be an instance of \textsc{Red-Blue Dominating Set}. If $B$ contains an isolated vertex, we output a trivial no-instance as such instance of \textsc{Red-Blue Dominating Set} has no solution. There is no vertex in $R$ to cover the isolated vertex in $B$. Assume that $B$ does not contain an isolated vertex. We construct an instance $(G',s,x,w)$ of \MSTGcore as follows (see \Cref{fig:rb_ds_reduction}). The reduction here is the same as in \Cref{thm:nph_twoapex}. Here, $G$ plays the role of $\widehat{G}$. For clarity and completeness, we repeat the description and all the claims. Start with the bipartite graph $G$. Add a vertex $g$ (the \emph{guard}) connected to all of $R$, then, add the vertex $\nu$ (the \emph{steiner vertex}) and connect it to $R\cup \{g\}$. Finally, add the supply vertex $s$ and connect it to $g$. We now specify the weights $w$ and the allocation $x$. All vertices $v$ outside $R\cup B$ satisfy $x(v)=0$. The vertices $v_r\in R$ satisfy $x(v_r)=k$, and $x(v_b)=2k+4$ for $v_b\in B$. The weights of the edges are as follows:
    \begin{align*}
        &w(\{s,g\})=|B|\cdot (k+2)-(k+1),\\
        \forall v\in R\cup \{\nu\}:&w(\{g,v\})=k+1,\\
        \forall v_r\in R:&w(\{\nu,v_r\})=k,\\
        \forall \{v_r,v_b\}\in E(G):&w(\{v_r,v_b\})=k+2.
    \end{align*}
    This finishes the description of the instance $(G',s,x,w)$.

    \begin{claim}\label{claim:rbds_spanning_tree}
       $x(G')=w(\MST(G'))$ 
    \end{claim}
    \begin{proof}
    Observe that $x(G')=|B|\cdot (2k+4)+|R|\cdot k$. The minimum spanning tree of $G'$ looks as follows. Consider running the Kruskal's algorithm for finding the minimum spanning tree on $G'$. First, all edges from $\nu$ to $R$ will be added. Then one of the edges with weight $k+1$ connecting the created component to the guard $g$. We can assume that the edge is $\{\nu,g\}$. Then, $|B|$ edges from $R$ to $B$ of weight $k+2$ are added and finally the edge $\{s,g\}$ is added. The total weight of the minimum spanning tree of $G'$ is thus:
    \[
    w(\MST(G'))=|R|\cdot k+(k+1)+|B|\cdot (k+2)+|B|\cdot(k+2)-(k+1)
    \]
    which is the same as $x(G')$.
    \end{proof}

    \begin{claim}\label{claim:rbds_first_dir}
        If $(G,k)$ is a yes-instance of \textsc{Red-Blue Dominating Set}, then $(G',s,x,w)$ is a yes-instance of \MSTGcore.
    \end{claim}
    \begin{proof}
    Suppose $(G,k)$ is yes-instance of \textsc{Red-Blue dominating Set} and let $S\subseteq R$ be a solution of $(G,k)$. We construct a solution $T$ in $(G',s,x,w)$ as follows. Include the vertices $s$ and $g$ and connect $g$ to $R$ using the edges $\{g,v_r\},v_r\in S$. Finally, add all vertices of $B$ as leaves. Since $N_G(S)=B$, there is an available neighbor in $S$ for $v_b\in B$. We have
    \begin{align*}
        x(T)&=|B|\cdot (2k+4)+|S|\cdot k,\\
        w(T)&=|B|\cdot (k+2) + |S|\cdot (k+1)+ |B|\cdot (k+2)-(k+1) .
    \end{align*}
    It follows that $x(T)-w(T)=-|S|+(k+1)\geq 1>0$, because $|S|\leq k$. Thus, $(G',s,x,w)$ is a yes-instance.
    \end{proof}

    \begin{claim}\label{claim:rbds_second_dir}
        If $(G',s,x,w)$ is a yes-instance of \MSTGcore, then $(G,k)$ is a yes-instance of \textsc{Red-Blue Dominating Set}.
    \end{claim}
    \begin{proof}
        Suppose that there is a solution $T$ for $(G',s,x,w)$, i.e., $x(T)-w(T)>0$ and $s\in V(T)$. Suppose moreover that $x(T)-w(T)$ is maximum possible. Consider the allocation $y\colon V(G')\to \mathbb{Z}$ given by the greedy algorithm computing the core allocation from the minimum spanning tree from the proof of \Cref{claim:rbds_spanning_tree}: $y(v)$ is given by the weight of the first edge on the path from $v$ to $s$ in a minimum spanning tree, and $y(s)=0$ (see~\cite{bird1976cost}). The allocation $y$ is given by:
        \begin{align*}
            \forall v_b\in B:&y(v_b)=k+2, \\
            \forall v_r\in R:&y(v_r)=k, \\
            &y(\nu)=w(\{g,\nu\}) = k+1, \\
            &y(g)=w(\{s,g\})=|B|\cdot (k+2)-(k+1).
        \end{align*}

        Since $y$ is in the core, we have $y(T)\leq w(T)$. In particular, $x(T)>y(T)$, thus $x(T)-y(T)>0$. Let $V_B=V(T)\cap B$ and $V_R=V(T)\cap R$.

        \begin{claim}\label{claim:rbds_guard_steiner}
            $g\in V(T)$ and $\nu \notin V(T)$.
        \end{claim}
        \begin{proof}
            Notice that $g$ is the only vertex connected to $s$, so $g\in V(T)$. For the sake of contradiction assume that $\nu \in V(T)$. Then $x(T)= |V_B|\cdot (2k+4)+|V_R|\cdot k$. On the other hand, $y(T)=|V_B|\cdot (k+2)+|V_R|\cdot k+(k+1)+|B|\cdot (k+2)-(k+1)$. Notice that $x(T)-y(T)=(|V_B|-|B|)(k+2)\leq 0$ because $|V_B|\leq |B|$. A contradiction, because $x(T)-y(T)>0$.
        \end{proof}

        \begin{claim}\label{claim:rbds_leaves}
            All vertices in $V_B$ are leaves in $T$.
        \end{claim}
        \begin{proof}
            For the sake of contradiction, suppose that there is a vertex $v_b\in V_B$ with $\deg_T(v_B)\geq 2$. Let $v_r\in V_R$ be a neighbor of $v_b$ in $T$ such that the tree $T\setminus e$ where $e=\{v_r,v_b\}$ contains no path from $g$ to $v_r$. Note that such $v_r$ must exist as otherwise the paths together with the edges incident to $v_b$ would create a cycle in $T$, which is impossible, because $T$ is a tree. Since $g$ and $v_r$ lie in distinct connected components of $T\setminus e$, we can connect them via the edge $\{g,v_r\}$, creating a new tree $T'$. Notice that $w(T')=w(T)-(k+2)+(k+1)=w(T)-1$, but this contradicts the maximality of $x(T)-w(T)$.
        \end{proof}

        By \Cref{claim:rbds_guard_steiner,claim:rbds_leaves}, every vertex of $V_B$ is connected via exactly one edge to a vertex of $V_R$ and every vertex of $V_R$ is connected to $g$ via an edge of weight $k+1$. It follows that
        \begin{align*}
            x(T)&= |V_B|\cdot (2k+4)+|V_R|\cdot k \\
            w(T)&=|V_B|\cdot (k+2)+|V_R|\cdot(k+1)+|B|(k+2)-(k+1)
        \end{align*}
        The assumption $x(T)>w(T)$ is equivalent to $x(T)-w(T)>0$. Substituting the expressions for $x(T)$ and $y(T)$ into this inequality yields
        \begin{align}\label{eqn:rbds}
            |V_B|\cdot (k+2)-|V_R|-|B|\cdot(k+2)+(k+1)>0 
        \end{align}

        \begin{claim}
            $|V_B|=|B|$.
        \end{claim}
        \begin{proof}
        It is sufficient to prove both inequalities $|V_B| \leq |B|$ and $|V_B| \geq |B|$. For the first inequality, we note that $|V_B| \leq |B|$ holds because $V_B \subseteq B $. Now for proving $|V_B| \geq |B|$, for the sake of contradiction we can assume that $|V_B| < |B|$. Since both numbers are integers, this is equivalent to $|V_B| \leq |B| -1$. By plugging into (\ref{eqn:rbds}) we get:
        \[
        (|B|-1)\cdot (k+2)-|V_R|-|B|\cdot(k+2)+(k+1)>0
        \]
        which entails $-1 > |V_R|$, which is impossible.
        \end{proof}

        \begin{claim}
            $|V_R|\leq k$.
        \end{claim}
        \begin{proof}
        We note that by plugging $|V_B| = |B|$ into \Cref{eqn:rbds} we obtain that $(k+1)>|V_R|$, i.e., $|V_R|\leq k$, as claimed.
        \end{proof}
    Since $T$ is connected and each $v_b\in V_B$ has a neighbor in $V_R$, it follows that $N_G(V_R)\supseteq V_B=B$, thus $V_R$ is a red-blue dominating set of size at most $k$ for $(G,k)$. This finishes the proof of \Cref{claim:rbds_second_dir}.
    \end{proof}

    The correctness of the construction follows from \Cref{claim:rbds_first_dir,claim:rbds_second_dir}. Note that the reduction can be computed in polynomial time. It remains to argue about the parameters. Note that the set $B\cup \{\nu,g\}$ is a vertex cover of $G'$, thus $\vcn(G')\leq |B|+2$. Thus, the reduction is a polynomial parameter transformation. This finishes the proof of \Cref{lem:rbdomset_mstg}.
\end{proof}
}%toappendix

As a direct corollary of \Cref{thm:rbds_no_kernel} and \Cref{lem:rbdomset_mstg} we obtain the following theorem.

\begin{theorem}\label{thm:no_vc_kernel}
    Unless $\NP\subseteq \coNP/_{\poly}$, \MSTGcore parameterized by the vertex cover number does not admit a polynomial compression.
\end{theorem}

\appsubsection{Polynomial kernel for planar graphs}{subsec:kernel_planar}
The ingredients to obtain a polynomial kernel in planar graphs are as follows. We make use of \Cref{rrule:leaves} that deals with vertices of degree $1$ and we introduce a new reduction rule that deals with many vertices of degree two with two fixed endpoints. Fix two vertices $u,v$ and let $v_1,v_2,\ldots,v_k$ be all the vertices with $N(v_i)=\{u,v\}$. Notice that at most one vertex among $v_1,\ldots,v_k$ will have degree $2$ in a possible solution. Let $e_i^v,e_i^v$ be the two edges incident to $v_i$ (with other endpoint $v$ and $u$, respectively).
Using $v_i$ as a degree-$2$ vertex has cost $c_i:=x(v_i)-w(e_i^v)-w(e_i^u)$. All the remaining vertices will be attached as leaves either to $u$ or $v$ (or none). If $v_i$ is used as a leaf, then the gain is either $x(v_i)-w(e_i^u)$ or $x(v_i)-w(e_i^v)$. If both of these quantities are negative, we may as well not use $v_i$ at all and the gain is $0$. Hence using $v_i$ as a leaf has cost $\ell_i:=\max\{x(v_i)-w(e_i^v),x(v_i)-w(e_i^u),0\}$. Note that the behavior of the vertex is determined by the relative order of the numbers $x(v_i),w(e_i^u),w(e_i^v)$. Since vertices of the same type will have uniform behavior, we can merge all of them into one super vertex $v_{\Sigma}$.

We now describe the situation more formally. The \emph{type} of $v_i$ is an ordered triple $(\varphi,\prec_1,\prec_2)$ where $\varphi$ is a bijection $\varphi\colon \{1,2,3\}\to \{v_i,e_i^v,e_i^v\}$ and $\prec_j\in \{<,=\}$ for $j\in\{1,2\}$. A vertex $v_i$ is of type $(\varphi,\prec_1,\prec_2)$ if $\xi(\varphi(1))\prec_1\xi(\varphi(2))\prec_2\xi(\varphi(3))$ where $\xi$ is either $x$ or $w$ based on whether $\varphi(\cdot)$ is a vertex or an edge. Formally $\xi\colon V\cup E\to \mathbb{Q}^+_0$, $\xi=w\cup x$.
Notice that every vertex is of exactly one type and there are at most $3!\cdot 2^2=24$ distinct types of vertices for every fixed $u,v$. We are now ready to formulate \Cref{rrule:deg2_planar_vc}.

\begin{rrule}\label{rrule:deg2_planar_vc}
    Let $(G,s,x,w)$ be an instance of \MSTGcore. Let $u,v\in V(G)$ be two vertices. Let $v_1,v_2,\ldots,v_k (k>2)$ be vertices avoiding $s$ and satisfying $N_G(v_i)=\{u,v\}$. Suppose that all of $v_i$ are of the same type. Let $i^*=\operatorname{argmax}_{i\in[k]}\{c_{i}+\sum_{j\neq i}\ell_j\}$. Merge the vertices in $\{v_1,v_2,\ldots,v_k\}\setminus \{v_i^*\}$ into a new vertex $v_{\Sigma}$, creating a new graph $G'$. All vertices and edges from the original instance have their values and weights unchanged. Let
    \begin{align*}
        x'(v_{\Sigma})=\sum_{i\neq i^*}x(v_i), \quad&& w'(\{v_{\Sigma},u\})=\sum_{i\neq i^*}w(e_i^u),\quad&& w'(\{v_{\Sigma},v\})=\sum_{i\neq i^*}w(e_i^v).
    \end{align*}
    Output the instance $(G',s,x',w')$.
\end{rrule}

\begin{apprestatable}{lemma}{lemdegtwoplanarvccorrectness}
    \Cref{rrule:deg2_planar_vc} is correct.
\end{apprestatable}
\toappendix{
\sv{\lemdegtwoplanarvccorrectness*}
\begin{proof}
    Let $(G,s,x,w)$ be an instance of \MSTGcore and let $(G',s,x',w')$ result from $(G,s,x,w)$ by application of \Cref{rrule:deg2_planar_vc}. We show that $(G,s,x,w)$ is a yes-instance if and only if $(G',s,x',w')$ is a yes-instance.

    $\Rightarrow$: Let $T$ be a solution for $(G,s,x,w)$. If $T$ avoids the vertices $v_1,\ldots,v_k$ then it is a solution for $(G',s,x',w')$. Otherwise, assume that $V(T)\cap \{v_1,v_2,\ldots,v_k\}\neq \emptyset$.
    \begin{description}
        \item[Case 1]: All vertices $v_i\in \{v_1,v_2,\ldots,v_k\}\cap V(T)$ are leaves of $T$. Since $T$ is connected, a vertex $v_i$ can only be a leaf if at least one of its neighbors ($u$ or $v$) are in $T$. We split this into three subcases.
        \begin{description}
            \item[Subcase 1.1.] $u\in V(T), v\notin V(T)$. The vertices $v_i$ can only be attached to $u$. For each $v_i$, including it as a leaf from $u$ has gain $x(v_i)-w(e_i^u)$ which has the same sign for every $i$, since all $v_i$ are of the same type. Hence we either include all of them (in case the gain is nonnegative) or none of them.
            \item[Subcase 1.2.] $v\in V(T), u\notin V(T)$. This is symmetric to subcase 1.1., the vertices can only be attached to $v$.
            \item[Subcase 1.3.] $u\in V(T), v\in V(T)$. Each $v_i$ can be attached to either $u$ or $v$. The maximum leaf gain for any $v_i$ is $\ell_i = \max\{x(v_i)-w(e_i^u), x(v_i)-w(e_i^v), 0\}$. Because all vertices share the same type, the choice that achieves this maximum is uniform across all $i \in [k]$. That is, either they should all be leaves at $u$, all be leaves at $v$, or all be excluded.
        \end{description}

        We obtain a solution $T'$ in $(G',s,x',w')$ by including the super vertex $v_{\Sigma}$ and possibly $v_{i^*}$ as leaves from the same vertices (either $u$ or $v$) as in $T$. As argued above, in all cases, either all vertices should be leaves from $u$, or leaves from $v$, or not touched at all.
        
        \item[Case 2]: At least one vertex $v_{i'}\in \{v_1,v_2,\ldots,v_k\}\cap V(T)$ has degree $2$ in $T$ for some $i'$. Note that exactly one such vertex can have degree $2$, as two or more would create a cycle in $T$. Consequently, both $u$ and $v$ must belong to $V(T)$, and any remaining vertex $v_i$ ($i\neq i'$) can only be a leaf or excluded.

        The total gain contribution of this configuration is precisely captured by the term $c_{i'}+\sum_{i\neq i'}\ell_i$. By definition, $i^*$ maximizes this exact algebraic expression. Thus, if $i'\neq i^*$, we can exchange $v_{i'}$ with $v_{i^*}$ in $T$ (making $v_{i^*}$ the degree-2 vertex and $v_{i'}$ a leaf/excluded matching the behavior of the other vertices) without decreasing the objective value $x(T) - w(T)$.

        Assuming without loss of generality that $i' = i^*$, all remaining vertices $v_i,i \neq i^*$ behave uniformly as leaves at $u$, leaves at $v$, or are excluded (per Subcase 1.3.). We form $T'$ by retaining $v_{i^*}$ as a degree-2 vertex and adding the super-vertex $v_{\Sigma}$ to $G'$ as a leaf at $u$, a leaf at $v$, or omitting it entirely, matching the optimal choice of the remaining vertices.
        
    \end{description}

    $\Leftarrow$: Let $T'$ be a solution for $(G',s,x',w')$. If $T'$ does not touch the newly created vertices, then $T:=T'$ is a solution in $(G,s,x,w)$. Otherwise, we create a solution $T$ of $(G,s,x,w)$ as follows. If $T'$ contains the vertex $v_{\Sigma}$, include all vertices $v_i$ with $i\neq i^*$ in $T$ with all its incident edges. Similarly with $v_{i^*}$ with corresponding incident edges. If $v_{\Sigma}$ has degree $2$ in $T'$, we would create a cycle, hence we delete all but one edge of the form $\{v_i,u\}$ where $i\in [k]\setminus \{i^*\}$. Note that since all the weights are nonnegative, $x(T)-w(T)\geq x'(T')-w'(T')$, hence $T$ is a solution for $(G,s,x,w)$, as we wanted to show.
    
\end{proof}
}%toappendix

\begin{apprestatable}{lemma}{lemvckernelbound}\label{lem:vc_kernel_bound}
    Let $(G,s,x,w)$ be an instance of \MSTGcore reduced with respect to \Cref{rrule:leaves,rrule:deg2_planar_vc}. If $G$ is planar, then $G$ has at most $O(\vcn(G)^3)$ vertices and edges.
\end{apprestatable}
\toappendix{
\sv{\lemvckernelbound*}
\begin{proof}    
    Let $S$ be a vertex cover of $G$ of size $\vcn(G)$. Let $S'=S\cup \{s\}$.
    We bound the number  of vertices outside of $S'$. Note that $G\setminus S'$ is, by definition, edgeless, hence $N_G(v)\subseteq S'$ for every $v\in V(G)\setminus S'$. Notice that no vertices of degree $1$ outside $S'$ can exist because we exhaustively applied \Cref{rrule:leaves}. For vertices of degree $2$, notice that for every choice of $\{u,v\}\subseteq S'$ there are at most $2$ vertices of each type (as otherwise we could apply \Cref{rrule:deg2_planar_vc}). Since there are at most $24$ distinct types, there are at most $48$ such vertices for any choice of $\{u,v\}$. Hence in total, there are at most $48\binom{|S'|}{2}$ vertices of degree $2$ outside $S'$. Finally, consider vertices of degree $\geq 3$ outside $S'$. Notice that there cannot be more than $2\binom{|S'|}{3}$ of them, because for any three vertices $\{u_1,u_2,u_3\}\subseteq S'$ there can be at most two vertices $v$ outside $S'$ with $N_G(v)\supseteq \{u_1,u_2,u_3\}$ as otherwise $G$ has $K_{3,3}$ as a subgraph, which is impossible, because $G$ is planar. Note that $|S'| \leq \vcn+1$. Hence, outside $S'$, we have at most $1+48\binom{|S'|}{2}+2\binom{|S'|}{3}\leq O(|S'|^3)\leq O(\vcn^3)$ vertices, hence the total number of vertices of $G'$ is at most $|S'|+O(\vcn^3)=O(\vcn^3)$ vertices. Since the graph is planar, it also has $O(\vcn^3)$ edges and this proves the theorem.
\end{proof}
}%toappendix

By applying the same approach as in \Cref{thm:snd_kernel_ft,thm:fen_kernel}, we obtain \Cref{thm:planar_vc_kernel}.

\begin{apprestatable}{theorem}{thmplanarvckernel}\label{thm:planar_vc_kernel}
\MSTGcore admits a polynomial kernel parameterized by the vertex cover number if the input graph is planar.
\end{apprestatable}

\toappendix{
\sv{\thmplanarvckernel*}
\begin{proof}
    Let $(G,s,x,w)$ be the input instance and suppose that $G$ is planar. We apply exhaustively \Cref{rrule:leaves,rrule:deg2_planar_vc} and obtain an instance $(G',s,x',w')$ with $O(\vcn(G')^3)$ vertices and edges. This follows from \Cref{lem:vc_kernel_bound}. Observe that the reduction rules do not increase the vertex cover, and they preserve planarity, hence $\vcn(G')\leq \vcn(G)$. By \Cref{lem:vc_kernel_bound} the reduced instance has $O(\vcn(G)^3)$ vertices. We then apply the algorithm from \Cref{lem:alg_reducing_weights} to obtain an equivalent instance with bit-size $O(\vcn(G)^{12})$.
\end{proof}
}

\section{Conclusion}\label{sec:conclusion}
We studied the computational complexity of deciding core membership for minimum-cost spanning tree games, formalized as the \MSTGcore problem. We strengthened previous \NP-hardness results by proving that the problem remains hard even when the underlying graph is close to being planar. On the algorithmic side, we presented positive results within the framework of parameterized complexity. Specifically, we showed that the problem is \FPT parameterized by the support size of the allocation, signed neighborhood diversity, and treewidth. Furthermore, we obtained polynomial kernels for the signed neighborhood diversity, the feedback edge number and, on planar graphs, for the vertex cover number. Complementing these positive results, we proved a kernelization lower bound for the vertex cover number on general graphs. We conclude by proposing three open questions for future research:
\begin{enumerate}
    \item \textbf{Does \MSTGcore remain \NP-hard on planar graphs even when the efficiency condition holds?} Because we found no direct algorithmic application of the efficiency condition, we conjecture that the problem remains hard even in this setting.
    \item \textbf{What is the parameterized complexity of \MSTGcore with respect to signed variants of modular-width and clique-width?}
    \item \textbf{Does \MSTGcore admit a polynomial kernel parameterized by the feedback vertex set number in planar graphs?}
\end{enumerate}

%\begin{credits}
%\subsubsection{\ackname} A bold run-in heading in small font size at the end of the paper is
%used for general acknowledgments, for example: This study was funded
%by X (grant number Y).

%\subsubsection{\discintname}
%It is now necessary to declare any competing interests or to specifically
%state that the authors have no competing interests. Please place the
%statement with a bold run-in heading in small font size beneath the
%(optional) acknowledgments\footnote{If EquinOCS, our proceedings submission
%system, is used, then the disclaimer can be provided directly in the system.},
%for example: The authors have no competing interests to declare that are
%relevant to the content of this article. Or: Author A has received research
%grants from Company W. Author B has received a speaker honorarium from
%Company X and owns stock in Company Y. Author C is a member of committee Z.
%\end{credits}
%
% ---- Bibliography ----
%
% BibTeX users should specify bibliography style 'splncs04'.
% References will then be sorted and formatted in the correct style.
%
\clearpage
 \bibliographystyle{splncs04}
 \bibliography{references}

\clearpage
\appendix
\appendixText

\end{document}